\documentclass[11pt]{article}
\usepackage{titling}
\usepackage[OT1]{fontenc}
\usepackage{amsmath}
\usepackage{mathrsfs}
\usepackage{amsfonts}
\usepackage{amssymb}
\usepackage{diagbox}
\usepackage{footnote}
\usepackage{graphicx}
\usepackage{float}
\usepackage{epstopdf}
\usepackage[margin=1in]{geometry}
\usepackage{setspace}
\usepackage{tikz}
\usepackage{pgfplots}
\usepackage[round]{natbib}
\usepackage{multirow}
\usepackage[toc,page]{appendix}
\usepackage{hyperref}
\usepackage{caption}
\usepackage{subcaption}
\usepackage{enumitem}
\usepackage{booktabs}
\usepackage{dsfont}
\usepackage{bbm}
\usepackage{titletoc}
\usepackage{bibunits}
\defaultbibliography{reference1}
\defaultbibliographystyle{apalike}
\pgfplotsset{compat=1.18}
\providecommand{\U}[1]{\protect\rule{.1in}{.1in}}
\providecommand{\keywords}[1]{\textsc{Keywords}: #1}

\hypersetup{
	colorlinks=true,
	linkcolor=blue,
	citecolor=blue,
	filecolor=magenta,
	urlcolor=blue,
	breaklinks=true,
}

\newtheorem{assumption}{Assumption}[section]

\newtheorem{definition}{Definition}[section]

\newtheorem{lemma}{Lemma}[section]

\newtheorem{proposition}{Proposition}[section]
\newtheorem{remark}{Remark}[section]

\newenvironment{proof}[1][Proof]{\noindent\textbf{#1.} }{\ \rule{0.5em}{0.5em}}

\numberwithin{equation}{section}
\numberwithin{table}{section}
\numberwithin{figure}{section}

\begin{document}
\begin{bibunit}
\title{Almost Dominance: Inference and Application\thanks{The authors are grateful to all the seminar and conference participants for their insightful comments.}}

\author{Xiaojun Song\thanks{Department of Business Statistics and Econometrics, Guanghua School of Management, Peking University, Email: \texttt{sxj@gsm.pku.edu.cn}. This work was supported by the National Natural Science Foundation of China [Grant Numbers 72373007 and 72333001]. The author also gratefully acknowledges the research support from the Center for Statistical Science of Peking University.}\\
	\and Zhenting Sun\thanks{Department of Economics, University of Melbourne, Email: \texttt{zhenting.sun@unimelb.edu.au}. This work was supported by the National Natural Science Foundation of China [Grant Number 72103004].}
	}

\maketitle
\begin{abstract}
This paper proposes a general framework for inference on three types of almost dominances: almost Lorenz dominance, almost inverse stochastic dominance, and almost stochastic dominance.  
We first generalize almost Lorenz dominance to almost upward and downward Lorenz dominances. We then provide a bootstrap inference procedure for the Lorenz dominance coefficients, which measure the degrees of almost Lorenz dominance. Furthermore, we propose almost upward and downward inverse stochastic dominances and provide inference on the inverse stochastic dominance coefficients. We also show that our results can easily be extended to almost stochastic dominance. Simulation studies demonstrate the finite sample properties of the proposed estimators and the bootstrap confidence intervals. This framework can be applied to economic analysis, particularly in the areas of social welfare, inequality, and decision making under uncertainty. As an empirical example, we apply the methods to the inequality growth in the United Kingdom and find evidence for almost upward inverse stochastic dominance.  
	
	
\end{abstract}

\newpage

\section{Introduction}
Suppose that there are two arbitrary cumulative distribution functions (CDFs)
$F_{1}:[0,\infty)\to[0,1]$ and $F_{2}:[0,\infty)\to[0,1]$ of two income (or wealth, etc.) distributions in two
populations. As introduced in \citet{atkinson1970measurement}, a Lorenz curve is a function that graphs the cumulative proportion of total income received by the bottom population. The distribution $F_1$ (weakly) Lorenz dominates $F_2$ if the Lorenz curve $L_1$ associated with $F_1$ is everywhere above the Lorenz curve $L_2$ associated with $F_2$. Such Lorenz dominance implies that the wealth is distributed more equally in population $F_1$ compared to $F_2$. Statistical tests of Lorenz dominance can be found in, for example,  \citet{mcfadden1989testing}, 
\citet{bishop1991international}, \citet{bishop1991lorenz}, \citet{dardanoni1999inference}, \citet{davidson2000statistical}, \citet{barrett2003consistent}, \citet{BDB14}, and \citet{Beare2017improved}.

\citet{zheng2018almost} introduces the notion of \emph{almost Lorenz dominance}: When two Lorenz curves cross, $F_1$ almost Lorenz dominates $F_2$ if $L_1$ is above $L_2$ almost everywhere. 
The degree of almost Lorenz dominance can be measured by a coefficient, \textit{Lorenz dominance coefficient} (LDC). The LDC was first proposed in \citet[Section 2.3]{zheng2018almost},\footnote{The original definition of LDC in \citet{zheng2018almost} is slightly different from the proposed one in this paper.} which is similar to the violation ratio in \citet{huang2021estimating}. \citet{zheng2018almost} shows that almost Lorenz dominance is highly related to Gini-type measures. 
Based on the seminal work of \citet{aaberge2009ranking}, we generalize almost Lorenz dominance to \textit{almost upward Lorenz dominance} and \textit{almost downward Lorenz dominance}.
We relate the generalized LDCs to a class of inequality measures.
It is straightforward to show that the smaller LDCs are, the more inequality measures demonstrate higher distribution equality in the dominant population. Thus, LDCs are of interest to us when we compare distributions in terms of distribution inequality. 
We then provide estimators for LDCs and establish the asymptotic properties of the estimators.\footnote{The proposed estimators are also different from that proposed by \citet{zheng2018almost} since the definitions of the proposed LDCs are different.} Based on the asymptotic distributions, we construct bootstrap confidence intervals (CIs) for LDCs. 
LDCs are nonlinear transformations of the difference functions of two Lorenz curves. We show that this map is Hadamard directionally differentiable, and then we apply the bootstrap method of \citet{fang2014inference} to approximate the asymptotic distributions of the estimators and obtain the confidence intervals.\footnote{For more discussions on Hadamard directional differentiability and its applications, see \citet{dumbgen1993nondifferentiable}, \citet{andrews2000inconsistency}, \citet{bickel2012resampling}, \citet{hirano2012impossibility}, \citet{beare2015nonparametric},
	\citet{Beare2016global}, \citet{hansen2017regression}, \citet{Seo2016tests},  \citet{Beare2015improved}, \citet{chen2019improved}, \citet{Beare2017improved}, and \citet{sun2018ivvalidity}.}

The second degree inverse stochastic dominance \citep{atkinson1970measurement}, also known as generalized Lorenz dominance, is widely used for ranking distribution functions according to social welfare. However, as pointed out by \citet{aaberge2021ranking}, when the distribution functions intersect, the second degree inverse dominance criterion may have limitations in attaining an unambiguous ranking \citep{davies1995making,atkinson2008more}.
\citet{aaberge2021ranking} propose an approach for comparing intersecting distribution functions based on 
high-degree upward and downward inverse stochastic dominance criteria. They provide equivalence results between inverse stochastic dominance and the ranks of social welfare functions. Based on the framework of 
\citet{aaberge2021ranking}, we introduce the notions of \emph{almost upward inverse stochastic dominance} and \emph{almost downward inverse stochastic dominance}. We then show that these types of almost dominances are also related to the ranks of social welfare functions. Similar to LDCs, the degree of almost inverse stochastic dominance can be measured by the \emph{inverse stochastic dominance coefficients} (ISDCs). The smaller the ISDCs are, the more social planner preference functions show higher social welfare for the dominant distribution. We provide inference and construct bootstrap confidence intervals for ISDCs.

Almost Lorenz dominance can be viewed as a special but nontrivial case of almost stochastic dominance, which was first introduced by \citet{leshno2002preferred}. Almost stochastic dominance has been extensively studied by \citet{guo2013note}, \cite{tzeng2013revisiting}, \citet{denuit2014almost}, \citet{guo2014moment}, \citet{tsetlin2015generalized}, \citet{guo2016almost}, and \citet{huang2021estimating}.\footnote{See a comprehensive review on stochastic dominance in \citet{whang2019econometric}.}
We show that our method for almost Lorenz dominance can be easily generalized to almost stochastic dominance and provide bootstrap confidence intervals for the \emph{stochastic dominance coefficients} (SDCs). 

We examine the inequality growth in the United Kingdom as an empirical example in the paper. We find that the estimated upward inverse stochastic dominance coefficients are small and the bootstrap confidence intervals are narrow, which implies that there may exist almost inverse stochastic dominance between the years considered in this example. 

For the sake of simplicity of exposition, we focus on the almost Lorenz dominances in the main text. The results for other types of almost dominances are provided in the supplementary appendix. The paper is organized as follows. 
Section \ref{sec.LDC} introduces the notions of Lorenz dominance coefficients, provides the properties of these coefficients, and relates them to measures of inequality. Section \ref{sec.LDC estimation and inference} proposes estimators for LDCs and establishes the asymptotic distributions of the estimators. We then construct the confidence intervals for LDCs based on these asymptotic distributions. Section \ref{sec.simulation} provides simulation evidence for the finite sample properties of the proposed method. In Section \ref{sec.empirical}, we apply our methods to the empirical application. In the supplementary appendix, Sections \ref{sec.ISDC} and \ref{sec.ASD} extend the results to almost inverse stochastic dominance and almost stochastic dominance. Section \ref{sec.additional simulation} provides additional simulation results. Section \ref{sec.proofs} contains the proofs of the results in the paper. 

{\textbf{Notation \citep[][]{Beare2017improved}:}}
Throughout the paper, all the random elements are defined on
a probability space $( \Omega, \mathcal{A}, \mathbb{P} )$. 
For an arbitrary set $A$, let $\ell^\infty(A)$
be the set of bounded real-valued functions on $A$ equipped with the supremum norm $\Vert \cdot \Vert_{\infty}$ such that
$\Vert f \Vert_{\infty}=\sup_{x\in A} \vert f(x) \vert$ for every $f\in \ell^\infty(A)$.
For a subset $B$ of a metric space,
let $C(B)$ be the set of continuous real-valued functions on $B$.
Let $\leadsto$ denote the weak convergence defined in \citet[p.~4]{van1996weak}.
We follow the convention in \citet[p.~45]{folland2013real} and define 
\begin{align}\label{eq.0timesinf}
    0\cdot\infty=0.
\end{align}

\section{Almost Lorenz Dominance}\label{sec.LDC}

We suppose that the CDFs $F_{1}$ and $F_{2}$ satisfy the following regularity conditions, which guarantee the weak convergence of the estimated quantile functions \citep{K17,Beare2017improved}.  
\begin{assumption}\label{ass.distribution}
	For $j=1,2$, $F_{j}\left(  0\right)  =0$ and $F_{j}$ is continuously
	differentiable on the interior of its support with $F_{j}^{\prime}\left(
	x\right)  >0$ for all $x\in\left(  0,\infty\right)  $. In addition, $F_{j}$
	has finite $\left(  2+\epsilon\right)  $th absolute moment for some
	$\epsilon>0$. 
\end{assumption}

With Assumption \ref{ass.distribution}, we introduce the Lorenz curves. Let $Q_1$ and $Q_2$ denote the quantile functions corresponding to $F_1$ and $F_2$, respectively, that is,
\begin{align}
	Q_j(p)&=\inf\left\{x\in [0,\infty):F_j(x)\ge p \right\},\quad p\in[0,1].
\end{align}
When $F_j$ has finite first moment $\mu_j$ as implied by Assumption \ref{ass.distribution}, the quantile function $Q_j$ is integrable with $\int_0^1 Q_j(p)\mathrm{d}p=\mu_j$. Under Assumption \ref{ass.distribution} with $\mu_j>0$, $Q_j(p)=F_j^{-1}(p)$ for all $p\in(0,1)$, and the Lorenz curve $L_j$ corresponding to $F_j$ can then be defined as
\begin{align}
	L_{j}(p)&=\frac{1}{\mu_{j}}\int_{0}^{p}Q_{j}(t)\mathrm{d}t,\quad p\in[0,1].
\end{align}
We say $F_1$ (weakly) Lorenz dominates $F_2$ if $L_1(p)\ge L_2(p)$ for all $p\in[0,1]$.
We now follow \citet{zheng2018almost} and introduce the notion of almost Lorenz dominance. For all arbitrary CDFs $F_{1}:[0,\infty)\to[0,1]$ and $F_{2}:[0,\infty)\to[0,1]$ such that the corresponding Lorenz curves exist, define
\[
S\left(  F_{1},F_{2}\right)  =\left\{  p\in\left[  0,1\right]  :L_{1}\left(
p\right)  <L_{2}\left(  p\right)  \right\}  .
\]

\begin{definition}
	\label{def.dominance} For every $\varepsilon\in[  0,1/2)  $, \emph{the CDF $F_{1}$
	$\varepsilon$-almost Lorenz dominates the CDF $F_{2}$} ($F_{1}$ $\varepsilon
	$-ALD $F_{2}$), if
	\begin{align}\label{eq.ALD inequality}
	\int_{S\left(  F_{1},F_{2}\right)  }\left(  L_{2}\left(  p\right)
	-L_{1}\left(  p\right)  \right)  \mathrm{d}p\leq\varepsilon\int_{0}%
	^{1}\left\vert L_{2}\left(  p\right)  -L_{1}\left(  p\right)  \right\vert
	\mathrm{d}p.
	\end{align}
\end{definition}

\begin{remark}
\citet{zheng2018almost} defines the almost Lorenz dominance for all $\varepsilon\in(0,1/2)$, while we allow $\varepsilon=0$. The inequality \eqref{eq.ALD inequality} holds with $\varepsilon=0$ if and only if $F_1$ Lorenz dominates $F_2$. Thus, by Definition \ref{def.dominance}, almost Lorenz dominance is a weaker and generalized version of Lorenz dominance. 
\end{remark}

\begin{lemma}\label{lemma.ald equivalence}
For every $\varepsilon\in [0,1/2)$, $F_{1}$ $\varepsilon$-ALD $F_{2}$ if and
	only if there is $c\in\left[  0,\varepsilon\right]  $ such that
	\begin{align}\label{eq.ALD2}
	\int_{S\left(  F_{1},F_{2}\right)  }\left(  L_{2}\left(  p\right)
	-L_{1}\left(  p\right)  \right)  \mathrm{d}p=c\int_{0}^{1}\left\vert
	L_{2}\left(  p\right)  -L_{1}\left(  p\right)  \right\vert \mathrm{d}p.
	\end{align}
\end{lemma}

We define the distance function of the two Lorenz curves by
\begin{align}\label{eq.phi}
\phi\left(  p\right)  =L_{2}\left(  p\right)  -L_{1}\left(  p\right),\quad p\in[0,1].
\end{align}
Based on the equality in \eqref{eq.ALD2}, we follow \citet{zheng2018almost} and define the \emph{$F_1$-$F_2$ Lorenz Dominance Coefficient} as follows.\footnote{In Section 2.3 of \citet{zheng2018almost}, the Lorenz dominance coefficient is defined in a slightly different way.}

\begin{definition}\label{def.LDC}
Suppose that there are two Lorenz curves $L_1$ and $L_2$ corresponding to the CDFs $F_1$ and $F_2$, respectively. The \emph{$F_1$-$F_2$ Lorenz dominance coefficient ($F_1$-$F_2$ LDC)}, denoted by $c(L_1,L_2)$, is defined as
\begin{align}\label{eq.ALDC}
    c(L_1,L_2)=\inf\left\{\varepsilon\in[0,1]:\int_{S\left(  F_{1},F_{2}\right)  }\left(  L_{2}\left(  p\right)
	-L_{1}\left(  p\right)  \right)  \mathrm{d}p\leq\varepsilon\int_{0}
	^{1}\left\vert L_{2}\left(  p\right)  -L_{1}\left(  p\right)  \right\vert
	\mathrm{d}p\right\}.
\end{align}
\end{definition}

The following lemma summarizes the properties of the LDC $c(L_1,L_2)$.

\begin{lemma}\label{lemma.c properties}
The LDC $c(L_1,L_2)$ is the smallest $\varepsilon$ in $[0,1]$ such that \eqref{eq.ALD inequality} holds, $c(L_1,L_2)=0$ if and only if $F_1$ Lorenz dominates $F_2$, and $c(L_1,L_2)=1$ implies that $F_2$ Lorenz dominates $F_1$. With \eqref{eq.0timesinf}, it follows that
\begin{align}\label{eq.LDC}
    c(L_1,L_2)=\frac{\int_{S\left(  F_{1},F_{2}\right)  }\left(  L_{2}\left(  p\right)
	-L_{1}\left(  p\right)  \right)  \mathrm{d}p}{\int_{0}^{1}\left\vert
	L_{2}\left(  p\right)  -L_{1}\left(  p\right)  \right\vert \mathrm{d}p}=\frac{\int_{0}^{1}\max\left\{  \phi\left(  p\right)  ,0\right\}
\mathrm{d}p}{\int_{0}^{1}\max\left\{  \phi\left(  p\right)  ,0\right\}
\mathrm{d}p+\int_{0}^{1}\max\left\{  -\phi\left(  p\right)  ,0\right\}
\mathrm{d}p}.
\end{align}
In addition, if $c(L_1,L_2)\in(0,1]$, then $c(L_2,L_1)=1-c(L_1,L_2)$.
\end{lemma}
  
According to Lemma \ref{lemma.c properties}, $F_1$ $\varepsilon$-ALD $F_2$ for all $\varepsilon\in[c(L_1,L_2),1/2)$ if $c(L_1,L_2)<1/2$. On the other hand, $c(L_1,L_2)>1/2$ implies that $F_2$ $\varepsilon$-ALD $F_1$ for all $\varepsilon\in[1-c(L_1,L_2),1/2)$. Thus, $c(L_1,L_2)$ presents the almost Lorenz dominance relationship between $F_1$ and $F_2$, and provides all $\varepsilon$ for which $\varepsilon$-ALD holds.

When $L_1$ and $L_2$ cross, given some $\varepsilon\in(0,1/2)$, one interesting hypothesis is
\begin{align*}
    \mathrm{H}_0: \int_{S\left(  F_{1},F_{2}\right)  }\left(  L_{2}\left(  p\right)
	-L_{1}\left(  p\right)  \right)  \mathrm{d}p\leq\varepsilon\int_{0}
	^{1}\left\vert L_{2}\left(  p\right)  -L_{1}\left(  p\right)  \right\vert
	\mathrm{d}p.
\end{align*}
Lemma \ref{lemma.c properties} shows that it is equivalent to test the following hypothesis on $c(L_1,L_2)$:
\begin{align*}
    \mathrm{H}_0: c(L_1,L_2)\le \varepsilon.
\end{align*}

The LDC provided in Definition \ref{def.LDC} is closely related to Gini-type measures that are weighted averages of the area between the
diagonal line and Lorenz curves. Following \citet{shorrocks2002approximating}, for every possible nonnegative weighting function $\theta$, a Gini-type measure for distribution $F$ can be defined as 
\begin{align}\label{eq.Gini measure}
    I(F,\theta)=\frac{\int_0^1 [p-L_F(p)]\theta(p)\mathrm{d}p}{\int_0^1 p\theta(p)\mathrm{d}p}, \quad p\in[0,1],
\end{align}
where $L_F$ denotes the Lorenz curve associated with the distribution $F$.

We follow \citet{zheng2018almost} and denote $\mathscr{B}$ the set of all Gini-type measures for all
possible weighting functions $\theta$. 
Define
\[
\tilde{\theta}\left(  p\right)  =\theta\left(  p\right)  /\int_{0}^{1}%
p\theta\left(  p\right)  \mathrm{d}p
\]
for every possible weighting function $\theta$.
It then follows that
\[
I\left(  F,\theta\right)  =\int_{0}^{1}\left[  p-L_{F}\left(  p\right)
\right]  \tilde{\theta}\left(  p\right)  \mathrm{d}p.
\]
For
every $0<\varepsilon<1/2$, we define
\begin{align}\label{eq.B*}
\mathscr{B}^{\ast}\left(  \varepsilon\right)  =\left\{  I\left(  \cdot
,\theta\right)  :\tilde{\theta}>0\text{ and }\sup_{p\in\left[
0,1\right]  }\left\{  \tilde{\theta}\left(  p\right)  \right\}
\leq\inf_{p\in\left[  0,1\right]  }\left\{  \tilde{\theta}\left(
p\right)  \right\}  \cdot\left(  \frac{1}{\varepsilon}-1\right)  \right\}  .    
\end{align}
As discussed in \citet{zheng2018almost}, the condition in \eqref{eq.B*} essentially requires that the largest weight in weighting the distance $p-L_F(p)$ cannot be larger than $(1/\varepsilon-1)$ times the smallest weight.

Proposition 1 of \citet{zheng2018almost} shows that for every $\varepsilon\in\left(  0,1/2\right)
$, $F$ $\varepsilon$-ALD $G$ if and only if $I\left(  F,\theta\right)  \leq
I\left(  G,\theta\right)  $ for all $I\left(  \cdot,\theta\right)
\in\mathscr{B}^{\ast}\left(  \varepsilon\right)  $. 
According to Lemma \ref{lemma.c properties}, for every $\varepsilon\geq c\left(  L_{1},L_{2}\right) $ with $c\left(  L_{1},L_{2}\right) \in(0,1/2)$, $F$
$\varepsilon$-ALD $G$ and thus $I\left(  F,\theta\right)  \leq I\left(
G,\theta\right)  $ for all $I\left(  \cdot,\theta\right)  \in\mathscr{B}^{\ast
}\left(  \varepsilon\right)  $. Therefore, we have the following proposition. 

\begin{proposition}\label{prop.c vs Gini}
If $c\left(  L_{1},L_{2}\right)  \in(0,1/2)$, it follows that
$I\left(  F_1,\theta\right)  \leq I\left(  F_2,\theta\right)$ \text{ for all Gini-type measures} $I\left(\cdot,\theta\right)  \in\cup_{\varepsilon\in[c\left(  L_{1},L_{2}\right),1/2)}\mathscr{B}^{\ast}\left(  \varepsilon\right)=\mathscr{B}^{\ast}\left( c\left(  L_{1},L_{2}\right) \right)$.
If $c\left(  L_{1},L_{2}\right)=0$, it follows that \linebreak$ I\left(  F_1,\theta\right)  \leq I\left(  F_2,\theta\right) \text{ for all }   I\left(\cdot,\theta\right)  \in\cup_{\varepsilon\in(0,1/2)}\mathscr{B}^{\ast}\left(  \varepsilon\right)$.
\end{proposition}

Proposition \ref{prop.c vs Gini} highlights the importance of LDC in relation to Gini-type measures. The smaller the LDC $c(L_1,L_2)$ is, the more Gini-type measures show higher equality in the distribution $F_1$ compared to $F_2$. With the knowledge of $c(L_1,L_2)$, we can infer the relationship between $F_1$ and $F_2$ based on a class of Gini-type measures. 

Let $P:[0,1]\to[0,1]$ be a continuous, differentiable, and concave preference function employed by a social decision maker that assigns weights to individual incomes according to their ranks in the income distribution \citep{yaari1988controversial,aaberge2001axiomatic,aaberge2009ranking}.
Following \citet{aaberge2009ranking}, we introduce the family of rank-dependent measures of inequality:
\begin{align}\label{eq.inequality measure J}
J_{P}(L_{j})=1-\int_{0}^{1}P^{\prime}\left(  t\right)  \mathrm{d}L_{j}\left(
t\right)  =1-\frac{1}{\mu_{j}}\int_{0}^{1}P^{\prime}\left(  t\right)
Q_{j}\left(  t\right)  \mathrm{d}t,
\end{align}
where $L_{j}$ is the Lorenz curve of the distribution $F_{j}$ with mean
$\mu_{j}$ and the weighting function $P^{\prime}$ is the derivative of the preference function $P$.
As illustrated by \citet{aaberge2001axiomatic}, $J_P(L)$ measures the degree of inequality within an income distribution, as represented by the Lorenz curve $L$. We follow \citet{aaberge2009ranking} and formally {define} the set $\mathrm{P}$ of preference functions by
\[
\mathrm{P}=\left\{
\begin{array}
[c]{c}%
P\in C([0,1]):P\left(  0\right)  =0,P\left(  1\right)  =1,P^{\prime}\left(
1\right)  =0,\\
P^{\prime}\left(  t\right)  >0\text{ and }P^{\left(  2\right)  }\left(
t\right)  <0\text{ for all }t\in\left(  0,1\right)
\end{array}
\right\}  .
\]

\subsection{Almost Upward Lorenz Dominance and Inequality Measure}
To unify notation, we let $L_j^1=L_j$ for $j=1,2$. For $m\ge 2$, define function
\[
L_{j}^{m}\left(  p\right)  =\int_{0}^{p}L_j^{m-1}\left(  t\right)
\mathrm{d}t,\quad p\in\left[  0,1\right]  .
\]

\citet{aaberge2009ranking} introduces the $m$th-degree upward Lorenz dominance for $m\ge2$ as illustrated in the following. 
\begin{definition}
A distribution $F_{1}$ $m$th-degree upward Lorenz
dominates a distribution $F_{2}$ if $L_{1}^{m}\left(  p\right)
\geq L_{2}^{m}\left(  p\right)  $ for all $p\in\left[  0,1\right]  $.
\end{definition}

We generalize upward Lorenz dominance in \citet{aaberge2009ranking} to almost upward Lorenz dominance. 
\begin{definition}\label{def.AmULD}
For every $\varepsilon_m\in [  0,1/2)  $, the distribution $F_{1}$
$\varepsilon_{m}$-almost $m$th-degree upward Lorenz dominates the distribution $F_{2}$
($F_{1}$ $\varepsilon_{m}$-A$m$ULD $F_{2}$), if
\begin{align}\label{eq.AmULD inequality}
\int_{0}^{1}\max\left(  L_{2}^{m}\left(  p\right)  -L_{1}^{m}\left(  p\right)
,0\right)  \mathrm{d}p\leq\varepsilon_{m}\int_{0}^{1}\left\vert L_{2}%
^{m}\left(  p\right)  -L_{1}^{m}\left(  p\right)  \right\vert \mathrm{d}p.
\end{align}    
\end{definition}

\begin{remark}
    In Definition \ref{def.AmULD}, if $\varepsilon_m=0$, then $F_1$ $m$th-degree upward Lorenz dominates $F_2$.
\end{remark}

For every $\varepsilon_{m}\in\left(  0,1/2\right)  $, we define two sets of preference functions
\[
\mathrm{P}_{m}\left(  \varepsilon_{m}\right)  =\left\{
\begin{array}
[c]{c}%
P\in\mathrm{P}:\left(  -1\right)  ^{m}P^{\left(  m+1\right)  }\geq0,P^{\left(
j\right)  }\left(  1\right)  =0\text{ for }j\in\{2,\ldots, m\},\\
\sup_{t}\left\{  \left(  -1\right)  ^{m}P^{\left(  m+1\right)  }\left(
t\right)  \right\}  \leq\inf_{t}\left\{  \left(  -1\right)  ^{m}P^{\left(
m+1\right)  }\left(  t\right)  \right\}  \cdot\left(  \frac{1}{\varepsilon
_{m}}-1\right)
\end{array}
\right\}
\]
and%
\[
\mathrm{P}_{m}^{\prime}\left(  \varepsilon_{m}\right)  =\left\{
\begin{array}
[c]{c}%
P\in\mathrm{P}:\left(  -1\right)  ^{m}P^{\left(  m+1\right)  }\geq0,\left(
-1\right)  ^{j-1}P^{\left(  j\right)  }\left(  1\right)  \geq0\text{ for
}j\in\{2,\ldots, m\},\\
\sup_{t}\left\{  \left(  -1\right)  ^{m}P^{\left(  m+1\right)  }\left(
t\right)  \right\}  \leq\inf_{t}\left\{  \left(  -1\right)  ^{m}P^{\left(
m+1\right)  }\left(  t\right)  \right\}  \cdot\left(  \frac{1}{\varepsilon
_{m}}-1\right)
\end{array}
\right\}  .
\]

The sets $\mathrm{P}_m(\varepsilon_m)$ and $\mathrm{P}'_m(\varepsilon_m)$ are constructed based on \citet[p.~248]{aaberge2009ranking}.
As discussed in \citet{aaberge2001axiomatic}, risk aversion is equivalent to second degree stochastic dominance and requires the concavity of the utility function. Theorem 2.1 of \citet{aaberge2009ranking} shows that the concavity on $P$ is related to first degree Lorenz dominance and ``inequality aversion'' \citep[Definition 2.3]{aaberge2009ranking}. The requirement $(-1)^mP^{(m+1)}\ge0$ in $\mathrm{P}_m(\varepsilon_m)$ and $\mathrm{P}'_m(\varepsilon_m)$ may be viewed as an extension of the concavity condition to higher degree upward dominance.
Moreover, as shown in the proof of Proposition 2.2 \citep[Proof of Theorem 3.1A]{aaberge2009ranking}, when $J_P(L_2)-J_P(L_1)$ is written as a function of $L^m_1-L^m_2$, $(-1)^mP^{(m+1)}$ plays a role of the weighting function. We require it to be nonnegative in $\mathrm{P}_m(\varepsilon_m)$ and $\mathrm{P}'_m(\varepsilon_m)$. 
The following proposition relates almost upward Lorenz dominance to the inequality measures $J_P$ in \eqref{eq.inequality measure J} for $P$ in $\mathrm{P}_m(\varepsilon_m)$ and $\mathrm{P}'_m(\varepsilon_m)$. 
\begin{proposition}\label{prop.AmULD}
For every $\varepsilon_{m}\in\left(  0,1/2\right)  $, if $F_{1}$
$\varepsilon_{m}$-A$m$ULD $F_{2}$, then $J_{P}\left(  L_{1}\right)  \leq
J_{P}\left(  L_{2}\right)  $ for every $P\in\mathrm{P}_{m}\left(
\varepsilon_{m}\right)  $. If $F_{1}$ $\varepsilon_{m}$-A$m$ULD $F_{2}$ and
$L_{1}^{j}\left(  1\right)  \geq L_{2}^{j}\left(  1\right)  $ for all
$j\in\left\{  2,\ldots,m\right\}  $, then $J_{P}\left(  L_{1}\right)  \leq
J_{P}\left(  L_{2}\right)  $ for every $P\in\mathrm{P}_{m}^{\prime}\left(
\varepsilon_{m}\right)  $.    
\end{proposition}

For $m\ge 2$, following Lemma \ref{lemma.ald equivalence}, it is straightforward to show that $F_{1}$ $\varepsilon_{m}$-A$m$ULD $F_{2}$
if and only if there is $c_{m}\in\left[  0,\varepsilon_{m}\right]  $ such
that
\[
\int_{0}^{1}\max\left(  L_{2}^{m}\left(  p\right)  -L_{1}^{m}\left(  p\right)
,0\right)  \mathrm{d}p= c_{m}\int_{0}^{1}\left\vert L_{2}^{m}\left(
p\right)  -L_{1}^{m}\left(  p\right)  \right\vert \mathrm{d}p.
\]
Similar to \eqref{eq.ALDC}, we define the $m$th-degree upward Lorenz dominance coefficient $c_m^u(L_1,L_2)$, where the superscript ``$u$'' represents ``upward''.
\begin{definition}\label{def.mULDC}
For $m\ge2$, the \emph{$F_1$-$F_2$ $m$th-degree upward Lorenz dominance coefficient ($F_1$-$F_2$ $m$ULDC)} $c_m^u(L_1,L_2)$ is defined as
\begin{align}
    &c_m^u(L_1,L_2)=\notag\\
    &\inf\left\{\varepsilon_m\in[0,1]:\int_0^1\max\left(  L_{2}^{m}\left(  p\right)  -L_{1}^{m  }\left(  p\right),0  \right)
\mathrm{d}p\leq\varepsilon_{m}\int_{0}^{1}\left\vert L_{2}^{m}\left(p  \right)  -L_{1}^{m}\left(  p\right)  \right\vert
\mathrm{d}p\right\}.
\end{align}
\end{definition}

For $m\ge2$, define the difference function
\[
\phi_m^u\left(  p\right)  =L_{2}^{m}\left(  p\right) -L_{1}^{m}\left(  p\right),\quad p\in[0,1].
\]
The following lemma summarizes the properties of the $m$ULDC $c_m^u(L_1,L_2)$.

\begin{lemma}\label{lemma.cmu L properties}
The $m$ULDC $c_m^u(L_1,L_2)$ is the smallest $\varepsilon_m$ in $[0,1]$ such that \eqref{eq.AmULD inequality} holds. With \eqref{eq.0timesinf}, it follows that
\begin{align}\label{eq.mULDC}
    c_m^u(L_1,L_2)&=\frac{\int_0^1\max\left(  L_{2}^{m}\left(  p\right)  -L_{1}^{m}\left(  p\right) ,0 \right)
\mathrm{d}p}{\int_{0}^{1}\left\vert L_{2}^{m}\left(  p\right) -L_{1}^{m}\left(  p\right)  \right\vert
\mathrm{d}p}\notag\\
&=\frac{\int_{0}^{1}\max\left\{  \phi_m^u\left(  p\right)  ,0\right\}
\mathrm{d}p}{\int_{0}^{1}\max\left\{  \phi_m^u\left(  p\right)  ,0\right\}
\mathrm{d}p+\int_{0}^{1}\max\left\{  -\phi_m^u\left(  p\right)  ,0\right\}
\mathrm{d}p}.
\end{align}
In addition, if $c_m^u(L_1,L_2)\in(0,1]$, then $c_m^u(L_2,L_1)=1-c_m^u(L_1,L_2)$.
\end{lemma}
  
According to Lemma \ref{lemma.cmu L properties}, $F_1$ $\varepsilon_m$-A$m$ULD $F_2$ for all $\varepsilon_m\in[c_m^u(L_1,L_2),1/2)$ if $c_m^u(L_1,L_2)<1/2$. On the other hand, $c_m^u(L_1,L_2)>1/2$ implies that $F_2$ $\varepsilon_m$-A$m$ULD $F_1$ for all $\varepsilon_m\in[1-c_m^u(L_1,L_2),1/2)$. Thus, $c_m^u(L_1,L_2)$ presents the degree of almost upward Lorenz dominance relationship between $F_1$ and $F_2$, and provides all $\varepsilon_m$ such that the $\varepsilon_m$-A$m$ULD holds.

\begin{proposition}\label{prop.ULDC properties}
If $c_m^u(L_1,L_2)\in(0,1/2)$, it then follows that 
$J_P(L_1)\le J_P(L_2)$ for all preference functions $P\in \cup_{\varepsilon_m\in[c_m^u(L_1,L_2),1/2)}\mathrm{P}_m(\varepsilon_m)$.    
If, in addition, 
$L_{1}^{j}\left(  1\right)  \geq L_{2}^{j}\left(  1\right)  $ for all
$j\in\left\{  2,\ldots,m\right\}  $, then $J_{P}\left(  L_{1}\right)  \leq
J_{P}\left(  L_{2}\right)  $ for all $P\in\cup_{\varepsilon_m\in[c_m^u(L_1,L_2),1/2)}\mathrm{P}_{m}^{\prime}\left(
\varepsilon_{m}\right)  $. 
\end{proposition}

Similar to Proposition \ref{prop.c vs Gini}, Proposition \ref{prop.ULDC properties} shows the importance of ULDC concerning inequality measures $J_P$. The smaller the ULDC $c_m^u(L_1,L_2)$ is, the more inequality measures show higher equality in the distribution $F_1$ compared to $F_2$. With the knowledge of $c_m^u(L_1,L_2)$, we can infer the relationship between $F_1$ and $F_2$ based on a class of inequality measures. 

\subsection{Almost Downward Lorenz Dominance and Inequality Measure}
Define function
\[
\tilde{L}_{j}^{2}\left(  p\right)  =\int_{p}^{1}\left(  1-L_{j}\left(
t\right)  \right)  \mathrm{d}t=\frac{1}{\mu_{j}}\int_{p}^{1}\left(
t-p\right)  Q_j\left(  t\right)  \mathrm{d}t,\quad p\in\left[  0,1\right],
\]
and for $m\geq3$,
\[
\tilde{L}_{j}^{m}\left(  p\right)  =\int_{p}^{1}\tilde{L}_j^{m-1}\left(
t\right)  \mathrm{d}t,\quad p\in\left[  0,1\right]  .
\]

\citet{aaberge2009ranking} introduces the $m$th-degree downward Lorenz dominance for $m\ge2$ as illustrated in the following. 
\begin{definition}
A distribution $F_{1}$ $m$th-degree downward Lorenz
dominates a distribution $F_{2}$ if $\tilde{L}_{1}^{m}\left(  p\right)
\le \tilde{L}_{2}^{m}\left(  p\right)  $ for all $p\in\left[  0,1\right]  $.
\end{definition}

We generalize downward Lorenz dominance in \citet{aaberge2009ranking} to almost downward Lorenz dominance. 
\begin{definition}\label{def.AmDLD}
For every $\varepsilon_m\in[  0,1/2)  $, the distribution $F_{1}$
$\varepsilon_{m}$-almost $m$th-degree downward Lorenz dominates the distribution
$F_{2}$ ($F_{1}$ $\varepsilon_{m}$-A$m$DLD $F_{2}$), if
\begin{align}\label{eq.AmDLD inequality}
\int_{0}^{1}\max\left(  \tilde{L}_{1}^{m}\left(  p\right)  -\tilde{L}_{2}%
^{m}\left(  p\right)  ,0\right)  \mathrm{d}p\leq\varepsilon_{m}\int_{0}%
^{1}\left\vert \tilde{L}_{1}^{m}\left(  p\right)  -\tilde{L}_{2}^{m}\left(
p\right)  \right\vert \mathrm{d}p.
\end{align}
\end{definition}

\begin{remark}
    In Definition \ref{def.AmDLD}, if $\varepsilon_m=0$, then $F_1$ $m$th-degree downward Lorenz dominates $F_2$.
\end{remark}

For every $\varepsilon_{m}\in\left(  0,1/2\right)  $, we define
\[
\mathrm{\tilde{P}}_{m}\left(  \varepsilon_{m}\right)  =\left\{
\begin{array}
[c]{c}%
P\in\mathrm{P}:P^{\left(  m+1\right)  }\leq0,P^{\left(  j\right)  }\left(
0\right)  =0\text{ for }j\in\{2,\ldots, m\},\\
\sup_{t}\left\{  -P^{\left(  m+1\right)  }\left(  t\right)  \right\}  \leq
\inf_{t}\left\{  -P^{\left(  m+1\right)  }\left(  t\right)  \right\}
\cdot\left(  \frac{1}{\varepsilon_{m}}-1\right)
\end{array}
\right\}
\]
and%
\[
\mathrm{\tilde{P}}_{m}^{\prime}\left(  \varepsilon_{m}\right)  =\left\{
\begin{array}
[c]{c}%
P\in\mathrm{P}:P^{\left(  m+1\right)  }\leq0,P^{\left(  j\right)  }\left(
0\right)  \leq0\text{ for }j\in\{2,\ldots, m\},\\
\sup_{t}\left\{  -P^{\left(  m+1\right)  }\left(  t\right)  \right\}  \leq
\inf_{t}\left\{  -P^{\left(  m+1\right)  }\left(  t\right)  \right\}
\cdot\left(  \frac{1}{\varepsilon_{m}}-1\right)
\end{array}
\right\}  .
\]

The sets $\tilde{\mathrm{P}}_m(\varepsilon_m)$ and $\tilde{\mathrm{P}}'_m(\varepsilon_m)$ are constructed based on \citet[p.~248]{aaberge2009ranking}.
The requirement $P^{(m+1)}\le0$ in $\tilde{\mathrm{P}}_m(\varepsilon_m)$ and $\tilde{\mathrm{P}}'_m(\varepsilon_m)$ may be viewed as an extension of the concavity condition to higher degree downward dominance.
Moreover, as shown in the proof of Proposition \ref{prop.AmDLD}, when $J_P(L_2)-J_P(L_1)$ is written as a function of $\tilde{L}^m_2-\tilde{L}^m_1$, $-P^{(m+1)}$ plays a role of the weighting function. We require it to be nonnegative in $\tilde{\mathrm{P}}_m(\varepsilon_m)$ and $\tilde{\mathrm{P}}'_m(\varepsilon_m)$. The following proposition relates almost downward Lorenz dominance to the inequality measures $J_P$ in \eqref{eq.inequality measure J} for $P$ in $\tilde{\mathrm{P}}_m(\varepsilon_m)$ and $\tilde{\mathrm{P}}'_m(\varepsilon_m)$. 
\begin{proposition}\label{prop.AmDLD}
For every $\varepsilon_{m}\in\left(  0,1/2\right)  $, if $F_{1}$
$\varepsilon_{m}$-A$m$DLD $F_{2}$, then $J_{P}\left(  L_{1}\right)  \leq
J_{P}\left(  L_{2}\right)  $ for every $P\in\mathrm{\tilde{P}}_{m}\left(
\varepsilon_{m}\right)  $. If $F_{1}$ $\varepsilon_{m}$-A$m$DLD $F_{2}$ and
$\tilde{L}_{1}^{j}\left(  0\right)  \leq\tilde{L}_{2}^{j}\left(  0\right)  $
for all $j\in\left\{  2,\ldots,m\right\}  $, then $J_{P}\left(  L_{1}\right)
\leq J_{P}\left(  L_{2}\right)  $ for every $P\in\mathrm{\tilde{P}}%
_{m}^{\prime}\left(  \varepsilon_{m}\right)  $.
\end{proposition}

For $m\ge 2$, following Lemma \ref{lemma.ald equivalence}, it is straightforward to show that $F_{1}$ $\varepsilon_{m}$-A$m$DLD $F_{2}$
if and only if there is $c_{m}\in\left[  0,\varepsilon_{m}\right]  $ such
that
\[
\int_{0}^{1}\max\left( \tilde L_{1}^{m}\left(  p\right)  -\tilde L_{2}^{m}\left(  p\right)
,0\right)  \mathrm{d}p= c_{m}\int_{0}^{1}\left\vert \tilde L_{1}^{m}\left(
p\right)  -\tilde L_{2}^{m}\left(  p\right)  \right\vert \mathrm{d}p.
\]
Similar to \eqref{eq.ALDC}, we now define the $m$th-degree downward Lorenz dominance coefficient  $c_m^d(L_1,L_2)$, where the superscript ``$d$'' represents ``downward''.
\begin{definition}\label{def.mDLDC}
For $m\ge2$, the \emph{$F_1$-$F_2$ $m$th-degree downward Lorenz dominance coefficient ($F_1$-$F_2$ $m$DLDC)} $c_m^d(L_1,L_2)$ is defined as
\begin{align}
    &c_m^d(L_1,L_2)=\notag\\
    &\inf\left\{\varepsilon_m\in[0,1]:\int_0^1\max\left( \tilde  L_{1}^{m}\left(  p\right)  -\tilde L_{2}^{m  }\left(  p\right),0  \right)
\mathrm{d}p\leq\varepsilon_{m}\int_{0}^{1}\left\vert \tilde L_{1}^{m}\left(p  \right)  -\tilde L_{2}^{m}\left(  p\right)  \right\vert
\mathrm{d}p\right\}.
\end{align}
\end{definition}

For $m\ge2$, we define the difference function
\[
\phi_m^d\left(  p\right)  =\tilde L_{1}^{m}\left(  p\right) -\tilde L_{2}^{m}\left(  p\right),\quad p\in[0,1].
\]
The following lemma summarizes the properties of the $m$DLDC $c_m^d(L_1,L_2)$.

\begin{lemma}\label{lemma.cmd L properties}
The $m$DLDC $c_m^d(L_1,L_2)$ is the smallest $\varepsilon_m$ in $[0,1]$ such that \eqref{eq.AmDLD inequality} holds. With \eqref{eq.0timesinf}, it follows that
\begin{align}\label{eq.mDLDC}
    c_m^d(L_1,L_2)&=\frac{\int_0^1\max\left( \tilde  L_{1}^{m}\left(  p\right)  -\tilde L_{2}^{m}\left(  p\right) ,0 \right)
\mathrm{d}p}{\int_{0}^{1}\left\vert\tilde  L_{1}^{m}\left(  p\right) -\tilde L_{2}^{m}\left(  p\right)  \right\vert
\mathrm{d}p}\notag\\
&=\frac{\int_{0}^{1}\max\left\{  \phi_m^d\left(  p\right)  ,0\right\}
\mathrm{d}p}{\int_{0}^{1}\max\left\{  \phi_m^d\left(  p\right)  ,0\right\}
\mathrm{d}p+\int_{0}^{1}\max\left\{  -\phi_m^d\left(  p\right)  ,0\right\}
\mathrm{d}p}.
\end{align}
In addition, if $c_m^d(L_1,L_2)\in(0,1]$, then $c_m^d(L_2,L_1)=1-c_m^d(L_1,L_2)$.
\end{lemma}
  
According to Lemma \ref{lemma.cmd L properties}, $F_1$ $\varepsilon_m$-A$m$DLD $F_2$ for all $\varepsilon_m\in[c_m^d(L_1,L_2),1/2)$ if $c_m^d(L_1,L_2)<1/2$. On the other hand, $c_m^d(L_1,L_2)>1/2$ implies that $F_2$ $\varepsilon_m$-A$m$DLD $F_1$ for all $\varepsilon_m\in[1-c_m^d(L_1,L_2),1/2)$. Thus, $c_m^d(L_1,L_2)$ presents the degree of almost downward Lorenz dominance relationship between $F_1$ and $F_2$, and provides all $\varepsilon_m$ such that the $\varepsilon_m$-A$m$DLD holds.

\begin{proposition}\label{prop.DLDC properties}
If $c_m^d(L_1,L_2)\in(0,1/2)$, it then follows that 
$J_P(L_1)\le J_P(L_2)$ for all preference functions $P\in \cup_{\varepsilon_m\in[c_m^d(L_1,L_2),1/2)}\tilde{\mathrm{P}}_m(\varepsilon_m)$.     
If, in addition, 
$\tilde L_{1}^{j}\left(  0\right)  \leq \tilde L_{2}^{j}\left(  0\right)  $ for all
$j\in\left\{  2,\ldots,m\right\}  $, then $J_{P}\left(  L_{1}\right)  \leq
J_{P}\left(  L_{2}\right)  $ for all $P\in\cup_{\varepsilon_m\in[c_m^d(L_1,L_2),1/2)}\tilde{\mathrm{P}}_{m}^{\prime}\left(
\varepsilon_{m}\right)  $. 
\end{proposition}

Similar to Proposition \ref{prop.c vs Gini}, Proposition \ref{prop.DLDC properties} shows the importance of DLDC concerning inequality measures $J_P$. The smaller the DLDC $c_m^d(L_1,L_2)$ is, the more inequality measures show higher equality in the distribution $F_1$ compared to $F_2$. With the knowledge of $c_m^d(L_1,L_2)$, we can infer the relationship between $F_1$ and $F_2$ based on a class of inequality measures.

\subsection{Estimation and Inference}\label{sec.LDC estimation and inference}

As discussed in the previous sections, LDCs play an important role in comparing the equality of income (or wealth) distributions between two populations. In this section, we consider the estimation and inference of LDCs. 
\subsubsection{Sampling Frameworks}\label{sec.LD sampling frameworks}

Following \citet{BDB14} and \citet{Beare2017improved}, we consider two alternative frameworks for sampling from $F_1$ and $F_2$. From $F_1$ and $F_2$, we draw identically and independently distributed (iid) samples $\{X_i^j\}_{i=1}^{n_j}$ ($j=1,2$) that satisfy the following assumption.  

\begin{assumption}\label{ass.data} \citep{BDB14,Beare2017improved}
	The iid samples $\{X_i^1\}_{i=1}^{n_1}$ and $\{X_i^2\}_{i=1}^{n_2}$ drawn from $F_1$ and $F_2$ satisfy one of the following conditions.
	\begin{enumerate}[label=(\roman*)]
		\item \emph{Independent samples:} $\{X_i^1\}_{i=1}^{n_1}$ and $\{X_i^2\}_{i=1}^{n_2}$ are mutually independent, and the sample sizes $n_1$ and $n_2$ are treated as functions of an underlying index $n\in\mathbb N$ such that as $n\to\infty$, 
  \begin{align}\label{samplesizes}
	\frac{n_1n_2}{n_1+n_2}\to\infty\quad\text{and}\quad\frac{n_1}{n_1+n_2}\to\lambda\in[0,1].
\end{align}

		\item \emph{Matched pairs:} The sample sizes $n_1$ and $n_2$ satisfy $n_1=n_2=n$, the pairs $\{(X_i^1,X_i^2)\}_{i=1}^{n}$ are iid, and the bivariate copula for those pairs has maximal correlation strictly less than one \citep[see, e.g.,][Definition 3.2]{beare2010copulas}.
	\end{enumerate}
\end{assumption}

\subsubsection{Construction of Estimator}

We define maps $\mathcal{F}_{1}:\ell^{\infty}\left(  \left[  0,1\right]  \right)
\rightarrow\mathbb{R}$, $\mathcal{F}_{2}:\ell^{\infty}\left(  \left[
0,1\right]  \right)  \rightarrow\mathbb{R}$, and $\mathcal{F}:\ell^{\infty
}\left(  \left[  0,1\right]  \right)  \rightarrow\mathbb{R}$ by
\begin{align}\label{eq.F map}
&\mathcal{F}_{1}\left(  \psi\right)     =\int_{0}^{1}\max\left\{  \psi\left(
p\right)  ,0\right\}  \mathrm{d}p,\quad \mathcal{F}_{2}\left(  \psi\right)
=\int_{0}^{1}\max\left\{  -\psi\left(  p\right)  ,0\right\}  \mathrm{d}p,\notag\\
&\text{and }\mathcal{F}\left(  \psi\right)     =\frac{\mathcal{F}_{1}\left(
\psi\right)  }{\mathcal{F}_{1}\left(  \psi\right)  +\mathcal{F}_{2}\left(
\psi\right)  },\quad\psi\in\ell^{\infty}\left(  \left[  0,1\right]  \right)  .
\end{align}
To unify notation, for $w\in\{u,d\}$, we let $\phi_1^w=\phi$ and $c_1^w(L_1,L_2)=c(L_1,L_2)$  defined in \eqref{eq.phi} and \eqref{eq.ALDC}, respectively.
Clearly, for $m\ge 1$, we have 
$c_m^w\left(  L_{1},L_{2}\right)  =\mathcal{F}\left(  \phi_m^w\right)$ by \eqref{eq.LDC}, \eqref{eq.mULDC}, and \eqref{eq.mDLDC}.
Following \citet{BDB14} and \citet{Beare2017improved},
for $j=1,2$, define the empirical CDF
\begin{align}\label{eq.empirical cdfs}
	\hat{F}_j(x)=\frac{1}{n_j}\sum_{i=1}^{n_j}
	{1}(X_i^j\leq x),\quad x\in[0,\infty),
\end{align}
the empirical quantile function
\begin{align}\label{eq.empirical quantile}
	\hat{Q}_j(p)=\inf\{x\in[0,\infty):\hat{F}_j(x)\geq p\},\quad p\in[0,1],
\end{align}
and the empirical Lorenz curve
\begin{align*}
	\hat{L}_j(p)=\frac{1}{\hat{\mu}_j}\int_0^p\hat{Q}_j(t)\mathrm{d}t,\quad p\in[0,1],
\end{align*}
where $\hat{\mu}_j=n_j^{-1}\sum_{i=1}^{n_j}X_i^j$ is the sample mean of $\{X_i^j\}_{i=1}^{n_j}$. 
Let $\hat{L}_j^1=\hat{L}_j$.
For $m\ge 2$, define
\[
\hat L_{j}^{m}\left(  p\right)  =\int_{0}^{p}\hat L_j^{m-1}\left(  t\right)
\mathrm{d}t,\quad p\in\left[  0,1\right]  .
\]
Define
\[
\hat{\tilde{L}}_{j}^{2}\left(  p\right)  =\int_{p}^{1}\left(  1-\hat L_{j}\left(
t\right)  \right)  \mathrm{d}t\text{ and for }m\ge 3, \hat{\tilde{L}}_{j}^{m}\left(  p\right)  =\int_{p}^{1}\hat{\tilde{L}}_j^{m-1}\left(
t\right)  \mathrm{d}t,\quad p\in\left[  0,1\right]  .
\]
The difference function between empirical Lorenz curves is defined by
\begin{align*}
	\hat{\phi}(p)=\hat{L}_2(p)-\hat{L}_1(p),\quad p\in[0,1].
\end{align*}
For $w\in\{u,d\}$, let $\hat{\phi}_1^w=\hat{\phi}$ for unified notation. For $m\ge 2$, define
\[
\hat{\phi}_m^u\left(  p\right)  =\hat L_{2}^{m}\left(  p\right) -\hat L_{1}^{m}\left(  p\right)\text{ and }\hat{\phi}_m^d\left(  p\right)  =\hat{\tilde{L}}_{1}^{m}\left(  p\right) -\hat{\tilde{L}}_{2}^{m}\left(  p\right),\quad p\in[0,1].
\]
The estimator of $c_m^w(L_1,L_2)$ is given by $\hat{c}_m^w(L_1,L_2)=\mathcal{F}(\hat{\phi}_m^w)$.
\subsubsection{Asymptotic Analysis}\label{secbsasym}

We now establish the consistency of $\hat{c}_m^w(L_1,L_2)$ and derive the asymptotic distribution of
\begin{align}
    \sqrt{T_n}(\hat{c}_m^w(L_1,L_2)-{c}_m^w(L_1,L_2))=\sqrt{T_n}(\mathcal{F}(\hat{\phi}_m^w)-\mathcal{F}(\phi_m^w))
\end{align}
for all $m\ge1$ and $w\in\{u,d\}$, 
where $T_n=n_1n_2/(n_1+n_2)$. We follow \citet{fang2014inference} and introduce the following definition of Hadamard directional differentiability.
\begin{definition}
	\label{def.Hadamard directional}Let $\mathbb{D}$ and $\mathbb{E}$ be normed spaces. A map $\mathcal{G}:\mathbb{D}\to\mathbb{E}$ is said to be Hadamard
	directionally differentiable at $\psi\in\mathbb{D}$ tangentially to $\mathbb{D}_0\subset\mathbb{D}$ if there is a continuous
	map $\mathcal{G}_{\psi}^{\prime}:\mathbb{D}_0\rightarrow
	\mathbb{E}$ such that
	\begin{align}\label{Hadamard directional derivative}
		\lim_{n\rightarrow\infty}\left\Vert\frac{\mathcal{G}(\psi+t_{n}h_{n}%
			)-\mathcal{G}(\psi)}{t_{n}}-\mathcal{G}_{\psi}^{\prime}(h)\right\Vert
		_{\mathbb{E}}=0,
	\end{align}
	for all sequences $\{h_{n}\}\subset\mathbb{D}$ and $\{t_{n}\}\subset
	\mathbb{R}_{+}$ such that $t_{n}\downarrow0$ and $h_{n}\rightarrow h\in\mathbb{D}_0$.
\end{definition}

For every $\psi\in\ell^{\infty}\left(  \left[  0,1\right]  \right)  $, define
\begin{align}\label{eq.B0 B+}
B_{0}\left(  \psi\right)  =\left\{  p\in\left[  0,1\right]  :\psi\left(
p\right)  =0\right\}  \text{ and }B_{+}\left(  \psi\right)  =\left\{
p\in\left[  0,1\right]  :\psi\left(  p\right)  >0\right\}  .
\end{align}
Here, $B_0(\phi_m^w)$ is the contact set we need to estimate later. For more discussions on the estimation of contact sets, see, for example, \citet{linton2010improved} for an improved bootstrap test of stochastic dominance and \citet{lee2018testing} for testing functional inequalities.
By Lemma S.4.5 of \citet{fang2014inference}, we have that $\mathcal{F}_{1}$ and $\mathcal{F}_{2}$ are both Hadamard directionally
differentiable at $\phi_m^w$ tangentially to $\ell^{\infty}\left(  \left[
0,1\right]  \right)  $ with
\[
\mathcal{F}_{1\phi_m^w}^{\prime}\left(  h\right)  =\int_{B_{+}\left(  \phi_m^w\right)
}h\left(  p\right)  \mathrm{d}p+\int_{B_{0}\left(  \phi_m^w\right)  }\max\left\{
h\left(  p\right)  ,0\right\}  \mathrm{d}p
\]
and%
\[
\mathcal{F}_{2\phi_m^w}^{\prime}\left(  h\right)  =\int_{B_{+}\left(
-\phi_m^w\right)  }-h\left(  p\right)  \mathrm{d}p+\int_{B_{0}\left(
\phi_m^w\right)  }\max\left\{  -h\left(  p\right)  ,0\right\}  \mathrm{d}p
\]
for all $h\in\ell^{\infty}\left(  \left[  0,1\right]  \right)  $. 

Following \citet{Beare2017improved}, we let $\mathcal{B}$ be a centered Gaussian random element in $C(\left[0,1\right]  ^{2})$ with covariance kernel
\[
Cov\left(  \mathcal{B}\left(  u,v\right)  ,\mathcal{B(}u^{\prime},v^{\prime
})\right)  =\mathbf{C}\left(  u\wedge u^{\prime},v\wedge v^{\prime}\right)  -\mathbf{C}\left(
u,v\right)  \mathbf{C}(u^{\prime},v^{\prime}),
\]
where $\mathbf{C}(u,v)=uv$ under Assumption \ref{ass.data}(i)  and  $\mathbf C$ is the unique copula function for the pair $(X_{i}^{1},X_{i}^{2})  $ under Assumption \ref{ass.data}(ii). Let $\mathcal{B}_{1}$ and $\mathcal{B}_{2}$ be
the centered Gaussian random elements in $C(\left[  0,1\right]  )$ such that
$\mathcal{B}_{1}\left(  u\right)  =\mathcal{B}\left(  u,1\right)  $ and
$\mathcal{B}_{2}\left(  v\right)  =\mathcal{B}\left(  1,v\right)  $.
Under Assumptions \ref{ass.distribution} and \ref{ass.data}, \citet{Beare2017improved} show that 
\begin{align}\label{eq.phi weak convergence}
\sqrt{T_n}(\hat{\phi}-\phi)\leadsto \mathbb{G}    
\end{align}
in $C([0,1])$ for some random element $\mathbb{G}=\lambda^{1/2}\mathcal{L}_{2}-\left(
1-\lambda\right)  ^{1/2}\mathcal{L}_{1}$, where $\mathcal{L}_{j}$ is a centered random element of $C([0,1])$ given by
	\begin{align}\label{eq.L}
	\mathcal L_j(p)=-\int_0^pL_j''(t)\mathcal B_j(t)\mathrm{d}t+L_j(p)\int_0^1L_j''(t)\mathcal B_j(t)\mathrm{d}t,\quad p\in[0,1].
	\end{align}
\citet[][Proposition 3.1]{Beare2017improved} show that
\begin{align}\label{eq.asymptotic variance}
    Var(\mathbb{G}(p))= Var\left( \frac{(1-\lambda)^{1/2}}{\mu_1}(L_1(p)X_i^1-Q_1(p)\wedge X_i^1)-\frac{\lambda^{1/2}}{\mu_2}(L_2(p)X_i^2-Q_2(p)\wedge X_i^2)            \right)
\end{align}
for every $p\in [0,1]$. We follow \citet{Beare2017improved} and provide the following lemma. 
\begin{lemma}\label{lemma.G kernel}
For all $p,p'\in[0,1]$, it follows that
\begin{align*}
E\left[  \mathbb{G}\left(  p\right)  \mathbb{G}(  p^{\prime})
\right]  
=&\,\lambda E\left[  \mathcal{L}_{2}\left(  p\right)  \mathcal{L}_{2}(
p^{\prime})  \right]  -\mathbb{\lambda}^{1/2}\left(  1-\lambda\right)
^{1/2}E\left[  \mathcal{L}_{2}\left(  p\right)  \mathcal{L}_{1}(
p^{\prime})  \right] \\& -\mathbb{\lambda}^{1/2}\left(  1-\lambda\right)
^{1/2}E\left[  \mathcal{L}_{1}\left(  p\right)  \mathcal{L}_{2}(
p^{\prime})  \right]  
+\left(  1-\lambda\right)  E\left[  \mathcal{L}_{1}\left(  p\right)
\mathcal{L}_{1}(  p^{\prime})  \right],
\end{align*}
where for $j,j^{\prime}\in\left\{  1,2\right\}  $ and $p,p^{\prime}\in\left[
0,1\right]  $,
\begin{align*}
 &E\left[  \mathcal{L}_{j}\left(  p\right)  \mathcal{L}_{j^{\prime}}(
p^{\prime})  \right]= \\
& Cov\left(  \frac{1}{\mu_{j}}\left\{  L_{j}\left(  p\right)  X_i^{j}%
-Q_{j}\left(  p\right)  \wedge X_i^{j}\right\}  ,\frac{1}{\mu_{j^{\prime}}%
}\left\{  L_{j'}(  p^{\prime})  X_i^{j^{\prime}}-Q_{j'}(
p^{\prime})  \wedge X_i^{j^{\prime}}\right\}  \right)  .
\end{align*}
\end{lemma}

For $m\geq2$ and $w\in\{u,d\}$, define $\mathcal{I}_{m}^{w}:\ell^{\infty}\left(  \left[  0,1\right]  \right)
\rightarrow\ell^{\infty}\left(  \left[  0,1\right]  \right)  $ such that
\begin{align*}
&\mathcal{I}_{m}^{u}\left(  f\right)  \left(  p\right)  =\int_{0}^{p}\cdots
\int_{0}^{t_{3}}\int_{0}^{t_{2}}f\left(  t_{1}\right)  \mathrm{d}t_{1}\mathrm{d}t_{2}\cdots\mathrm{d}
t_{m-1} \text{ and } \\ 
&\mathcal{I}_{m}^{d}\left(  f\right)  \left(  p\right)  =\int_{p}^{1}\cdots
\int_{t_{3}}^{1}\int_{t_{2}}^{1}f\left(  t_{1}\right)  \mathrm{d}t_{1}\mathrm{d}t_{2}\cdots\mathrm{d}%
t_{m-1},\quad p\in\left[  0,1\right]  .   
\end{align*}
To unify notation, we let $\mathcal{I}_{1}^{w}:\ell^{\infty}\left(  \left[  0,1\right]  \right)
\rightarrow\ell^{\infty}\left(  \left[  0,1\right]  \right)  $ such that $\mathcal{I}_{1}^{w}(f)=f$ for $w\in\{u,d\}$. Then we obtain the following lemma.
\begin{lemma}\label{lemma.phi weak convergence}
For $m\ge1$ and $w\in\{u,d\}$, under Assumptions \ref{ass.distribution} and \ref{ass.data}, it follows that
\begin{align}\label{eq.phi_m^u weak convergence}
\sqrt{T_n}(\hat{\phi}_m^w-\phi_m^w)=\mathcal{I}_{m}^{w}\left(\sqrt{T_n}(\hat{\phi}-\phi)\right)\leadsto \mathcal{I}_{m}^{w}(\mathbb{G})    
\end{align}
with 
\begin{align}\label{eq.Var upward}
    &Var(\mathcal{I}_{m}^{u}(\mathbb{G})(p))=\notag\\
    &\int_{0}^{p}\cdots\int_{0}^{t_{3}^{\prime}}\int_{0}^{t_{2}^{\prime}}\left(  \int_{0}^{p}\cdots\int_{0}^{t_{3}}\int_{0}^{t_{2}}E\left[
\mathbb{G}(  t_{1})  \mathbb{G}(  t_{1}^{\prime
})  \right]  \mathrm{d}t_{1}\mathrm{d}t_{2}\cdots\mathrm{d}t_{m-1}\right)  \mathrm{d}t_{1}^{\prime}\mathrm{d}t_{2}^{\prime}\cdots\mathrm{d}t_{m-1}^{\prime}
\end{align}
and 
\begin{align}\label{eq.Var downward}
    &Var(\mathcal{I}_{m}^{d}(\mathbb{G})(p))=\notag\\
    &\int_{p}^{1}\cdots\int_{t_{3}^{\prime}}^1\int_{t_{2}^{\prime}}^1\left(  \int_{p}^{1}\cdots\int_{t_{3}}^1\int_{t_{2}}^1E\left[
\mathbb{G}(  t_{1})  \mathbb{G}(  t_{1}^{\prime
})  \right]  \mathrm{d}t_{1}\mathrm{d}t_{2}\cdots\mathrm{d}t_{m-1}\right)  \mathrm{d}t_{1}^{\prime}\mathrm{d}t_{2}^{\prime}\cdots\mathrm{d}t_{m-1}^{\prime}
\end{align}
for every $p\in[0,1]$.
\end{lemma}

We then introduce the following proposition for the asymptotic properties of the estimator $\hat{c}_m^w(L_1,L_2)$. Recall that for $w\in\{u,d\}$, $\phi_1^w=L_2-L_1$; for all $m\ge2$, $\phi_m^u=L_2^m-L_1^m$ and $\phi_m^d=\tilde L_1^m-\tilde L_2^m$. 
\begin{proposition}\label{prop.c asymptotic limit}
Suppose that Assumptions \ref{ass.distribution} and \ref{ass.data} hold. For $m\ge 1$ and $w\in\{u,d\}$, if $\phi_m^w(p)\neq 0$  for some $p\in[0,1]$, then $\hat{c}_m^w(L_1,L_2)\to c_m^w(L_1,L_2)$ almost surely (a.s.), and
\begin{align}\label{eq.c asymptotic limit}
    \sqrt{T_n}(\hat{c}_m^w(L_1,L_2)-{c}_m^w(L_1,L_2))=\sqrt{T_n}(\mathcal{F}(\hat{\phi}_m^w)-\mathcal{F}(\phi_m^w))\leadsto \mathcal{F}'_{\phi_m^w}(\mathcal{I}_m^w(\mathbb{G})),
\end{align}
where for every $h\in\ell^{\infty}([0,1])$,
\begin{align*}
\mathcal{F}_{\phi_m^w}^{\prime}\left(  h\right)  =\frac{\mathcal{F}_{1\phi_m^w}^{\prime}\left(  h\right)  \mathcal{F}_{2}\left(
\phi_m^w\right)  -\mathcal{F}_{1}\left(  \phi_m^w\right)  \mathcal{F}_{2\phi_m^w}^{\prime
}\left(  h\right)  }{\left(  \mathcal{F}_{1}\left(  \phi_m^w\right)
+\mathcal{F}_{2}\left(  \phi_m^w\right)  \right)  ^{2}}.
\end{align*}
\end{proposition}

Proposition \ref{prop.c asymptotic limit} provides the asymptotic distribution of $\hat{c}_m^w(L_1,L_2)$ when $\phi_m^w(p)\neq0$ for some $p\in[0,1]$. We next construct a bootstrap confidence interval for $c_m^w(L_1,L_2)$ based on this asymptotic distribution. 

\subsubsection{Bootstrap Confidence Interval}\label{sec.bootstrap CI Lorenz}

Since the distribution of $\mathcal{F}'_{\phi_m^w}(\mathcal{I}_m^w(\mathbb{G}))$ in \eqref{eq.c asymptotic limit} is unknown and depends on the underlying data generating processes (DGPs), we use the bootstrap method to approximate this distribution. Because $\mathcal{F}'_{\phi_m^w}$ is nonlinear, the standard bootstrap method may not consistently approximate the distribution of $\mathcal{F}'_{\phi_m^w}(\mathcal{I}_m^w(\mathbb{G}))$ \citep{andrews2000inconsistency,bickel2012resampling,fang2014inference}. We next employ the bootstrap method of \citet{fang2014inference} to obtain a consistent approximation of the asymptotic distribution and construct valid critical values.
By Lemma \ref{lemma.G kernel}, for $j,j^{\prime}\in\left\{  1,2\right\}  $ and $p,p^{\prime}\in\left[
0,1\right]  $, the estimator for $E\left[  \mathcal{L}_{j}\left(  p\right)
\mathcal{L}_{j^{\prime}}\left(  p^{\prime}\right)  \right]  $, denoted by $\hat{E}\left[  \mathcal{L}_{j}\left(  p\right)
\mathcal{L}_{j^{\prime}}\left(  p^{\prime}\right)  \right]  $, is defined as the sample covariance of the two samples
\[
\left\{  \frac{1}{\hat{\mu}_{j}}\left(  \hat{L}_{j}(p)X_{i}^{j}-\hat{Q}%
_{j}(p)\wedge X_{i}^{j}\right)  \right\}  _{i=1}^{n_{j}}%
\text{ and }
\left\{  \frac{1}{\hat{\mu}_{j^{\prime}}}\left(  \hat{L}_{j^{\prime}%
}(p^{\prime})X_{i}^{j^{\prime}}-\hat{Q}_{j^{\prime}}(p^{\prime})\wedge
X_{i}^{j^{\prime}}\right)  \right\}  _{i=1}^{n_{j^{\prime}}}.
\]
Let $\hat{\lambda}=n_{1}/(n_{1}+n_{2})$. By Lemma \ref{lemma.G kernel}, for independent samples, we estimate
$E\left[  \mathbb{G}\left(  p\right)  \mathbb{G}\left(  p^{\prime}\right)
\right]  $ by
\begin{align}\label{eq.estimated cov independent}
\hat{E}\left[  \mathbb{G}\left(  p\right)  \mathbb{G}(p^{\prime})\right]
=(  1-\hat{\lambda})  \hat{E}\left[  \mathcal{L}_{1}\left(
p\right)  \mathcal{L}_{1}(p^{\prime})\right]  +\hat{\lambda}\hat{E}\left[
\mathcal{L}_{2}\left(  p\right)  \mathcal{L}_{2}(p^{\prime})\right] ,\end{align}
and for matched pairs, we estimate
$E\left[  \mathbb{G}\left(  p\right)  \mathbb{G}\left(  p^{\prime}\right)
\right]  $ by
\begin{align}\label{eq.estimated cov matched pairs}
\hat{E}\left[  \mathbb{G}\left(  p\right)  \mathbb{G}(p^{\prime})\right]   
=&\,(  1-\hat{\lambda})  \hat{E}\left[  \mathcal{L}_{1}\left(
p\right)  \mathcal{L}_{1}(p^{\prime})\right]  -\sqrt{\hat{\lambda}(
1-\hat{\lambda})  }\hat{E}\left[  \mathcal{L}_{1}\left(  p\right)
\mathcal{L}_{2}(p^{\prime})\right] \notag \\
&  -\sqrt{\hat{\lambda}(  1-\hat{\lambda})  }\hat{E}\left[
\mathcal{L}_{2}\left(  p\right)  \mathcal{L}_{1}(p^{\prime})\right]
+\hat{\lambda}\hat{E}\left[  \mathcal{L}_{2}\left(  p\right)  \mathcal{L}%
_{2}(p^{\prime})\right]  .
\end{align}
We let  
$\hat\sigma(p)^2=\hat{E}[  \mathbb{G}\left(  p\right) ^2]$ for all $p\in[0,1]$, and for $w\in\{u,d\}$, we let $\hat{\sigma}_{1}^{w}=\hat{\sigma}$.
Based on \eqref{eq.Var upward} and \eqref{eq.Var downward} with the estimates in \eqref{eq.estimated cov independent} and \eqref{eq.estimated cov matched pairs}, for $m\ge2$, we estimate ${Var}(\mathcal{I}_{m}^{u}(\mathbb{G})(p))$ and ${Var}(\mathcal{I}_{m}^{d}(\mathbb{G})(p))$ by
\begin{align}
    \hat{\sigma}_{m}^{u}(p)^2=\int_{0}^{p}\cdots\int_{0}^{t_{3}^{\prime}}\int_{0}^{t_{2}^{\prime}}\left(  \int_{0}^{p}\cdots\int_{0}^{t_{3}}\int_{0}^{t_{2}}\hat{E}\left[
\mathbb{G}(  t_{1})  \mathbb{G}(  t_{1}^{\prime
})  \right]  \mathrm{d}t_{1}\mathrm{d}t_{2}\cdots\mathrm{d}t_{m-1}\right)  \mathrm{d}t_{1}^{\prime}\mathrm{d}t_{2}^{\prime}\cdots\mathrm{d}t_{m-1}^{\prime}
\end{align}
and 
\begin{align}
    \hat{\sigma}_m^{d}(p)^2=\int_{p}^{1}\cdots\int_{t_{3}^{\prime}}^1\int_{t_{2}^{\prime}}^1\left(  \int_{p}^{1}\cdots\int_{t_{3}}^1\int_{t_{2}}^1\hat{E}\left[
\mathbb{G}(  t_{1})  \mathbb{G}(  t_{1}^{\prime
})  \right]  \mathrm{d}t_{1}\mathrm{d}t_{2}\cdots\mathrm{d}t_{m-1}\right)  \mathrm{d}t_{1}^{\prime}\mathrm{d}t_{2}^{\prime}\cdots\mathrm{d}t_{m-1}^{\prime},
\end{align}
respectively. 
Following \citet{Beare2017improved}, we construct the estimators of $B_{+}\left(  \phi_m^w\right)  $, $B_{+}\left(-\phi_m^w\right)  $, and $B_{0}\left(  \phi_m^w\right)  $ by
\begin{align*}
&\widehat{B_{+}\left(  \phi_m^w\right)  } =\left\{  p\in\left[  0,1\right]
:\frac{\sqrt{T_{n}}\hat{\phi}_m^w\left(  p\right)  }{\xi_{0}\vee\hat{\sigma
}_m^w\left(  p\right)  }>t_{n}\right\}  , \widehat{B_{+}\left(  -\phi_m^w\right)
}=\left\{  p\in\left[  0,1\right]  :\frac{\sqrt{T_{n}}\hat{\phi}_m^w\left(
p\right)  }{\xi_{0}\vee\hat{\sigma}_m^w\left(  p\right)  }<-t_{n}\right\}  ,\\
&\text{and }\widehat{B_{0}\left(  \phi_m^w\right)  } =\left\{  p\in\left[
0,1\right]  :\left\vert \frac{\sqrt{T_{n}}\hat{\phi}_m^w\left(  p\right)  }%
{\xi_{0}\vee\hat{\sigma}_m^w\left(  p\right)  }\right\vert \leq t_{n}\right\}  ,
\end{align*}
where $t_n$ is some tuning parameter such that $t_{n}\rightarrow\infty$ and $t_{n}/\sqrt{T_{n}}\rightarrow0$ as
$n\rightarrow\infty$ and $\xi_0$ is a small trimming parameter that bounds $\hat{\sigma}_m^w$ away from zero. In our simulations and application, $\xi_0$ is set to $0.001$. We then construct the estimators of $\mathcal{F}_{1\phi_m^w
}^{\prime}$ and $\mathcal{F}_{2\phi_m^w}^{\prime}$ by%
\begin{align*}
&\mathcal{\hat{F}}_{1\phi_m^w}^{\prime}\left(  h\right)  =\int_{\widehat
{B_{+}\left(  \phi_m^w\right)  }}h\left(  p\right)  \mathrm{d}p+\int
_{\widehat{B_{0}\left(  \phi_m^w\right)  }}\max\left\{  h\left(  p\right)
,0\right\}  \mathrm{d}p\\
&\text{and } \mathcal{\hat{F}}_{2\phi_m^w}^{\prime}\left(  h\right)
=\int_{\widehat{B_{+}\left(  -\phi_m^w\right)  }}-h\left(  p\right)
\mathrm{d}p+\int_{\widehat{B_{0}\left(  \phi_m^w\right)  }}\max\left\{  -h\left(
p\right)  ,0\right\}  \mathrm{d}p
\end{align*}
for every $h$ $\in\ell^{\infty}\left(  \left[  0,1\right]  \right)  $. The
estimator of $\mathcal{F}_{\phi_m^w}^{\prime}$ is defined by
\[
\mathcal{\hat{F}}_{\phi_m^w}^{\prime}\left(  h\right)  =\frac{\mathcal{\hat{F}%
}_{1\phi_m^w}^{\prime}\left(  h\right)  \mathcal{F}_{2}(  \hat{\phi}_m^w)
-\mathcal{F}_{1}(  \hat{\phi}_m^w)  \mathcal{\hat{F}}_{2\phi_m^w}^{\prime
}\left(  h\right)  }{\left(  \mathcal{F}_{1}(  \hat{\phi}_m^w)
+\mathcal{F}_{2}(  \hat{\phi}_m^w)  \right)  ^{2}}
\]
for every $h$ $\in\ell^{\infty}\left(  \left[  0,1\right]  \right)  $.

We follow \citet{BDB14} and \citet{Beare2017improved} to draw bootstrap samples:
For independent samples, we draw a bootstrap sample $\{\hat{X}_{i}^{j}%
\}_{i=1}^{n_j}$ identically and independently from $\{X_{i}^{j}\}_{i=1}^{n_j}$ for
$j=1,2$, where $\{\hat{X}_{i}^{1}\}_{i=1}^{n_1}$ is jointly independent of
$\{\hat{X}_{i}^{2}\}_{i=1}^{n_2}$; for matched pairs, we draw a bootstrap sample
$\{(\hat{X}_{i}^{1},\hat{X}_{i}^{2})\}_{i=1}^{n}$ identically and independently from $\{(X_{i}^{1},X_{i}^{2})\}_{i=1}^{n}$. Then following \citet{BDB14} and \citet{Beare2017improved}, for $j=1,2$, we define the bootstrap empirical CDF
\begin{align*}
	\hat{F}_j^*(x)=\frac{1}{n_j}\sum_{i=1}^{n_j}
	{1}(\hat{X}_i^j\leq x),\quad x\in[0,\infty),
\end{align*}
the bootstrap empirical quantile function
\begin{align}\label{eq.bootstrap quantile}
	\hat{Q}^*_j(p)=\inf\{x\in[0,\infty):\hat{F}^*_j(x)\geq p\},\quad p\in[0,1],
\end{align}
and the bootstrap empirical Lorenz curve
\begin{align*}
	\hat{L}^*_j(p)=\frac{1}{\hat{\mu}^*_j}\int_0^p\hat{Q}^*_j(t)\mathrm{d}t,\quad p\in[0,1],
\end{align*}
where $\hat{\mu}^*_j$ is the sample mean of $\{\hat{X}_i^j\}_{i=1}^{n_j}$. 
Let $\hat{L}_j^{1*}=\hat{L}_j^*$.
For $m\ge 2$, define
\[
\hat L_{j}^{m*}\left(  p\right)  =\int_{0}^{p}\hat L_j^{m-1*}\left(  t\right)
\mathrm{d}t,\quad p\in\left[  0,1\right]  .
\]
Define
\[
\hat{\tilde{L}}_{j}^{2*}\left(  p\right)  =\int_{p}^{1}\left(  1-\hat L_{j}^*\left(
t\right)  \right)  \mathrm{d}t,\quad p\in\left[  0,1\right],
\]
and for $m\geq3$,
\[
\hat{\tilde{L}}_{j}^{m*}\left(  p\right)  =\int_{p}^{1}\hat{\tilde{L}}^{m-1*}\left(
t\right)  \mathrm{d}t,\quad p\in\left[  0,1\right]  .
\]
The bootstrap approximation of $\phi$ is defined by
\begin{align*}
	\hat{\phi}^*(p)=\hat{L}^*_2(p)-\hat{L}^*_1(p),\quad p\in[0,1].
\end{align*}
For $w\in\{u,d\}$, let $\hat{\phi}_1^{w*}=\hat{\phi}^*$. For $m\ge 2$, define
\[
\hat{\phi}_m^{u*}\left(  p\right)  =\hat L_{2}^{m*}\left(  p\right) -\hat L_{1}^{m*}\left(  p\right)\text{ and }\hat{\phi}_m^{d*}\left(  p\right)  =\hat{\tilde{L}}_{1}^{m*}\left(  p\right) -\hat{\tilde{L}}_{2}^{m*}\left(  p\right),\quad p\in[0,1].
\]
We define the bootstrap approximation of $c_m^w(L_1,L_2)$ by 
\begin{align}
    \hat{c}_m^{w*}(L_1,L_2)=\mathcal{\hat{F}}_{\phi_m^{w}}^{\prime}(\sqrt{T_n}(\hat{\phi}_m^{w*}-\hat{\phi}_m^w)).
\end{align}
For every $\beta\in(0,1)$, let $c_{m,\beta}^w$ denote the $\beta$ quantile of the distribution of $\mathcal{F}'_{\phi_m^w}(\mathcal{I}_m^w(\mathbb{G}))$.
We construct the bootstrap approximation of $c_{m,\beta}^w$ by 
\begin{align}
    \hat{c}_{m,\beta}^w=\inf\left\{c: \mathbb{P}\left(\hat{c}_m^{w*}(L_1,L_2)\le c|\{X^1_i\}_{i=1}^{n_1},\{X^2_i\}_{i=1}^{n_2}\right)\ge \beta\right\}.
\end{align}
In practice, we approximate $\hat{c}_{m,\beta}^w$ by computing the $\beta$ quantile of the $n_B$ independently generated $\hat{c}_m^{w*}(L_1,L_2)$, where $n_B$ is chosen as large as is computationally convenient. 

For a nominal significance level $\alpha\in(0,1/2)$, we construct the $1-\alpha$ bootstrap confidence interval of $c_m^w(L_1,L_2)$ by
\begin{align}
    \mathrm{CI}_{m,1-\alpha}^w=[\hat{c}_m^w(L_1,L_2)-T_n^{-1/2}\hat{c}_{m,1-\alpha/2}^w,\hat{c}_m^w(L_1,L_2)-T_n^{-1/2}\hat{c}_{m,\alpha/2}^w].
\end{align}

\begin{proposition}\label{prop.confidence interval}
Suppose that Assumptions \ref{ass.distribution} and \ref{ass.data} hold. For $m\ge1$ and $w\in\{u,d\}$, if $c_m^w(L_1,L_2)\in(0,1)$ and the CDF of $\mathcal{F}'_{\phi_m^w}(\mathcal{I}_m^w(\mathbb{G}))$ is continuous and increasing at $c_{m,\alpha/2}^w$ and $c_{m,1-\alpha/2}^w$, then it follows that
\begin{align}
    \lim_{n\rightarrow\infty}\mathbb{P}(c_m^w(L_1,L_2)\in\mathrm{CI}_{m,1-\alpha}^w)= 1-\alpha.
\end{align}

\end{proposition}

Proposition \ref{prop.confidence interval} provides the two-sided confidence interval for $c_m^w(L_1,L_2)$. It is straightforward to extend the result to one-sided confidence intervals. 

\begin{remark}
In Proposition \ref{prop.confidence interval}, it is required that $c_m^w(L_1,L_2)\in(0,1)$, which excludes some cases where Lorenz dominance holds. In these cases, $\mathcal{F}'_{\phi_m^w}(\mathcal{I}_m^w(\mathbb{G}))=0$ or $c_m^w(L_1,L_2)$ may incur the zero denominator issue (in the case where $F_1=F_2$). As a consequence, the asymptotic results may not hold. In practice, we suggest first performing tests for Lorenz dominance and proceeding with the proposed method if the dominance relationship is rejected. 
\end{remark}

\begin{remark}
The asymptotic results in the paper are established under a fixed DGP, which may not hold uniformly over all possible DGPs in general. For example, if $c_m^w(L_1,L_2)$ drifts to zero as $T_n\to\infty$, the results in Proposition \ref{prop.confidence interval} may not hold. Then the proposed method may not work well in finite samples. In practice, when $c_m^w(L_1,L_2)$ is close to $0$ or $1$ and the sample size is relatively small, we may use the bounds of $c_m^w(L_1,L_2)$ to improve the confidence interval, that is, we may use $\mathrm{CI}_{m,1-\alpha}^w\cap [0,1]$ as the confidence interval.
\end{remark}

\section{Simulation Evidence}
\label{sec.simulation}
We run Monte Carlo simulations to provide evidence for the finite sample properties of the estimators and the bootstrap confidence intervals. 
In this section, we focus on the Lorenz dominance coefficient. Additional simulations are provided in Section \ref{sec.additional simulation} in the appendix. 
The number of bootstrap samples, $n_B$, is $1000$. The number of Monte Carlo iterations is $1000$. The nominal significance level is set to $0.05$. We present the mean (Mean), the bias (Bias), the standard error (SE), and the root mean square error (RMSE) of the estimators, as well as the coverage rate (CR) of the bootstrap confidence intervals in $1000$ iterations.

As discussed in \citet{Beare2017improved}, \citet{reed2001pareto,reed2003pareto} and \citet{toda2012double} show that income distributions can be well approximated by members of the double Pareto parametric family. We use distributions from this family to construct the simulations. The density function for the double Pareto family is 
\[
f\left(  x\right)  =\left\{
\begin{array}
[c]{c}%
\frac{\alpha\beta}{\alpha+\beta}M^{\alpha}x^{-\alpha-1}\\
\frac{\alpha\beta}{\alpha+\beta}M^{-\beta}x^{\beta-1}%
\end{array}
\right.
\begin{array}
[c]{c}%
x\geq M,\\
0\leq x<M,
\end{array}
\]
where $M>0$ is a scale parameter. For $\alpha,\beta>0$, we write $X\sim \mathrm{dP}(\alpha,\beta)$ to denote that the random variable $X$ has the double Pareto distribution with $M$ normalized to one and shape parameters $\alpha,\beta$. Assumption \ref{ass.distribution} is satisfied when $\alpha>2$. We let $X^1\sim \mathrm{dP}(3,1.5)$ and $X^2_{(\beta)}\sim \mathrm{dP}(2.1,\beta)$ with $\beta\in\{2,3,4,5\}$. Figure \ref{fig:Lorenz} displays the Lorenz curves $L_1$ and $L_{2(\beta)}$ corresponding to the above DGPs. The LDC $c(L_1,L_2)=$ $0.04703$, $0.31489$, $0.45198$, and $0.51960$, respectively, for $\beta=$ $2$, $3$, $4$, and $5$. We choose the tuning parameter from $t_n\in (0,20]$. In Section \ref{sec.tuning parameter selection}, we show how to choose $t_n$ empirically to construct the bootstrap confidence intervals. For independent samples, we let $$(n_1,n_2)\in\{(200,200),(200,500),(200,1000),(1000,2000),(10000,10000)\}.$$
For matched pairs, we let $$(n_1,n_2)\in\{(200,200),(500,500),(1000,1000),(2000,2000),(10000,10000)\}.$$

Tables \ref{tab:CR IS} and \ref{tab:CR MP} show the simulation results for independent samples and matched pairs. For all the DGPs, as $n_1$ and $n_2$ increase, Mean gets close to $c(L_1,L_2)$, Bias decreases to $0$, and both SE and RMSE decrease; under appropriate choices of $t_n$, CR approaches $95\%$.  

\begin{figure} [ht!]
\caption{Lorenz Curves for Four DGPs}
\label{fig:Lorenz}
\centering
\begin{subfigure}[b]{0.45\textwidth}
	\centering
\scalebox{0.2}{
\includegraphics{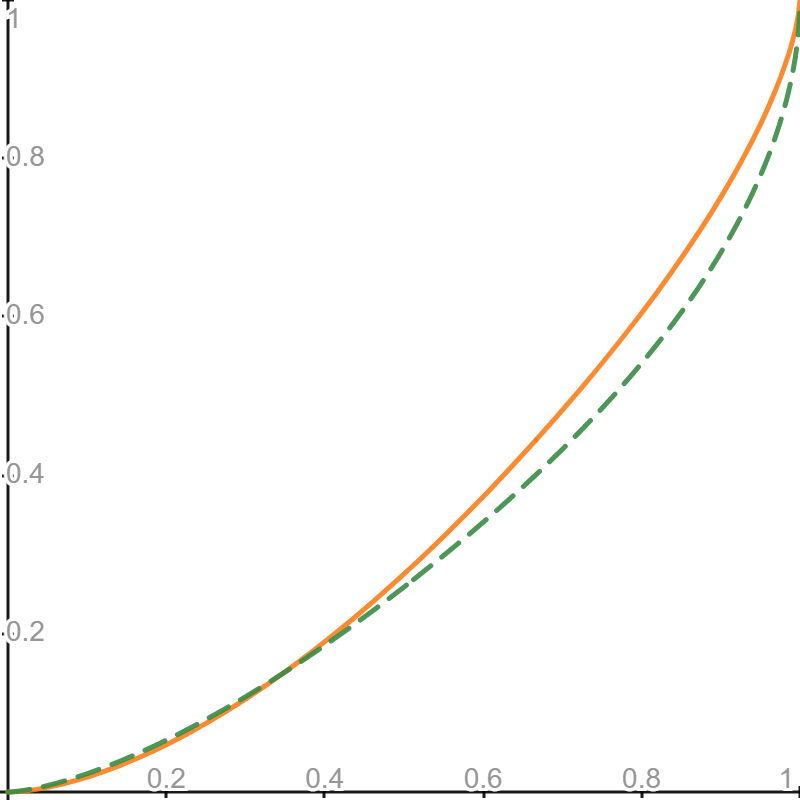}
}
\subcaption{$L_1$ (solid) and $L_{2(2)}$ (dashed)}

\end{subfigure}
\begin{subfigure}[b]{0.45\textwidth}
	\centering
\scalebox{0.2}{
\includegraphics{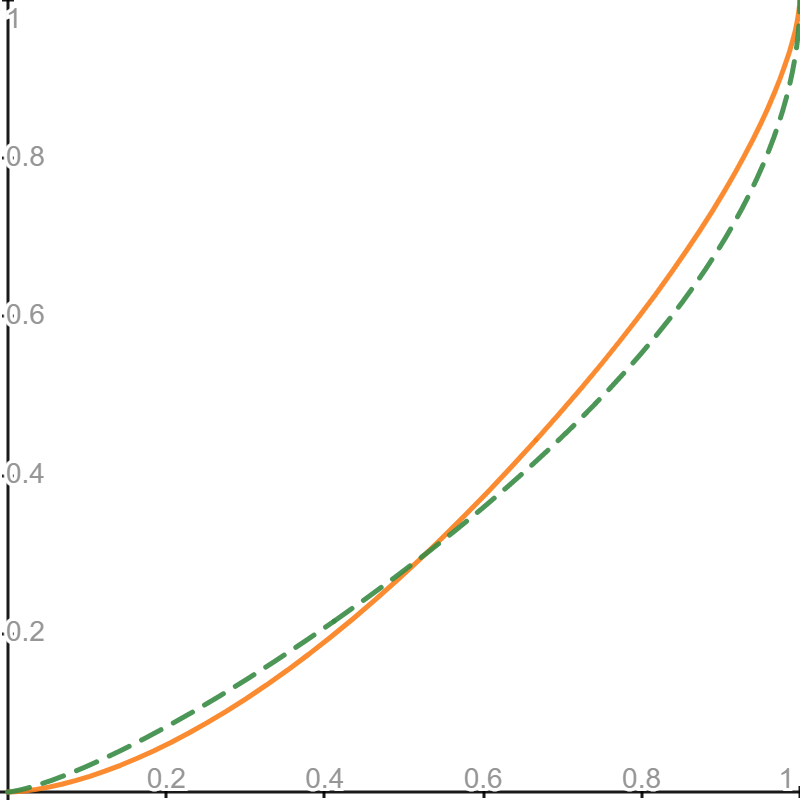}
}
\subcaption{$L_1$ (solid) and $L_{2(3)}$ (dashed)}

\end{subfigure}
\begin{subfigure}[b]{0.45\textwidth}
	\centering
\scalebox{0.2}{
\includegraphics{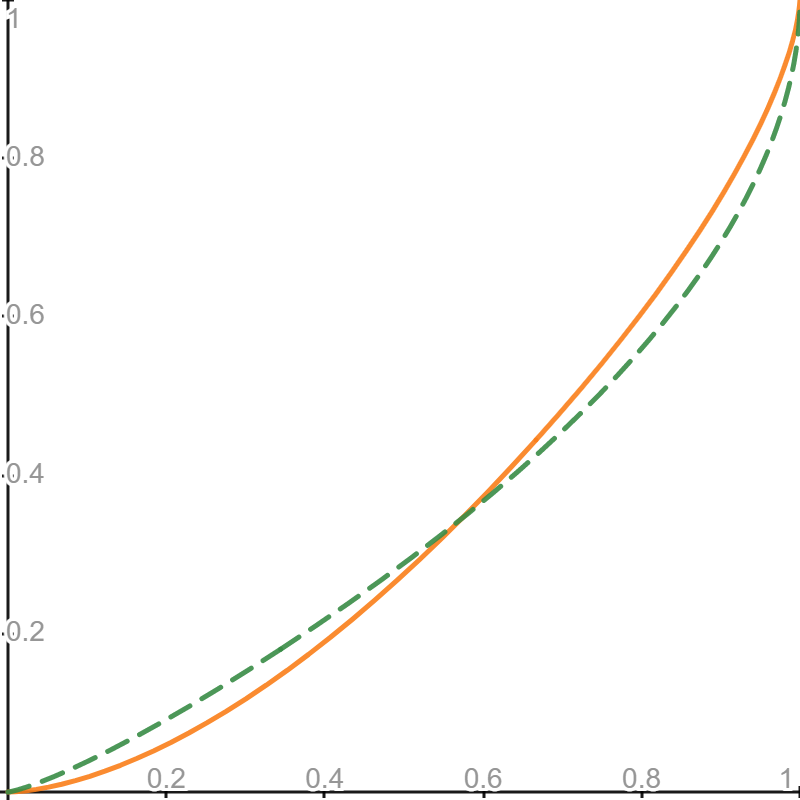}
}
\subcaption{$L_1$ (solid) and $L_{2(4)}$ (dashed)}

\end{subfigure}
\begin{subfigure}[b]{0.45\textwidth}
	\centering
\scalebox{0.2}{
\includegraphics{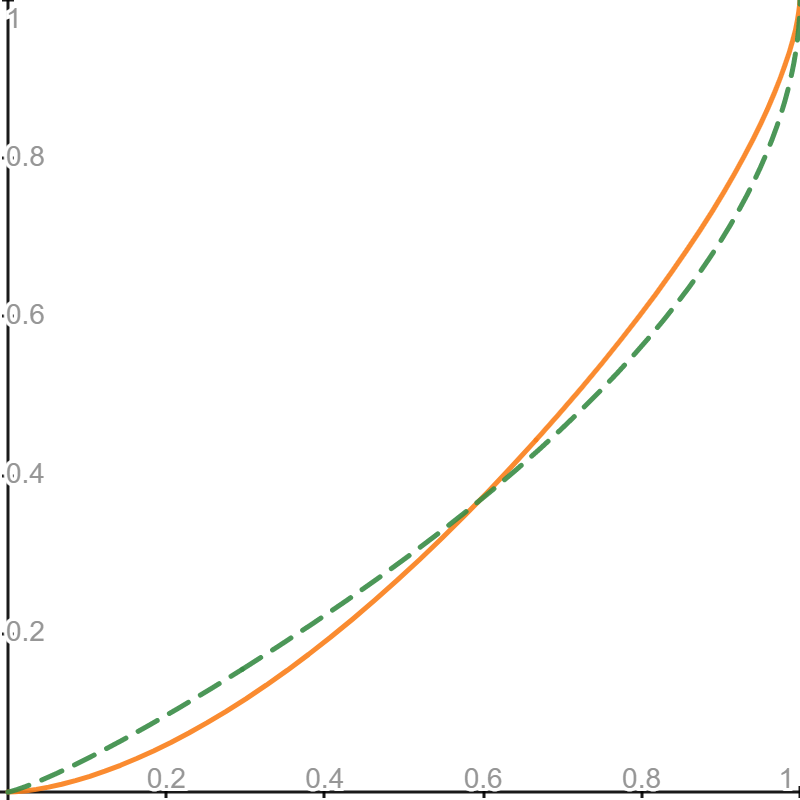}
}
\subcaption{$L_1$ (solid) and $L_{2(5)}$ (dashed)}

\end{subfigure}
\end{figure}

\begin{table}[ht!]
\centering
	\caption{Estimations and Coverage Rates for LDC (Independent Samples)}
		
		\scalebox{0.95}{
			\begin{tabular}{cccccccccc}
				\hline\hline
			{DGP} & 
			 $c(L_1,L_2)$	&	{$n_1 $}	& {$n_2 $}		 & Mean & Bias & SE & RMSE & $t_n$ & CR  \\
			 \hline
			 \multirow{5}{*}{(a)}	 & \multirow{5}{*}{0.04703} 
		               & 200 & 200   & 0.2305 & 0.1834 & 0.3088 & 0.3592 & 0.001 & 0.7040 \\
			 &         & 200 & 500   & 0.1529 & 0.1058 & 0.2420 & 0.2641 & 0.001 & 0.7270 \\
			 &         & 200 & 1000  & 0.1258 & 0.0788 & 0.2144 & 0.2284 & 0.001 & 0.7200 \\
			 &         & 1000 & 2000 & 0.0733 & 0.0263 & 0.0754 & 0.0798 & 0.001 & 0.8630 \\
			 &         & 10000& 10000& 0.0523 & 0.0053 & 0.0233 & 0.0239 & 0.001 & 0.9330\\	 
			 \hline
			 \multirow{5}{*}{(b)}	 & \multirow{5}{*}{0.31489} 
		               & 200 & 200   & 0.4753 & 0.1604 & 0.3341 & 0.3706 & 0.001 & 0.6970 \\
			 &         & 200 & 500   & 0.3986 & 0.0837 & 0.2865 & 0.2985 & 0.001 & 0.7840 \\
			 &         & 200 & 1000  & 0.3695 & 0.0546 & 0.2650 & 0.2706 & 0.001 & 0.7900 \\
			 &         & 1000 & 2000 & 0.3515 & 0.0366 & 0.1612 & 0.1653 & 0.001 & 0.9140 \\
			 &         & 10000& 10000& 0.3260 & 0.0111 & 0.0718 & 0.0727 & 0.101 & 0.9340\\	 	 
			 \hline
			 \multirow{5}{*}{(c)}	 & \multirow{5}{*}{0.45198} 
		               & 200 & 200   & 0.5680 & 0.1160 & 0.3182 & 0.3387 & 0.001 & 0.6720 \\
			 &         & 200 & 500   & 0.5056 & 0.0537 & 0.2770 & 0.2821 & 0.001 & 0.7800 \\
			 &         & 200 & 1000  & 0.4811 & 0.0292 & 0.2574 & 0.2590 & 0.001 & 0.7940 \\
			 &         & 1000 & 2000 & 0.4804 & 0.0284 & 0.1678 & 0.1702 & 0.001 & 0.9100 \\
			 &         & 10000& 10000& 0.4619 & 0.0099 & 0.0799 & 0.0805 & 0.001 & 0.9260\\	 
			 \hline
			 \multirow{5}{*}{(d)}	 & \multirow{5}{*}{0.51960} 
		               & 200 & 200   & 0.6114 & 0.0918 & 0.3060 & 0.3195 & 0.001 & 0.6690 \\
			 &         & 200 & 500   & 0.5575 & 0.0379 & 0.2674 & 0.2701 & 0.001 & 0.7740 \\
			 &         & 200 & 1000  & 0.5360 & 0.0164 & 0.2484 & 0.2489 & 0.001 & 0.7850 \\
			 &         & 1000 & 2000 & 0.5427 & 0.0231 & 0.1655 & 0.1671 & 0.001 & 0.9030 \\
			 &         & 10000& 10000& 0.5283 & 0.0087 & 0.0806 & 0.0811 & 0.001 & 0.9280\\	 
				\hline\hline                         	
			\end{tabular}
		}

		\label{tab:CR IS}

\end{table}

\begin{table}[ht!]
\centering
	\caption{Estimations and Coverage Rates for LDC (Matched Pairs)}
		
		\scalebox{0.95}{
			\begin{tabular}{cccccccccc}
				\hline\hline
			{DGP} & 
			 $c(L_1,L_2)$	&	{$n_1 $}	& {$n_2 $}		 & Mean & Bias & SE & RMSE & $t_n$ & CR  \\
			 \hline
			 \multirow{5}{*}{(a)}	 & \multirow{5}{*}{0.04703} 
		               & 200 & 200   & 0.1733 & 0.1262 & 0.2616 & 0.2904 & 0.001 & 0.7330 \\
			 &         & 500 & 500   & 0.1173 & 0.0703 & 0.1783 & 0.1917 & 0.001 & 0.8280 \\
			 &         & 1000 & 1000 & 0.0914 & 0.0444 & 0.1004 & 0.1097 & 0.001 & 0.9010 \\
			 &         & 2000 & 2000 & 0.0692 & 0.0222 & 0.0585 & 0.0626 & 0.001 & 0.9070 \\
			 &         & 10000& 10000& 0.0491 & 0.0021 & 0.0179 & 0.0181 & 0.501   & 0.9480\\	 
			 \hline
			 \multirow{5}{*}{(b)}	 & \multirow{5}{*}{0.31489} 
		               & 200 & 200   & 0.4158 & 0.1009 & 0.3068 & 0.3230 & 0.001 & 0.7620 \\
			 &         & 500 & 500   & 0.4162 & 0.1013 & 0.2334 & 0.2544 & 0.001 & 0.8770 \\
			 &         & 1000 & 1000 & 0.3907 & 0.0758 & 0.1808 & 0.1961 & 0.001 & 0.9270 \\
			 &         & 2000 & 2000 & 0.3568 & 0.0419 & 0.1434 & 0.1494 & 0.001 & 0.9250 \\
			 &         & 10000& 10000& 0.3196 & 0.0047 & 0.0602 & 0.0604 & 7.7   & 0.9500\\	 
			 \hline
			 \multirow{5}{*}{(c)}	 & \multirow{5}{*}{0.45198} 
		               & 200 & 200   & 0.5185 & 0.0665 & 0.2986 & 0.3059 & 0.001 & 0.7560 \\
			 &         & 500 & 500   & 0.5473 & 0.0953 & 0.2249 & 0.2443 & 0.001 & 0.8590 \\
			 &         & 1000 & 1000 & 0.5168 & 0.0648 & 0.1811 & 0.1923 & 0.001 & 0.9060 \\
			 &         & 2000 & 2000 & 0.4903 & 0.0383 & 0.1546 & 0.1593 & 0.001 & 0.9170 \\
			 &         & 10000& 10000& 0.4561 & 0.0041 & 0.0686 & 0.0687 & 4.4   & 0.9500\\
			 \hline
			 \multirow{5}{*}{(d)}	 & \multirow{5}{*}{0.51960} 
		               & 200 & 200   & 0.5682 & 0.0486 & 0.2898 & 0.2938 & 0.001 & 0.7520 \\
			 &         & 500 & 500   & 0.6094 & 0.0898 & 0.2144 & 0.2325 & 0.001 & 0.8510 \\
			 &         & 1000 & 1000 & 0.5752 & 0.0556 & 0.1760 & 0.1845 & 0.001 & 0.9000 \\
			 &         & 2000 & 2000 & 0.5547 & 0.0351 & 0.1549 & 0.1589 & 0.001 & 0.9060 \\
			 &         & 10000& 10000& 0.5232 & 0.0036 & 0.0700 & 0.0701 & 3.5 & 0.9500 \\
				\hline\hline                         	
			\end{tabular}
		}

		\label{tab:CR MP}

\end{table}

\section{Empirical Application}\label{sec.empirical}

In this section, we revisit the example of \citet{aaberge2021ranking} regarding the inequality growth in the United Kingdom over the past few decades, which may be related to the business cycle \citep{blundell2010consumption}. 
We use the replication data of \citet{aaberge2021ranking}, which initially came from the European Community Household Panel (ECHP) for
1995–2001, and from the European Union Statistics on Income and Living Conditions (EU-SILC) for 2005–2010. 
As described by \citet{aaberge2021ranking}, the dataset is limited to households with a male aged 25–64, and focuses on the distribution of individual equivalent income, adjusted for inflation as well as household size and composition.

Tables 3 and 4 of \citet{aaberge2021ranking} report the successful rankings based on the 3rd-degree upward and downward inverse stochastic dominances. They show that the 3rd-degree upward inverse stochastic dominance does not provide a complete ranking for several years. We apply the proposed method to estimate the third degree upward inverse stochastic dominance coefficients (3UISDCs) and construct confidence intervals for these coefficients for specified years.\footnote{See theoretical and simulation results for upward and downward inverse stochastic dominance coefficients in Sections \ref{sec.ISDC} and \ref{sec.additional simulation} in the appendix.} We compare the latter years ($F_1$) with the former years ($F_2$) shown in Table \ref{tab:CI Application 1}, and find that all the estimated 3UISDCs are small. Furthermore, the bootstrap confidence intervals are narrow. These facts provide evidence that the income distributions of the latter years may almost upward inverse stochastically dominate the income distributions of the former years, which complements the results of \citet{aaberge2021ranking}.

\begin{table}[H]
\centering
	\caption{Inequality Growth in the United Kingdom}
		
		\scalebox{1}{
			\begin{tabular}{cccccc}
				\hline\hline
			Year & 
			 $3$UISDC	& CI 	 & Year	 & $3$UISDC & CI   \\
			\hline
94--98&3.2067e-08&[0, 3.0653e-07]&96--98&2.9089e-06&[0, 1.6521e-05]\\
94--99&1.287e-07&[0, 8.9855e-07]&96--99&5.9121e-06&[0, 3.0879e-05]\\
94--00&3.0482e-07&[0, 1.7727e-06]&96--00&3.8907e-06&[0, 1.6773e-05]\\
94--01&3.9715e-07&[0, 2.5549e-06]&96--01&1.2211e-05&[0, 5.0282e-05]\\
95--98&1.0184e-07&[0, 8.0492e-07]&97--98&0.0063743&[0, 0.056701]\\
95--99&3.6099e-07&[0, 2.5261e-06]&97--99&0.0091885&[0, 0.069839]\\
95--00&6.5108e-07&[0, 3.5034e-06]&97--00&8.1315e-05&[0, 0.0004528]\\
95--01&1.7317e-06&[0, 9.5748e-06]&97--01&0.00027027&[0, 0.0011134]\\
				\hline\hline                         	
			\end{tabular}
		}

		\label{tab:CI Application 1}

\end{table}

In Table \ref{tab:CI LDC}, we provide the estimated $F_1$-$F_2$ LDCs and their corresponding confidence intervals, using the same data as in Table \ref{tab:CI Application 1}. For most of the comparisons in Table \ref{tab:CI LDC}, the conclusions are similar to those in Table \ref{tab:CI Application 1}, while the results for 97--98, 97--99, 97--00, and 97--01 are different. In these four groups of comparisons, the LDCs are all estimated to be $1$ and the confidence intervals all shrink to $[1,1]$. This may indicate Lorenz dominance between the Year 1997 and the other years. We then run the Lorenz dominance tests of \citet{Beare2017improved} for the null hypothesis that $F_2$ Lorenz dominates $F_1$, and the null can not be rejected, which is consistent with the estimations in Table \ref{tab:CI LDC}. It is worth noting that, although the asymptotic results in the paper may not hold in cases where $c_m^w(L_1,L_2)\in\{0,1\}$, the results in Table \ref{tab:CI LDC} demonstrate that the proposed method performs well in practice.

\begin{table}[H]
\centering
	\caption{$F_1$-$F_2$ Lorenz Dominance Coefficients}
		
		\scalebox{1}{
			\begin{tabular}{cccccc}
				\hline\hline
			Year & 
			 LDC	& CI 	 & Year	 & LDC & CI   \\
			\hline
94--98 & 1.1151e-05 & [0, 1.3027e-05] & 96--98 & 1.0914e-02 & [0, 9.0150e-02] \\
94--99 & 3.7292e-05 & [0, 6.4719e-05] & 96--99 & 9.1456e-02 & [0, 9.1430e-01] \\
94--00 & 9.0555e-05 & [0, 1.4074e-04] & 96--00 & 2.3386e-02 & [0, 1.8231e-01] \\
94--01 & 2.2390e-04 & [0, 3.4754e-04] & 96--01 & 5.4596e-02 & [0, 2.9047e-01] \\
95--98 & 8.3603e-05 & [0, 1.3646e-04] & 97--98 & 1 & [1, 1] \\
95--99 & 4.1599e-04 & [0, 3.0698e-03] & 97--99 & 1 & [1, 1] \\
95--00 & 7.3659e-04 & [0, 4.5734e-03] & 97--00 & 1 & [1, 1] \\
95--01 & 4.6061e-03 & [0, 2.4259e-02] & 97--01 & 1 & [1, 1] \\
				\hline\hline                         	
			\end{tabular}
		}

		\label{tab:CI LDC}

\end{table}

Note that Lorenz dominance is related to the inequality measure $J_P(L_F)$ \citep[Theorem 2.1]{aaberge2009ranking}, while inverse stochastic dominance is related to the social welfare function $W_P(F)=\mu_F(1-J_P(L_F))$ \citep[p.~645]{aaberge2021ranking}, where $L_F$ is the Lorenz curve for the distribution $F$, and $\mu_F$ is the mean for $F$. For illustration, we consider a preference function $P(t)=t^3-3t^2+3t$ with $t\in[0,1]$ such that $P\in\cap_{\varepsilon_3\in(0,1)}\mathcal{P}_3(\varepsilon_3)$, where $\mathcal{P}_3(\varepsilon_3)$ is defined as in Section \ref{sec.AUISD}. The values of $\mu_F$, $W_P(F)$, and $J_P(L_F)$ for Year 1997, 1998, 1999, 2000, and 2001 are listed in Table \ref{tab: measures}, which are consistent with the results in Tables \ref{tab:CI Application 1} and \ref{tab:CI LDC}.

\begin{table}[h]
\caption{Values of $\mu_F$, $W_P(F)$, and $J_P(L_F)$}
\label{tab: measures}
\centering
\begin{tabular}{c c c c c c}
\hline
\hline
 & 1997 & 1998 & 1999 & 2000 & 2001 \\
\hline
$\mu_F$ & 22625.08788 & 23927.0261 & 24276.49701 & 27607.35936 & 28253.11169 \\

$W_P(F)$ & 13164.54128 & 13596.19132 & 13574.86124 & 15645.45898 & 15940.32593 \\

$J_P(L_F)$ & 0.418144082 & 0.431764262 & 0.440822898 & 0.433286655 & 0.43580282 \\
\hline
\hline
\end{tabular}
\end{table}

\section{Conclusion}

In this paper, we develop a general framework for inferring three types of almost dominances: almost Lorenz dominance, almost inverse stochastic dominance, and almost stochastic dominance. A bootstrap inference procedure is introduced for the corresponding coefficients, which quantify the extent of almost dominances. Simulation studies evaluate the finite sample performance of the proposed estimators and the bootstrap confidence intervals. The proposed methods are applied to analyze the growth of inequality in the United Kingdom, revealing evidence of almost upward inverse stochastic dominance.

While traditional dominance concepts are powerful, strict requirements can limit their applicability. The almost dominance framework allows for small violations of dominance rules, and provides a more nuanced and flexible approach to conducting robust social welfare and inequality comparisons, generalizing investor preferences, and gaining a deeper understanding of income distributional changes.

\putbib
\end{bibunit}

\newpage
\doublespacing
\begin{bibunit}
\appendix

\setcounter{page}{1} 

\begin{center}
    \Large{Almost Dominance: Inference and Application\\Online Supplementary Appendix}

    \vspace{0.5cm}
    
    \large{Xiaojun Song \qquad Zhenting Sun}
\end{center}

\startcontents[sections]
\printcontents[sections]{l}{1}{\setcounter{tocdepth}{2}}

\section{Almost Inverse Stochastic Dominance}\label{sec.ISDC}

Following \citet{leshno2002preferred}, \citet{zheng2018almost}, and \citet{aaberge2021ranking}, in this section, we propose \emph{almost inverse stochastic dominances}, define the \emph{inverse stochastic dominance coefficients}, and establish an inference procedure for these coefficients. 
Let $F_1$ and $F_2$ be two distributions that satisfy Assumption \ref{ass.distribution}.
Define
\[
\Lambda_{j}^{2}\left(  p\right)  =\int_{0}^{p}F_{j}^{-1}\left(  t\right)
\mathrm{d}t,\quad p\in\left[  0,1\right]  .
\]

Let $P$ denote the social planner's preference function. 
As pointed out by \citet{yaari1988controversial} and \citet{aaberge2021ranking}, the social welfare functions are consistent with the condition of second-degree stochastic dominance if and
only if $P^{\prime}  >0$ and $P^{(2)}  <0$. We follow \citet{aaberge2021ranking} and define the set of social planner's preference functions
\[
\mathcal{P}=\left\{
\begin{array}
[c]{c}%
P\in C([0,1]):P\left(  0\right)  =0,P\left(  1\right)  =1,P^{\prime}\left(
1\right)  =0,\\
P^{\prime}\left(  t\right)  >0\text{ and }P^{\left(  2\right)  }\left(
t\right)  <0\text{ for all }t\in\left(  0,1\right)
\end{array}
\right\}  .
\]
We then define the welfare functions $W_P$ such that for every $P\in\mathcal{P}$ and every CDF $F$, 
\[
W_{P}\left(  F\right)  =\int_{0}^{1}P^{\prime}\left(  t\right)  F^{-1}\left(
t\right)  \mathrm{d}t.
\]
Theorem 1 of \citet{aaberge2001axiomatic}
demonstrates that a person who supports Axioms 1--4 in \citet{aaberge2001axiomatic} ranks Lorenz curves according to the criterion $W_P$.

\begin{remark}
\citet{aaberge2021ranking} discuss the relationship between the inequality measure $J_P$ defined in \eqref{eq.inequality measure J} and the social welfare measure $W_P$. \citet{ebert1987size} shows that for every preference function $P\in\mathcal{P}$ and every distribution function $F$ that satisfies Assumption \ref{ass.distribution}, $W_P(F)=\mu_F(1-J_P(L_F))$, where $\mu_F$ denotes the population mean for $F$ and $L_F$ denotes the Lorenz curve corresponding to $F$. As pointed out by \citet{aaberge2021ranking}, the product $\mu_F J_P(L_F)$ presents the loss in social welfare due to inequality in the distribution.
\end{remark}

\subsection{Almost Upward Inverse Stochastic Dominance and Social Welfare}\label{sec.AUISD}

For $m\geq3$, define
\begin{align*}
\Lambda_{j}^{m}\left(  p\right)   &  =\int_{0}^{p}\Lambda_{j}^{m-1}\left(
t\right)  \mathrm{d}t=\int_{0}^{p}\cdots\int_{0}^{t_{3}}\int_{0}^{t_{2}}%
F_{j}^{-1}\left(  t_{1}\right)  \mathrm{d}t_{1}\mathrm{d}t_{2}\cdots
\mathrm{d}t_{m-1}\\
&  =\frac{1}{\left(  m-2\right)  !}\int_{0}^{p}\left(  p-t\right)
^{m-2}F^{-1}\left(  t\right)  \mathrm{d}t,\quad p\in\left[  0,1\right]  .
\end{align*}
\citet{aaberge2021ranking} introduce the $m$th-degree upward inverse stochastic dominance for $m\ge3$. 

\begin{definition}
A distribution $F_{1}$ $m$th-degree upward inverse stochastically
dominates a distribution $F_{2}$ if $\Lambda_{1}^{m}\left(  p\right)
\geq\Lambda_{2}^{m}\left(  p\right)  $ for all $p\in\left[  0,1\right]  $.
\end{definition}

When $m=2$, the upward inverse stochastic dominance is the generalized Lorenz dominance.
We generalize the upward inverse stochastic dominance in \citet{aaberge2021ranking} to almost upward inverse stochastic dominance.

\begin{definition}
For every $\varepsilon_m\in[  0,1/2)  $, the CDF $F_{1}$
$\varepsilon_{m}$-almost $m$th-degree upward inverse stochastically dominates
the CDF $F_{2}$ ($F_{1}$ $\varepsilon_{m}$-A$m$UISD $F_{2}$), if
\begin{align}\label{eq.AmUISD inequality}
\int_0^1\max\left(  \Lambda_{2}^{m}\left(  p\right)  -\Lambda_{1}^{m}\left(p\right)  ,0\right)  \mathrm{d}p\leq\varepsilon_{m}\int_{0}^{1}\left\vert
\Lambda_{2}^{m}\left(  p\right)  -\Lambda_{1}^{m}\left(  p\right)  \right\vert\mathrm{d}p.
\end{align}    
\end{definition}

For every $\varepsilon_{m}\in\left(  0,1/2\right)  $, we define
\[
\mathcal{P}_{m}\left(  \varepsilon_{m}\right)  =\left\{
\begin{array}
[c]{c}%
P\in\mathcal{P}:\left(  -1\right)  ^{m-1}P^{\left(  m\right)  }\geq
0,P^{\left(  j\right)  }\left(  1\right)  =0\text{ for }j\in\{2,\ldots,m-1\},\\
\sup_{t}\left\{  \left(  -1\right)  ^{m-1}P^{\left(  m\right)  }\left(
t\right)  \right\}  \leq\inf_{t}\left\{  \left(  -1\right)  ^{m-1}P^{\left(m\right)  }\left(  t\right)  \right\}  \cdot\left(  \frac{1}{\varepsilon_{m}}-1\right)
\end{array}
\right\}
\]
and%
\[
\mathcal{P}_{m}^{\prime}\left(  \varepsilon_{m}\right)  =\left\{
\begin{array}
[c]{c}%
P\in\mathcal{P}:\left(  -1\right)  ^{m-1}P^{\left(  m\right)  }\geq0,\left(
-1\right)  ^{j-1}P^{\left(  j\right)  }\left(  1\right)  \geq0\text{ for
}j\in\{2,\ldots,m-1\},\\
\sup_{t}\left\{  \left(  -1\right)  ^{m-1}P^{\left(  m\right)  }\left(t\right)  \right\}  \leq\inf_{t}\left\{  \left(  -1\right)  ^{m-1}P^{\left(m\right)  }\left(  t\right)  \right\}  \cdot\left(  \frac{1}{\varepsilon_{m}}-1\right)
\end{array}
\right\}  .
\]
Then the following proposition links A$m$UISD to the social welfare function $W_P$ with $P$ from $\mathcal{P}_m(\varepsilon_m)$ and $\mathcal{P}'_m(\varepsilon_m)$.
\begin{proposition}\label{prop.AmUISD}
For every $\varepsilon_{m}\in\left(  0,1/2\right)  $, if $F_{1}$
$\varepsilon_{m}$-A$m$UISD $F_{2}$, then $W_{P}\left(  F_{1}\right)  \geq
W_{P}\left(  F_{2}\right)  $ for every $P\in\mathcal{P}_{m}\left(
\varepsilon_{m}\right)  $. If $F_{1}$ $\varepsilon_{m}$-A$m$UISD $F_{2}$ and
$\Lambda_{1}^{j}\left(  1\right)  \geq\Lambda_{2}^{j}\left(  1\right)  $ for
all $j\in\left\{  3,\ldots,m\right\}  $, then $W_{P}\left(  F_{1}\right)  \geq
W_{P}\left(  F_{2}\right)  $ for every $P\in\mathcal{P}_{m}^{\prime}\left(
\varepsilon_{m}\right)  $.
\end{proposition}

For $m\ge 3$, following Lemma \ref{lemma.ald equivalence}, it is straightforward to show that $F_{1}$ $\varepsilon_{m}$-A$m$UISD $F_{2}$
if and only if there is $c_{m}\in\left[  0,\varepsilon_{m}\right]  $ such
that
\[
\int_0^1\max\left(  \Lambda_{2}^{m}\left(  p\right)  -\Lambda_{1}^{m}\left(p\right)  ,0\right)  \mathrm{d}p=c_{m}\int_{0}^{1}\left\vert \Lambda_{2}
^{m}\left(  p\right)  -\Lambda_{1}^{m}\left(  p\right)  \right\vert
\mathrm{d}p.
\]
We now define the upward inverse stochastic dominance coefficient as follows. 
\begin{definition}\label{def.mUISDC}
For $m\ge3$, the \emph{$F_1$-$F_2$ $m$th-degree upward inverse stochastic dominance coefficient ($F_1$-$F_2$ $m$UISDC)}, denoted by $c_m^u(F_1,F_2)$, is defined as
\begin{align}
    &c_m^u(F_1,F_2)=\notag\\
    &\inf\left\{\varepsilon_m\in[0,1]:\int_0^1\max\left(  \Lambda_{2}^{m}\left(  p\right)  -\Lambda_{1}^{m  }\left(  p\right),0  \right)
\mathrm{d}p\leq\varepsilon_{m}\int_{0}^{1}\left\vert \Lambda_{2}^{m}\left(p  \right)  -\Lambda_{1}^{m}\left(  p\right)  \right\vert
\mathrm{d}p\right\}.
\end{align}
\end{definition}

We define the difference function
\begin{align}\label{eq.phi m u}
\phi_m^u\left(  p\right)  =\Lambda_{2}^{m}\left(  p\right) -\Lambda_{1}^{m}\left(  p\right),\quad p\in[0,1].
\end{align}
The following lemma summarizes the properties of the $m$UISDC $c_m^u(F_1,F_2)$.

\begin{lemma}\label{lemma.cmu properties}
The $m$UISDC $c_m^u(F_1,F_2)$ is the smallest $\varepsilon_m$ in $[0,1]$ such that \eqref{eq.AmUISD inequality} holds. With \eqref{eq.0timesinf}, it follows that
\begin{align}
    c_m^u(F_1,F_2)&=\frac{\int_0^1\max\left(  \Lambda_{2}^{m}\left(  p\right)  -\Lambda_{1}^{m}\left(  p\right) ,0 \right)
\mathrm{d}p}{\int_{0}^{1}\left\vert \Lambda_{1}^{m}\left(  p\right) -\Lambda_{2}^{m}\left(  p\right)  \right\vert
\mathrm{d}p}\notag\\
&=\frac{\int_{0}^{1}\max\left\{  \phi_m^u\left(  p\right)  ,0\right\}
\mathrm{d}p}{\int_{0}^{1}\max\left\{  \phi_m^u\left(  p\right)  ,0\right\}
\mathrm{d}p+\int_{0}^{1}\max\left\{  -\phi_m^u\left(  p\right)  ,0\right\}
\mathrm{d}p}.
\end{align}
In addition, if $c_m^u(F_1,F_2)\in(0,1]$, then $c_m^u(F_2,F_1)=1-c_m^u(F_1,F_2)$.
\end{lemma}
  
According to Lemma \ref{lemma.cmu properties}, $F_1$ $\varepsilon_m$-A$m$UISD $F_2$ for all $\varepsilon_m\in[c_m^u(F_1,F_2),1/2)$ if $c_m^u(F_1,F_2)<1/2$. On the other hand, $c_m^u(F_1,F_2)>1/2$ implies that $F_2$ $\varepsilon_m$-A$m$UISD $F_1$ for all $\varepsilon_m\in[1-c_m^u(F_1,F_2),1/2)$. Thus, $c_m^u(F_1,F_2)$ presents the degree of almost upward inverse stochastic dominance relationship between $F_1$ and $F_2$, and provides all $\varepsilon_m$ such that the $\varepsilon_m$-A$m$UISD holds.

\begin{proposition}\label{prop.UISDC properties}
If $c_m^u(F_1,F_2)\in(0,1/2)$, it then follows that 
$W_P(F_1)\ge W_P(F_2)$ for all preference functions $P\in \cup_{\varepsilon_m\in[c_m^u(F_1,F_2),1/2)}\mathcal{P}_m(\varepsilon_m)$.  
If, in addition, 
$\Lambda_{1}^{j}\left(  1\right)  \geq\Lambda_{2}^{j}\left(  1\right)  $ for
all $j\in\left\{  3,\ldots,m\right\}  $, it follows that $W_{P}\left(  F_{1}\right)  \geq
W_{P}\left(  F_{2}\right)  $ for every $P\in\cup_{\varepsilon_m\in[c_m^u(F_1,F_2),1/2)}\mathcal{P}_{m}^{\prime}\left(
\varepsilon_{m}\right)  $.
\end{proposition}

Similar to Proposition \ref{prop.c vs Gini}, Proposition \ref{prop.UISDC properties} shows the importance of UISDC concerning welfare functions $W_P$. The smaller the UISDC $c_m^u(F_1,F_2)$ is, the more welfare functions show higher welfare for the distribution $F_1$ compared to $F_2$. With the knowledge of $c_m^u(F_1,F_2)$, we can infer the relationship between $F_1$ and $F_2$ based on a class of welfare functions. 

\subsection{Almost Downward Inverse Stochastic Dominance and Social Welfare}
\citet{aaberge2021ranking} introduce the $m$th-degree downward inverse stochastic dominance for $m\ge3$. For $j=1,2$, define
\begin{align*}
\tilde{\Lambda}_{j}^{3}\left(  p\right)   &  =\int_{p}^{1}\Lambda_{j}%
^{2}\left(  t\right)  \mathrm{d}t=\int_{0}^{1}\Lambda_{j}^{2}\left(  t\right)
\mathrm{d}t-\Lambda_{j}^{3}\left(  p\right) \\
&  =\left(  1-p\right)  \int_{0}^{1}F_{j}^{-1}\left(  t\right)  \mathrm{d}%
t-\int_{p}^{1}\left(  t-p\right)  F_{j}^{-1}\left(  t\right)  \mathrm{d}%
t,\quad p\in\left[  0,1\right]  .
\end{align*}
For $m\geq4$, define
\begin{align*}
\tilde{\Lambda}_{j}^{m}\left(  p\right)   &  =\int_{p}^{1}\tilde{\Lambda}%
_{j}^{m-1}\left(  t\right)  \mathrm{d}t=\int_{p}^{1}\cdots\int_{t_{3}}^{1}\int_{0}^{t_{2}}F_{j}^{-1}\left(
t_{1}\right)  \mathrm{d}t_{1}\mathrm{d}t_{2}\cdots\mathrm{d}t_{m-1}\\
&  =\frac{1}{\left(  m-2\right)  !}\left[  \left(  1-p\right)  ^{m-2}\int
_{0}^{1}F^{-1}\left(  t\right)  \mathrm{d}t-\int_{p}^{1}\left(  t-p\right)
^{m-2}F^{-1}\left(  t\right)  \mathrm{d}t\right]  ,\quad p\in\left[
0,1\right].
\end{align*}

\begin{definition}
A distribution $F_{1}$ $m$th-degree downward inverse
stochastically dominates a distribution $F_{2}$ if $\tilde{\Lambda}_{1}%
^{m}\left(  p\right)  \geq\tilde{\Lambda}_{2}^{m}\left(  p\right)  $ for all
$p\in\left[  0,1\right]  $.    
\end{definition}

We generalize the downward inverse stochastic dominance in \citet{aaberge2021ranking} to almost downward inverse stochastic dominance.

\begin{definition}
For every $\varepsilon_m\in[  0,1/2)  $, the CDF $F_{1}$
$\varepsilon_{m}$-almost $m$th-degree downward inverse stochastically
dominates the CDF $F_{2}$ ($F_{1}$ $\varepsilon_{m}$-A$m$DISD $F_{2}$), if
\[
\int_0^1\max\left(  \tilde{\Lambda}_{2}^{m}\left(  p\right)  -\tilde{\Lambda}%
_{1}^{m}\left(  p\right)  ,0\right)  \mathrm{d}p\leq\varepsilon_{m}\int
_{0}^{1}\left\vert \tilde{\Lambda}_{2}^{m}\left(  p\right)  -\tilde{\Lambda
}_{1}^{m}\left(  p\right)  \right\vert \mathrm{d}p.
\]
\end{definition}

For every $\varepsilon_{m}\in\left(  0,1/2\right)  $, we define
\[
\mathcal{\tilde{P}}_{m}\left(  \varepsilon_{m}\right)  =\left\{
\begin{array}
[c]{c}%
P\in\mathcal{P}:P^{\left(  m\right)  }\leq0,P^{\left(
j\right)  }\left(  0\right)  =0\text{ for }j\in\{2,\ldots,m-1\},\\
\sup_{t}\left\{  -P^{\left(  m\right)  }\left(  t\right)  \right\}  \leq
\inf_{t}\left\{  -P^{\left(  m\right)  }\left(  t\right)  \right\}
\cdot\left(  \frac{1}{\varepsilon_{m}}-1\right)
\end{array}
\right\}
\]
and%
\[
\mathcal{\tilde{P}}_{m}^{\prime}\left(  \varepsilon_{m}\right)  =\left\{
\begin{array}
[c]{c}%
P\in\mathcal{P}:P^{\left(  m\right)  }\leq0,P^{\left(
j\right)  }\left(  0\right)\le0\text{ for }j\in\{2,\ldots,m-1\},\\
\sup_{t}\left\{  -P^{\left(  m\right)  }\left(  t\right)  \right\}  \leq
\inf_{t}\left\{  -P^{\left(  m\right)  }\left(  t\right)  \right\}
\cdot\left(  \frac{1}{\varepsilon_{m}}-1\right)
\end{array}
\right\}  .
\]
Then the following proposition links A$m$DISD to the social welfare function $W_P$ with $P$ from $\tilde{\mathcal{P}}_m(\varepsilon_m)$ and $\tilde{\mathcal{P}}'_m(\varepsilon_m)$.

\begin{proposition}\label{prop.AmDISD} 
For every $\varepsilon_{m}\in\left(  0,1/2\right)  $, if $F_{1}$
$\varepsilon_{m}$-A$m$DISD $F_{2}$, then $W_{P}\left(  F_{1}\right)  \geq
W_{P}\left(  F_{2}\right)  $ for every $P\in\mathcal{\tilde{P}}_{m}\left(
\varepsilon_{m}\right)  $. If $F_{1}$ $\varepsilon_{m}$-A$m$DISD $F_{2}$ and
$\tilde{\Lambda}_{1}^{j}\left(  0\right)  \geq\tilde{\Lambda}_{2}^{j}\left(
0\right)  $ for all $j\in\left\{  3,\ldots,m\right\}  $, then $W_{P}\left(
F_{1}\right)  \geq W_{P}\left(  F_{2}\right)  $ for every $P\in\mathcal{\tilde
{P}}_{m}^{\prime}\left(  \varepsilon_{m}\right)  $.    
\end{proposition}

For $m\ge 3$, following Lemma \ref{lemma.ald equivalence}, it is straightforward to show that $F_{1}$ $\varepsilon_{m}$-A$m$DISD $F_{2}$ if and only if there is $c_{m}\in\left[  0,\varepsilon_{m}\right]  $ such
that
\[
\int_0^1\max\left(  \tilde{\Lambda}_{2}^{m}\left(  p\right)  -\tilde{\Lambda}_{1}^{m}\left(
p\right)  ,0\right)  \mathrm{d}p=c_{m}\int_{0}^{1}\left\vert \tilde{\Lambda}_{2}
^{m}\left(  p\right)  -\tilde{\Lambda}_{1}^{m}\left(  p\right)  \right\vert
\mathrm{d}p.
\]
We now define the downward inverse stochastic dominance coefficient as follows. 

\begin{definition}\label{def.mDISDC}
For $m\ge3$, the \emph{$F_1$-$F_2$ $m$th-degree downward inverse stochastic dominance coefficient ($F_1$-$F_2$ $m$DISDC)}, denoted by $c_m^d(F_1,F_2)$, is defined as
\begin{align}\label{eq.AmDISD inequality}
    &c_m^d(F_1,F_2)=\notag\\
    &\inf\left\{\varepsilon_m\in[0,1]:\int_0^1\max\left(  \tilde{\Lambda}_{2}^{m}\left(  p\right)  -\tilde{\Lambda}_{1}^{m  }\left(  p\right),0  \right)
\mathrm{d}p\leq\varepsilon_{m}\int_{0}^{1}\left\vert \tilde{\Lambda}_{2}^{m}\left(p  \right)  -\tilde{\Lambda}_{1}^{m}\left(  p\right)  \right\vert
\mathrm{d}p\right\}.
\end{align}
\end{definition}

We define the difference function
\begin{align}\label{eq.phi m d}
\phi_m^d\left(  p\right)  =\tilde{\Lambda}_{2}^{m}\left(  p\right) -\tilde{\Lambda}_{1}^{m}\left(  p\right),\quad p\in[0,1].
\end{align}
The following lemma summarizes the properties of the $m$DISDC $c_m^d(F_1,F_2)$.

\begin{lemma}\label{lemma.cmd properties}
The $m$DISDC $c_m^d(F_1,F_2)$ is the smallest $\varepsilon_m$ in $[0,1]$ such that \eqref{eq.AmDISD inequality} holds. With \eqref{eq.0timesinf}, it follows that
\begin{align}
    c_m^d(F_1,F_2)&=\frac{\int_0^1\max\left(  \tilde{\Lambda}_{2}^{m}\left(  p\right)  -\tilde{\Lambda}_{1}^{m}\left(  p\right) ,0 \right)
\mathrm{d}p}{\int_{0}^{1}\left\vert \tilde{\Lambda}_{1}^{m}\left(  p\right) -\tilde{\Lambda}_{2}^{m}\left(  p\right)  \right\vert
\mathrm{d}p}\notag\\
&=\frac{\int_{0}^{1}\max\left\{  \phi_m^d\left(  p\right)  ,0\right\}
\mathrm{d}p}{\int_{0}^{1}\max\left\{  \phi_m^d\left(  p\right)  ,0\right\}
\mathrm{d}p+\int_{0}^{1}\max\left\{  -\phi_m^d\left(  p\right)  ,0\right\}
\mathrm{d}p}.
\end{align}
In addition, if $c_m^d(F_1,F_2)\in(0,1]$, then $c_m^d(F_2,F_1)=1-c_m^d(F_1,F_2)$.
\end{lemma}
  
According to Lemma \ref{lemma.cmd properties}, $F_1$ $\varepsilon_m$-A$m$DISD $F_2$ for all $\varepsilon_m\in[c_m^d(F_1,F_2),1/2)$ if $c_m^d(F_1,F_2)<1/2$. On the other hand,  $c_m^d(F_1,F_2)>1/2$ implies that $F_2$ $\varepsilon_m$-A$m$DISD $F_1$ for all $\varepsilon_m\in[1-c_m^d(F_1,F_2),1/2)$. Thus, $c_m^d(F_1,F_2)$ presents the degree of almost downward inverse stochastic dominance relationship between $F_1$ and $F_2$, and provides all $\varepsilon_m$ such that the $\varepsilon_m$-A$m$DISD holds.

\begin{proposition}\label{prop.DISDC properties}
If $c_m^d(F_1,F_2)\in(0,1/2)$, it then follows that 
$W_P(F_1)\ge W_P(F_2)$ for all preference functions $P\in \cup_{\varepsilon_m\in[c_m^d(F_1,F_2),1/2)}\tilde{\mathcal{P}}_m(\varepsilon_m)$. If, in addition,
$\tilde{\Lambda}_{1}^{j}\left(  0\right)  \geq\tilde{\Lambda}_{2}^{j}\left(
0\right)  $ for all $j\in\left\{  3,\ldots,m\right\}  $, then $W_{P}\left(
F_{1}\right)  \geq W_{P}\left(  F_{2}\right)  $ for every $P\in\cup_{\varepsilon_m\in[c_m^d(F_1,F_2),1/2)}\mathcal{\tilde
{P}}_{m}^{\prime}\left(  \varepsilon_{m}\right)  $.    
\end{proposition}

Similar to Proposition \ref{prop.c vs Gini}, Proposition \ref{prop.DISDC properties} shows the importance of DISDC concerning welfare functions $W_P$. The smaller the DISDC $c_m^d(F_1,F_2)$ is, the more welfare functions show higher welfare for the distribution $F_1$ compared to $F_2$. With the knowledge of $c_m^d(F_1,F_2)$, we can infer the relationship between $F_1$ and $F_2$ based on a class of welfare functions.

\subsection{Estimation and Inference}
Let $\mathcal{F}$, $\hat{F}_j$, and $\hat{Q}_j$ be defined as in \eqref{eq.F map}, \eqref{eq.empirical cdfs}, and \eqref{eq.empirical quantile}. Let $\mathcal{B}_{1}$ and $\mathcal{B}_{2}$ be
the centered Gaussian random elements in $C(\left[  0,1\right]  )$ defined in Section \ref{secbsasym}. It is
straightforward to show that (see the proof of Lemma \ref{lemma.F asymptotic limit})
\[
\left(
\begin{array}
[c]{c}%
n_{1}^{1/2}\left(  \hat{F}_{1}-F_{1}\right)  \\
n_{2}^{1/2}\left(  \hat{F}_{2}-F_{2}\right)
\end{array}
\right)  \leadsto\left(
\begin{array}
[c]{c}%
\mathcal{B}_{1}\circ F_{1}\\
\mathcal{B}_{2}\circ F_{2}%
\end{array}
\right)
\]
in $\ell^{\infty}\left(  [0,\infty)\right)  \times\ell^{\infty}\left(
[0,\infty)\right)  $.
\citet{Beare2017improved} show that under Assumptions \ref{ass.distribution} and \ref{ass.data}, we can apply the results of \citet{K17} and obtain
\[
\left(
\begin{array}
[c]{c}%
n_{1}^{1/2}\left(  \hat{Q}_{1}-Q_{1}\right)  \\
n_{2}^{1/2}\left(  \hat{Q}_{2}-Q_{2}\right)
\end{array}
\right)  \leadsto\left(
\begin{array}
[c]{c}%
-Q_{1}^{\prime}\cdot\mathcal{B}_{1}\\
-Q_{2}^{\prime}\cdot\mathcal{B}_{2}%
\end{array}
\right)
\]
in $L^{1}\left(  [0,1]\right)  \times L^{1}\left(  [0,1]\right)  $. Recall that for all $p\in\left[  0,1\right]  $, $\Lambda_{j}^{2}\left(
p\right)  =\int_{0}^{p}Q_{j}\left(  t\right)  \mathrm{d}t$ and we estimate
$\Lambda_{j}^{2}\left(  p\right)  $ by $\hat{\Lambda}_{j}^{2}\left(  p\right)
=\int_{0}^{p}\hat{Q}_{j}\left(  t\right)  \mathrm{d}t$. As shown in \citet{jiang2023nonparametric}, by continuous
mapping theorem,
\begin{align}\label{eq.Lambda weak convergence}
\left(
\begin{array}
[c]{c}%
n_{1}^{1/2}\left(  \hat{\Lambda}_{1}^2-{\Lambda}_{1}^2\right)  \\
n_{2}^{1/2}\left(  \hat{\Lambda}_{2}^2-{\Lambda}_{2}^2\right)
\end{array}
\right)
\leadsto\binom{\mathcal{V}_{1}}{\mathcal{V}_{2}},
\end{align}
where $\mathcal{V}_{j}=\int_{0}^{p}-Q_{j}^{\prime}\left(  t\right)
\mathcal{B}_{j}(t)\mathrm{d}t$. By the continuous mapping theorem again, it follows
that
\[
\sqrt{T_{n}}\left\{  (\hat{\Lambda}_{2}^{2}-\hat{\Lambda}_{1}^{2})-\left(
\Lambda_{2}^{2}-\Lambda_{1}^{2}\right)  \right\}  \leadsto\mathbb{G}_{\Lambda},
\]
where $\mathbb{G}_{\Lambda}=\sqrt{\lambda
}\mathcal{V}_{2}-\sqrt{1-\lambda}\mathcal{V}_{1}$ and for all $p,p^{\prime}\in\left[  0,1\right]  $,
\begin{align*}
E\left[  \mathbb{G}_{\Lambda}\left(  p\right)  \mathbb{G}_{\Lambda}\left(
p^{\prime}\right)  \right]    =&\,\left(  1-\lambda\right)  E\left[
\mathcal{V}_{1}\left(  p\right)  \mathcal{V}_{1}\left(  p^{\prime}\right)
\right]  -\sqrt{\lambda\left(  1-\lambda\right)  }E\left[  \mathcal{V}%
_{1}\left(  p\right)  \mathcal{V}_{2}\left(  p^{\prime}\right)  \right]  \\
& -\sqrt{\lambda\left(  1-\lambda\right)  }E\left[  \mathcal{V}_{2}\left(
p\right)  \mathcal{V}_{1}\left(  p^{\prime}\right)  \right]  +\lambda E\left[
\mathcal{V}_{2}\left(  p\right)  \mathcal{V}_{2}\left(  p^{\prime}\right)
\right].
\end{align*}
By Lemma 3.1 of \citet{jiang2023nonparametric}, 
for $j,j^{\prime}\in\left\{  1,2\right\}  $ and $p,p^{\prime}\in\left[
0,1\right]  $, we have that
\begin{align}\label{eq.Cov V}
E\left[  \mathcal{V}_{j}\left(  p\right)  \mathcal{V}_{j^{\prime}}(p^{\prime})  \right]  =Cov\left(  Q_{j}\left(  p\right)  \wedge
X_i^{j},Q_{j^{\prime}}(  p^{\prime})  \wedge X_i^{j^{\prime}}\right)  .
\end{align}

For $j=1,2$, $m\geq3$, and $p\in[0,1]$, define
\begin{align*}
\hat{\Lambda}_{j}^{m}\left(  p\right)  =\int_{0}^{p}\cdots\int_{0}^{t_{3}}\int_{0}^{t_{2}}%
\hat{Q}_{j}\left(  t_{1}\right)  \mathrm{d}t_{1}\mathrm{d}t_{2}\cdots
\mathrm{d}t_{m-1} =\int_{0}^{p}\cdots\int_{0}^{t_{3}}\hat{\Lambda}_{j}^{2}\left(  t_{2}\right)  \mathrm{d}t_{2}\cdots
\mathrm{d}t_{m-1}
\end{align*}
and 
\begin{align*}
\hat{\tilde{\Lambda}}_{j}^{m}\left(  p\right)  =\int_{p}^{1}\cdots
\int_{t_{3}}^{1}\int_{0}^{t_{2}}\hat{Q}_{j}\left(  t_{1}\right)
\mathrm{d}t_{1}\mathrm{d}t_{2}\cdots\mathrm{d}t_{m-1}=\int_{p}^{1}\cdots\int_{t_{3}}^{1}\hat{\Lambda}_{j}^{2}\left(  t_{2}\right)
\mathrm{d}t_{2}\cdots\mathrm{d}t_{m-1}.
\end{align*}
For $m\ge3$, define the empirical difference functions
\[
\hat{\phi}_m^u\left(  p\right)  =\hat{\Lambda}_{2}^{m}\left(  p\right) -\hat{\Lambda}_{1}^{m}\left(  p\right) \text{ and } {\hat{\phi}}_m^d\left(  p\right)  =\hat{\tilde{\Lambda}}_{2}^{m}\left(  p\right) -\hat{\tilde{\Lambda}}_{1}^{m}\left(  p\right), \quad p\in[0,1].
\]
The estimator for $c_{m}^w\left(  F_{1},F_{2}\right)  $ can be constructed by
\[
\hat{c}_{m}^w\left(  F_{1},F_{2}\right)  =\frac{\mathcal{F}_{1}(  \hat
{\phi}_{m}^w)  }{\mathcal{F}_{1}(  \hat{\phi}_{m}^w)
+\mathcal{F}_{2}(  \hat{\phi}_{m}^w)  }=\mathcal{F}(  \hat{\phi
}_{m}^w)  .
\]

Similar to Proposition \ref{prop.c asymptotic limit}, we can show that $\mathcal{F}$ is Hadamard directionally differentiable at $\phi_m^w$ defined in \eqref{eq.phi m u} and \eqref{eq.phi m d} such that for every $h\in\ell^{\infty}([0,1])$,
\begin{align*}
\mathcal{F}_{\phi_m^w}^{\prime}\left(  h\right)  =\frac{\mathcal{F}_{1\phi_m^w}^{\prime}\left(  h\right)  \mathcal{F}_{2}\left(
\phi_m^w\right)  -\mathcal{F}_{1}\left(  \phi_m^w\right)  \mathcal{F}_{2\phi_m^w}^{\prime
}\left(  h\right)  }{\left(  \mathcal{F}_{1}\left(  \phi_m^w\right)
+\mathcal{F}_{2}\left(  \phi_m^w\right)  \right)  ^{2}},
\end{align*}
where 
\[
\mathcal{F}_{1\phi_m^w}^{\prime}\left(  h\right)  =\int_{B_{+}\left(  \phi_m^w\right)
}h\left(  p\right)  \mathrm{d}p+\int_{B_{0}\left(  \phi_m^w\right)  }\max\left\{
h\left(  p\right)  ,0\right\}  \mathrm{d}p
\]
and
\[
\mathcal{F}_{2\phi_m^w}^{\prime}\left(  h\right)  =\int_{B_{+}\left(
-\phi_m^w\right)  }-h\left(  p\right)  \mathrm{d}p+\int_{B_{0}\left(
\phi_m^w\right)  }\max\left\{  -h\left(  p\right)  ,0\right\}  \mathrm{d}p
\]
with $B_0$ and $B_+$ defined in \eqref{eq.B0 B+}.

\begin{proposition}\label{prop.cm inverse asymptotic limit}
Under Assumptions \ref{ass.distribution} and \ref{ass.data}, for $m\ge3$ and $w\in\{u,d\}$, it follows that
\begin{align*}
\sqrt{T_{n}}(  \hat{\phi}_{m}^w-\phi_{m}^w)    \leadsto\mathbb{G}_{m}^w
\end{align*}
for some random element $\mathbb{G}_m^w$ with 
\begin{align*}
Var(\mathbb{G}_m^u(p))=\int_{0}^{p}\cdots\int_{0}^{t_{3}^{\prime}}\left(  \int_{0}^{p}\cdots\int_{0}^{t_{3}}E\left[
\mathbb{G}_{\Lambda}(  t_{2})  \mathbb{G}_{\Lambda}(  t_{2}^{\prime
})  \right]  \mathrm{d}t_{2}\cdots\mathrm{d}t_{m-1}\right)  \mathrm{d}t_{2}^{\prime}\cdots\mathrm{d}t_{m-1}^{\prime}
\end{align*}
and 
\begin{align*}
Var(\mathbb{G}_m^d(p))=    \int_{p}^{1}\cdots\int_{t_{3}^{\prime}}^1\left(  \int_{p}^{1}\cdots\int_{t_{3}}^1E\left[
\mathbb{G}_{\Lambda}(  t_{2})  \mathbb{G}_{\Lambda}(  t_{2}^{\prime
})  \right]  \mathrm{d}t_{2}\cdots\mathrm{d}t_{m-1}\right)  \mathrm{d}t_{2}^{\prime}\cdots\mathrm{d}t_{m-1}^{\prime}
\end{align*}
for every $p\in[0,1]$.
Moreover, we have that
\[
\sqrt{T_{n}}\left\{  \hat{c}_{m}^w\left(  F_{1},F_{2}\right)  -c_{m}^w\left(
F_{1},F_{2}\right)  \right\}  \leadsto\mathcal{F}_{\phi_{m}^w}^{\prime}\left(
\mathbb{G}_{m}^w\right).
\]
\end{proposition}

Let $\hat{\lambda}=n_1/(n_1+n_2)$. For $j,j^{\prime}\in\left\{  1,2\right\}  $ and $p,p^{\prime}\in\left[
0,1\right]  $, by \eqref{eq.Cov V}, we estimate $E\left[  \mathcal{V}_{j}\left(  p\right)
\mathcal{V}_{j^{\prime}}\left(  p^{\prime}\right)  \right]  $ by $\hat E\left[  \mathcal{V}_{j}\left(  p\right)
\mathcal{V}_{j^{\prime}}\left(  p^{\prime}\right)  \right]  $ which is defined as the sample covariance of the two samples
\begin{align*}
\left\{\hat{Q}_j(p)\wedge X^j_i\right\}_{i=1}^{n_j} \text{ and } \left\{\hat{Q}_{j'}(p')\wedge X^{j'}_i\right\}_{i=1}^{n_{j'}}.
\end{align*}
For independent samples, we estimate $E\left[  \mathbb{G}_{\Lambda}\left(
p\right)  \mathbb{G}_{\Lambda}\left(  p^{\prime}\right)  \right]  $ by
\[
\hat{E}\left[  \mathbb{G}_{\Lambda}\left(  p\right)  \mathbb{G}_{\Lambda
}(p^{\prime})\right]  =(  1-\hat{\lambda})  \hat{E}\left[
\mathcal{V}_{1}\left(  p\right)  \mathcal{V}_{1}(p^{\prime})\right]
+\hat{\lambda}\hat{E}\left[  \mathcal{V}_{2}\left(  p\right)  \mathcal{V}%
_{2}(p^{\prime})\right]  .
\]
For matched pairs,
\begin{align*}
\hat{E}\left[  \mathbb{G}_{\Lambda}\left(  p\right)  \mathbb{G}_{\Lambda
}(p^{\prime})\right]     =&\,(  1-\hat{\lambda})  \hat{E}\left[
\mathcal{V}_{1}\left(  p\right)  \mathcal{V}_{1}(p^{\prime})\right]
-\sqrt{\hat{\lambda}(  1-\hat{\lambda})  }\hat{E}\left[
\mathcal{V}_{1}\left(  p\right)  \mathcal{V}_{2}(p^{\prime})\right]  \\
& -\sqrt{\hat{\lambda}(  1-\hat{\lambda})  }\hat{E}\left[
\mathcal{V}_{2}\left(  p\right)  \mathcal{V}_{1}(p^{\prime})\right]
+\hat{\lambda}\hat{E}\left[  \mathcal{V}_{2}\left(  p\right)  \mathcal{V}%
_{2}(p^{\prime})\right]  .
\end{align*}
We then estimate the variance $Var(\mathbb{G}_m^w(p))$ by
\begin{align*}
    \hat{\sigma}_m^{u}(p)^2=\int_{0}^{p}\cdots\int_{0}^{t_{3}^{\prime}}\left(  \int_{0}^{p}\cdots\int_{0}^{t_{3}}\hat{E}\left[
\mathbb{G}_{\Lambda}(  t_{2})  \mathbb{G}_{\Lambda}(  t_{2}^{\prime
})  \right]  \mathrm{d}t_{2}\cdots\mathrm{d}t_{m-1}\right)  \mathrm{d}t_{2}^{\prime}\cdots\mathrm{d}t_{m-1}^{\prime}
\end{align*}
and 
\begin{align*}
    \hat{\sigma}_m^{d}(p)^2=\int_{p}^{1}\cdots\int_{t_{3}^{\prime}}^1\left(  \int_{p}^{1}\cdots\int_{t_{3}}^1\hat{E}\left[
\mathbb{G}_{\Lambda}(  t_{2})  \mathbb{G}_{\Lambda}(  t_{2}^{\prime
})  \right]  \mathrm{d}t_{2}\cdots\mathrm{d}t_{m-1}\right)  \mathrm{d}t_{2}^{\prime}\cdots\mathrm{d}t_{m-1}^{\prime}
\end{align*}
for $p\in[0,1]$.

\subsubsection{Bootstrap Confidence Intervals for UISDC and DISDC}
 
Similar to Section \ref{sec.bootstrap CI Lorenz}, for $m\ge3$ and $w\in\{u,d\}$, we construct the estimators of $B_{+}\left(  \phi_m^w\right)  $, $B_{+}\left(
-\phi_m^w\right)  $, and $B_{0}\left(  \phi_m^w\right)  $ by
\begin{align*}
&\widehat{B_{+}\left(  \phi^w_m\right)  } =\left\{  p\in\left[  0,1\right]
:\frac{\sqrt{T_{n}}\hat{\phi}^w_m\left(  p\right)  }{\xi_{0}\vee\hat{\sigma
}^w_m\left(  p\right)  }>t_{n}\right\}  ,\widehat{B_{+}\left(  -\phi^w_m\right)
}=\left\{  p\in\left[  0,1\right]  :\frac{\sqrt{T_{n}}\hat{\phi}^w_m\left(
p\right)  }{\xi_{0}\vee\hat{\sigma}^w_m\left(  p\right)  }<-t_{n}\right\}  ,\\
&\text{and }\widehat{B_{0}\left(  \phi_m^w\right)  } =\left\{  p\in\left[
0,1\right]  :\left\vert \frac{\sqrt{T_{n}}\hat{\phi}^w_m\left(  p\right)  }
{\xi_{0}\vee\hat{\sigma}^w_m\left(  p\right)  }\right\vert \leq t_{n}\right\}  ,
\end{align*}
where $t_{n}\rightarrow\infty$ and $t_{n}/\sqrt{T_{n}}\rightarrow0$ as
$n\rightarrow\infty$. We then construct the estimator of $\mathcal{F}_{1\phi_m
^w}^{\prime}$ and $\mathcal{F}_{2\phi_m^w}^{\prime}$ by
\begin{align*}
&\mathcal{\hat{F}}_{1\phi_m^w}^{\prime}\left(  h\right)  =\int_{\widehat
{B_{+}\left(  \phi_m^w\right)  }}h\left(  p\right)  \mathrm{d}p+\int
_{\widehat{B_{0}\left(  \phi_m^w\right)  }}\max\left\{  h\left(  p\right)
,0\right\}  \mathrm{d}p\\
&\text{and } \mathcal{\hat{F}}_{2\phi_m^w}^{\prime}\left(  h\right)
=\int_{\widehat{B_{+}\left(  -\phi_m^w\right)  }}-h\left(  p\right)
\mathrm{d}p+\int_{\widehat{B_{0}\left(  \phi_m^w\right)  }}\max\left\{  -h\left(p\right)  ,0\right\}  \mathrm{d}p
\end{align*}
for every $h$ $\in\ell^{\infty}\left(  \left[  0,1\right]  \right)  $. The
estimator of $\mathcal{F}_{\phi_m^w}^{\prime}$ is defined by
\[
\mathcal{\hat{F}}_{\phi_m^w}^{\prime}\left(  h\right)  =\frac{\mathcal{\hat{F}
}_{1\phi_m^w}^{\prime}\left(  h\right)  \mathcal{F}_{2}(  \hat{\phi}_m^w)
-\mathcal{F}_{1}(  \hat{\phi}_m^w)  \mathcal{\hat{F}}_{2\phi_m^w}^{\prime
}\left(  h\right)  }{\left(  \mathcal{F}_{1}(  \hat{\phi}_m^w)
+\mathcal{F}_{2}(  \hat{\phi}_m^w)  \right)  ^{2}}
\]
for every $h$ $\in\ell^{\infty}\left(  \left[  0,1\right]  \right)  $.

As discussed in Section \ref{sec.bootstrap CI Lorenz}, for independent samples, we draw a bootstrap sample $\{\hat{X}_{i}^{j}%
\}_{i=1}^{n_j}$ identically and independently from $\{X_{i}^{j}\}_{i=1}^{n_j}$ for
$j=1,2$, where $\{\hat{X}_{i}^{1}\}_{i=1}^{n_1}$ is jointly independent of
$\{\hat{X}_{i}^{2}\}_{i=1}^{n_2}$. For matched pairs, we draw a bootstrap sample
$\{(\hat{X}_{i}^{1},\hat{X}_{i}^{2})\}_{i=1}^{n}$ identically and
independently from $\{(X_{i}^{1},X_{i}^{2})\}_{i=1}^{n}$. 
Let $\hat{Q}_j^{*}$ be defined as in \eqref{eq.bootstrap quantile}.
For $j=1,2$, $m\geq4$, and $p\in[0,1]$, define
\begin{align*}
\hat{\Lambda}_{j}^{m*}\left(  p\right)  =\int_{0}^{p}\cdots\int_{0}^{t_{3}}\int_{0}^{t_{2}}%
\hat{Q}_{j}^*\left(  t_{1}\right)  \mathrm{d}t_{1}\mathrm{d}t_{2}\cdots
\mathrm{d}t_{m-1}=\frac{1}{\left(  m-2\right)  !}\int_{0}^{p}\left(  p-t\right)
^{m-2}\hat{Q}_j^*\left(  t\right)  \mathrm{d}t
\end{align*}
and 
\begin{align*}
\hat{\tilde{\Lambda}}_{j}^{m*}\left(  p\right)   &  =\int_{p}^{1}\cdots
\int_{t_{3}}^{1}\int_{0}^{t_{2}}\hat{Q}_{j}^*\left(  t_{1}\right)
\mathrm{d}t_{1}\mathrm{d}t_{2}\cdots\mathrm{d}t_{m-1}\\
&  =\frac{1}{\left(  m-2\right)  !}\left[  \left(  1-p\right)  ^{m-2}\int
_{0}^{1}\hat{Q}_{j}^*\left(  t\right)  \mathrm{d}t-\int_{p}^{1}\left(
t-p\right)  ^{m-2}\hat{Q}_{j}^*\left(  t\right)  \mathrm{d}t\right].
\end{align*}
Define the bootstrap difference functions
\[
\hat{\phi}_m^{u*}\left(  p\right)  =\hat{\Lambda}_{2}^{m*}\left(  p\right) -\hat{\Lambda}_{1}^{m*}\left(  p\right) \text{ and } {\hat{\phi}}_m^{d*}\left(  p\right)  =\hat{\tilde{\Lambda}}_{2}^{m*}\left(  p\right) -\hat{\tilde{\Lambda}}_{1}^{m*}\left(  p\right), \quad p\in[0,1].
\]
For $w\in\{u,d\}$, we define the bootstrap estimation of $c_m^w(F_1,F_2)$ by 
\begin{align}
    \hat{c}_m^{w*}(F_1,F_2)=\mathcal{\hat{F}}_{\phi_m^w}^{\prime}(\sqrt{T_n}(\hat{\phi}_m^{w*}-\hat{\phi}_m^w)).
\end{align}
For every $\beta\in(0,1)$, let $c_{m,\beta}^w$ denote the $\beta$ quantile of the distribution of $\mathcal{F}'_{\phi_m^w}(\mathbb{G}_m^w)$.
We construct the bootstrap approximation of $c_{m,\beta}^w$ by 
\begin{align}
    \hat{c}_{m,\beta}^w=\inf\left\{c: \mathbb{P}\left(\hat{c}_m^{w*}(F_1,F_2)\le c|\{X^1_i\}_{i=1}^{n_1},\{X^2_i\}_{i=1}^{n_2}\right)\ge \beta\right\}.
\end{align}
Empirically, we approximate $\hat{c}_{m,\beta}^w$ by computing the $\beta$ quantile of the $n_B$ independently generated $\hat{c}_m^{w*}(F_1,F_2)$, where $n_B$ is chosen as large as is computationally convenient. 

For a nominal significance level $\alpha\in(0,1/2)$, we construct the $1-\alpha$ confidence interval by
\begin{align}
    \mathrm{CI}_{m,1-\alpha}^w=[\hat{c}_m^w(F_1,F_2)-T_n^{-1/2}\hat{c}_{m,1-\alpha/2}^w,\hat{c}_m^w(F_1,F_2)-T_n^{-1/2}\hat{c}_{m,\alpha/2}^w].
\end{align}

\begin{proposition}\label{prop.confidence interval ISDC}
Suppose that Assumptions \ref{ass.distribution} and \ref{ass.data} hold. For $m\ge3$ and $w\in\{u,d\}$, if $c_m^w(F_1,F_2)\in(0,1)$ and the CDF of $\mathcal{F}'_{\phi_m^w}(\mathbb{G}_m^w)$ is continuous and increasing at $c_{m,\alpha}^w$ and $c_{m,1-\alpha}^w$, then it follows that
\begin{align}
    \lim_{n\rightarrow\infty}\mathbb{P}(c_m^w(F_1,F_2)\in\mathrm{CI}^w_{m,1-\alpha})= 1-\alpha.
\end{align}
\end{proposition}

\section{Almost Stochastic Dominance}\label{sec.ASD}
In this section, we extend our results to almost stochastic dominance. 
We follow \citet{tsetlin2015generalized} and let $u$ be a decision maker's utility function and $u^{\left(
k\right)  }$ the $k$th derivative of $u$. Let $F_{1}$ and $F_{2}$ be CDFs that are restricted on the
support $\left[  a,b\right]  $ such that $F_1(a)=F_2(a)=0$ and $F_1(b)=F_2(b)=1$. For $j=1,2$ and $k\geq2$, we define
\[
F_{j}^{\left(  k\right)  }\left(  x\right)  =\int_{a}^{x}F_{j}^{\left(
k-1\right)  }\left(  t\right)  \mathrm{d}t,\quad x\in[a,b],
\]
with $F_{j}^{\left(  1\right)  }\left(  x\right)  =F_{j}\left(  x\right)  $ for all $x\in[a,b]$. To fix the idea, we first consider almost first-degree stochastic dominance. Define
\[
S_{1}\left(  F_{1},F_{2}\right)  =\left\{  x\in\left[  a,b\right]
:F_{1}\left(  x\right)  > F_{2}\left(  x\right)  \right\}.
\]
The original definition of \emph{almost first-degree stochastic dominance} provided by \citet{leshno2002preferred} is as follows.

\begin{definition}\label{def.AFSD}
For every $\varepsilon_{1}\in\left[  0,1/2\right)  $, $F_{1}$
$\varepsilon_{1}$-almost first-degree stochastically dominates $F_{2}$ ($F_1$ $ \varepsilon_{1}$-AFSD $F_2$) if
\begin{align}\label{eq.AFSD}
\int_{S_{1}\left(  F_{1},F_{2}\right)  }(  F_{1}\left(  x\right)
-F_{2}\left(  x\right)  )  \mathrm{d}x\leq\varepsilon_{1}\int_{a}%
^{b}\left\vert F_{1}\left(  x\right)  -F_{2}\left(  x\right)  \right\vert
\mathrm{d}x.    
\end{align}
\end{definition}

For every $\varepsilon_{1}\in\left(  0,1\right)  $, define the utility function class
\[
U_{1}\left(  \varepsilon_{1}\right)  =\left\{  u:u^{\left(  1\right)  }%
>0,\sup_x\left\{  u^{\left(  1\right)  }\left(  x\right)  \right\}  \leq
\inf_x\left\{  u^{\left(  1\right)  }\left(  x\right)  \right\}  \left(
\frac{1}{\varepsilon_{1}}-1\right)  \right\}  .
\]
\citet[][Theorem 1]{leshno2002preferred} associate $\varepsilon_{1}$-AFSD with the class
$U_{1}\left(  \varepsilon_{1}\right)  $ by showing that for every $\varepsilon
_{1}\in\left(  0,1/2\right)  $, $F_{1}$ $\varepsilon_{1}$-AFSD $F_{2}$ if and
only if $E_{F_{1}}[  u(X)]  \geq E_{F_{2}}[  u(X)]  $ for all
$u\in U\left(  \varepsilon_{1}\right)  $, where $E_{F}$ denotes the expectation for $X\sim F$. As pointed out by \citet{tsetlin2015generalized}, $\varepsilon_{1}$-AFSD is a
weaker version of first-degree stochastic dominance (FSD): $\varepsilon_{1}$-AFSD requires that the ratio of the
area between $F_{1}$ and $F_{2}$ for which $F_{1}  \geq
F_{2}  $ (i.e., $\int_{S_{1}\left(  F_{1},F_{2}\right)  }(
F_{1}\left(  x\right)  -F_{2}\left(  x\right)  )  \mathrm{d}x$) to the
total area between $F_{1}$ and $F_{2}$ (i.e., $\int_{a}^{b}\left\vert F_{1}\left(
x\right)  -F_{2}\left(  x\right)  \right\vert \mathrm{d}x$) is less than or equal to $\varepsilon_{1}$; FSD requires that this ratio is zero, i.e., $\int_{S_{1}\left(  F_{1},F_{2}\right)  }(  F_{1}\left(
x\right)  -F_{2}\left(  x\right)  )  \mathrm{d}x=0$. It is
interesting to figure out the smallest value of $\varepsilon_{1}$ such that \eqref{eq.AFSD} holds. If we denote this value by $c_{1}$, then it is easy to show (by proof similar to that of Lemma \ref{lemma.c properties}) that $E_{F_{1}}[  u(X)]  \geq
E_{F_{2}}[  u(X)]  $ for all $u\in\cup_{\varepsilon_{1}\in\lbrack
c_{1},1/2)}U_1\left(  \varepsilon_{1}\right)  $. We now consider general \emph{$\varepsilon_m$-almost $m$th-degree stochastic dominance}. We follow \citet{tsetlin2015generalized} and define 
\[
S_{m}\left(  F_{1},F_{2}\right)  =\left\{  x\in\left[  a,b\right]
:F_{1}^{\left(  m\right)  }\left(  x\right)  > F_{2}^{\left(  m\right)
}\left(  x\right)  \right\}  .
\]

\begin{definition}\label{def.AmSD}
For every $\varepsilon_{m}\in\left[  0,1/2\right)  $, $F_{1}$
$\varepsilon_{m}$-almost $m$th-degree stochastically dominates $F_{2}$ ($F_1$ $\varepsilon_m$-A$m$SD $F_2$) if
\begin{align}\label{eq.ASD inequality}
\int_{S_{m}\left(  F_{1},F_{2}\right)  }\left(  F_{1}^{\left(  m\right)
}\left(  x\right)  -F_{2}^{\left(  m\right)  }\left(  x\right)  \right)
\mathrm{d}x\leq\varepsilon_{m}\int_{a}^{b}\left\vert F_{1}^{\left(  m\right)
}\left(  x\right)  -F_{2}^{\left(  m\right)  }\left(  x\right)  \right\vert
\mathrm{d}x.
\end{align}
\end{definition}

For every $\varepsilon_{m}\in\left(  0,1\right)  $, define the utility function class
\[
U_{m}\left(  \varepsilon_{m}\right)  =\left\{
\begin{array}
[c]{c}
u:\left(  -1\right)  ^{k+1}u^{\left(  k\right)  }>0,k\in\{1,\ldots,m\},\\
\sup_x\left\{  (-1)^{m+1}u^{\left(  m\right)  }\left(  x\right)  \right\}
\leq\inf_x\left\{  (-1)^{m+1}u^{\left(  m\right)  }\left(  x\right)  \right\}
\left(  \frac{1}{\varepsilon_{m}}-1\right)
\end{array}
\right\}  .
\]
Theorem 4 of \citet{tsetlin2015generalized} shows that $E_{F_1}[u(X)]\ge E_{F_2}[u(X)]$ for all $u\in U_m(\varepsilon_m)$ if and only if $F_1$ $\varepsilon_m$-A$m$SD $F_2$ and $F_{1}^{\left(  k\right)  }\left(  b\right)  \leq F_2^{\left(  k\right)
}\left(  b\right)  $ for all $k\in\left\{  2,\ldots,m\right\}  $. Similarly, we are interested in finding the smallest $\varepsilon_m$ such that \eqref{eq.ASD inequality} holds.
We propose the $m$th-degree stochastic dominance coefficient as follows for all $m\ge1$. 

\begin{definition}\label{def.SDC}
The \emph{$F_1$-$F_2$ $m$th-degree stochastic dominance coefficient ($F_1$-$F_2$ $m$SDC)}, denoted by $c_m(F_1,F_2)$, is defined as
\begin{align}
    &c_m(F_1,F_2)=\notag\\
    &\inf\left\{\varepsilon_m\in[0,1]:\int_{S_{m}\left(  F_{1},F_{2}\right)  }\left(  F_{1}^{\left(  m\right)
}\left(  x\right)  -F_{2}^{\left(  m\right)  }\left(  x\right)  \right)
\mathrm{d}x\leq\varepsilon_{m}\int_{a}^{b}\left\vert F_{1}^{\left(  m\right)
}\left(  x\right)  -F_{2}^{\left(  m\right)  }\left(  x\right)  \right\vert
\mathrm{d}x\right\}.
\end{align}
\end{definition}

We define the difference function
\[
\phi_m\left(  x\right)  =F^{(m)}_{1}\left(  x\right)  -F^{(m)}_{2}\left(  x\right),\quad x\in[a,b].
\]
The following lemma summarizes the properties of the SDC $c_m(F_1,F_2)$.

\begin{lemma}\label{lemma.cm properties}
The SDC $c_m(F_1,F_2)$ is the smallest $\varepsilon_m$ in $[0,1]$ such that \eqref{eq.ASD inequality} holds. With \eqref{eq.0timesinf}, it follows that
\begin{align}\label{eq.SDC}
    c_m(F_1,F_2)&=\frac{\int_{S_{m}\left(  F_{1},F_{2}\right)  }\left(  F_{1}^{\left(  m\right)
}\left(  x\right)  -F_{2}^{\left(  m\right)  }\left(  x\right)  \right)
\mathrm{d}x}{\int_{a}^{b}\left\vert F_{1}^{\left(  m\right)
}\left(  x\right)  -F_{2}^{\left(  m\right)  }\left(  x\right)  \right\vert
\mathrm{d}x}\notag\\
&=\frac{\int_{a}^{b}\max\left\{  \phi_m\left(  x\right)  ,0\right\}
\mathrm{d}x}{\int_{a}^{b}\max\left\{  \phi_m\left(  x\right)  ,0\right\}
\mathrm{d}x+\int_{a}^{b}\max\left\{  -\phi_m\left(  x\right)  ,0\right\}
\mathrm{d}x}.
\end{align}
In addition, if $c_m(F_1,F_2)\in(0,1]$, then $c_m(F_2,F_1)=1-c_m(F_1,F_2)$.
\end{lemma}
  
According to Lemma \ref{lemma.cm properties}, $F_1$ $\varepsilon_m$-A$m$SD $F_2$ for all $\varepsilon_m\in[c_m(F_1,F_2),1/2)$ if $c_m(F_1,F_2)<1/2$. On the other hand, $c_m(F_1,F_2)>1/2$ implies that $F_2$ $\varepsilon_m$-A$m$SD $F_1$ for all $\varepsilon_m\in[1-c_m(F_1,F_2),1/2)$. Thus, $c_m(F_1,F_2)$ presents the degree of almost stochastic dominance relationship between $F_1$ and $F_2$, and provides all $\varepsilon_m$ such that the $\varepsilon_m$-A$m$SD holds.

\begin{proposition}\label{prop.SDC properties}
If $c_m(F_1,F_2)\in(0,1/2)$ and $F_{1}^{\left(  k\right)  }\left(  b\right)  \leq F_2^{\left(  k\right)
}\left(  b\right)  $ for all $k\in\left\{  2,\ldots,m\right\}  $, it then follows that 
    $E_{F_1}[u(X)]\ge E_{F_2}[u(X)]$ \text{ for all } $u\in\cup_{\varepsilon_m\in[c_m(F_1,F_2),1/2)}U_m(\varepsilon_m)$.
\end{proposition}

Similar to Proposition \ref{prop.c vs Gini}, Proposition \ref{prop.SDC properties} shows the importance of SDC concerning utility functions. The smaller the SDC $c_m(F_1,F_2)$ is, the more utility functions show higher expected utility for the distribution $F_1$ compared to $F_2$. With the knowledge of $c_m(F_1,F_2)$, we can infer the relationship between $F_1$ and $F_2$ based on a class of utility functions.

\subsection{Estimation and Inference}
Next, we provide an approach of estimating $c_m(F_1,F_2)$ and conducting inference on it. 

\begin{assumption}\label{ass.distribution SD}
    Let random variables $X^1$ and $X^2$ have the joint CDF $F_{12}$ with marginal CDFs $F_1$ and $F_2$, respectively. The mean vector $E[(X^1,X^2)]$ is finite and $(X^1,X^2)$ has a finite covariance matrix.
\end{assumption}

Similar to Sections \ref{sec.LDC} and \ref{sec.ISDC}, 
we define maps $\mathcal{F}_{1}:\ell^{\infty}\left(  \left[  a,b\right]  \right)
\rightarrow\mathbb{R}$, $\mathcal{F}_{2}:\ell^{\infty}\left(  \left[
a,b\right]  \right)  \rightarrow\mathbb{R}$, and $\mathcal{F}:\ell^{\infty
}\left(  \left[  a,b\right]  \right)  \rightarrow\mathbb{R}$ by
\begin{align}\label{eq.F map SD}
&\mathcal{F}_{1}\left(  \psi\right)     =\int_{a}^{b}\max\left\{  \psi\left(x\right)  ,0\right\}  \mathrm{d}x,\quad \mathcal{F}_{2}\left(  \psi\right)
=\int_{a}^{b}\max\left\{  -\psi\left(  x\right)  ,0\right\}  \mathrm{d}x,\notag\\
&\text{and }\mathcal{F}\left(  \psi\right)     =\frac{\mathcal{F}_{1}\left(
\psi\right)  }{\mathcal{F}_{1}\left(  \psi\right)  +\mathcal{F}_{2}\left(
\psi\right)  },\quad\psi\in\ell^{\infty}\left(  \left[  a,b\right]  \right)  .
\end{align}
Given the sample as in Assumption \ref{ass.data}, let $\hat{F}_j$ ($j=1,2$) be defined as in \eqref{eq.empirical cdfs} for $x\in[a,b]$. For matched pairs, define the empirical joint CDF
\begin{align}\label{eq.empirical joint cdf}
	\hat{F}_{12}(x,x')=\frac{1}{n}\sum_{i=1}^{n}
	{1}(X_i^1\leq x,X_i^2\leq x'),\quad x,x'\in[a,b].
\end{align} For $j=1,2$, define
\[
\hat{F}_{j}^{\left(  m\right)  }\left(  x\right)  =\int_{a}^{x}\hat{F}%
_{j}^{\left(  m-1\right)  }\left(  t\right)  \mathrm{d}t=\int_{a}^{x}%
\cdots\int_{a}^{t_{3}}\int_{a}^{t_{2}}\hat{F}_{j}\left(  t_{1}\right)
\mathrm{d}t_{1}\mathrm{d}t_{2}\cdots\mathrm{d}t_{m-1},\quad m\geq2,
\]
with $\hat{F}_{j}^{\left(  1\right)  }\left(  x\right)  =\hat{F}_{j}\left(
x\right)  $ for all $x\in\left[  a,b\right]  $. Define
\begin{align}\label{eq.Im}
\mathcal{I}_1(f)=f \text{ and for }m\ge2, \mathcal{I}_{m}\left(  f\right)(x)  =\int_{a}^{x}\cdots\int_{a}^{t_{3}}\int
_{a}^{t_{2}}f\left(  t_{1}\right)  \mathrm{d}t_{1}\mathrm{d}t_{2}%
\cdots\mathrm{d}t_{m-1}, \quad x\in [a,b],
\end{align}
for every measurable function $f$.
For $m\ge1$, define the empirical version of $\phi_m$ by
\[
\hat{\phi}_{m}\left(  x\right)  =\hat{F}_{1}^{\left(  m\right)  }\left(
x\right)  -\hat{F}_{2}^{\left(  m\right)  }\left(  x\right)  ,\quad
x\in\left[  a,b\right]  .
\]
The estimator for $c_{m}\left(  F_{1},F_{2}\right)  $ can be constructed
by
\[
\hat{c}_{m}\left(  F_{1},F_{2}\right)  =\frac{\mathcal{F}_{1}(  \hat
{\phi}_{m})  }{\mathcal{F}_{1}(  \hat{\phi}_{m})
+\mathcal{F}_{2}(  \hat{\phi}_{m})  }=\mathcal{F}(  \hat{\phi
}_{m})  .
\]

For every $\psi\in\ell^{\infty}\left(  \left[  a,b\right]  \right)  $, define
\[
C_{0}\left(  \psi\right)  =\left\{  x\in\left[  a,b\right]  :\psi\left(
x\right)  =0\right\}  \text{ and }C_{+}\left(  \psi\right)  =\left\{
x\in\left[  a,b\right]  :\psi\left(  x\right)  >0\right\}  .
\]

\begin{lemma}\label{lemma.F asymptotic limit}
Under Assumptions \ref{ass.data} and \ref{ass.distribution SD},
\begin{align}
    \sqrt{T_{n}}\left\{  (\hat{F}_{1}-\hat{F}_{2})-(F_{1}-F_{2})\right\}
\leadsto\mathbb{G}_{F},
\end{align}
for some random element $\mathbb{G}_{F}$ such that for all $x,x^{\prime}\in[a,b]$,
\begin{align}\label{eq.GxGx'}
&E\left[  \mathbb{G}_{F}\left(  x\right)  \mathbb{G}_{F}(  x^{\prime
})  \right]  \notag\\
  =&\,\left(  1-\lambda\right)  F_{1}(  x\wedge x^{\prime})  -\left(
1-\lambda\right)  F_{1}\left(  x\right)  F_{1}(  x^{\prime})
+\lambda F_{2}(  x\wedge x^{\prime})
-\lambda F_{2}\left(  x\right)  F_{2}(  x^{\prime})\notag\\
  &-\sqrt{\lambda\left(  1-\lambda\right)  }F_{12}(  x,x^{\prime})+\sqrt{\lambda\left(  1-\lambda\right)  }F_{1}\left(  x\right)  F_{2}(
x^{\prime})
\notag\\
&-\sqrt{\lambda\left(  1-\lambda\right)  }F_{12}(  x^{\prime},x)+\sqrt{\lambda\left(  1-\lambda\right)  }F_{1}(  x^{\prime})
F_{2}\left(  x\right).
\end{align}
\end{lemma}

\begin{proposition}\label{prop.cm asymptotic limit}
Under Assumptions \ref{ass.data} and \ref{ass.distribution SD}, it follows that
\begin{align*}
\sqrt{T_{n}}(  \hat{\phi}_{m}-\phi_{m})    \leadsto\mathbb{G}_{m}
\end{align*}
for some random element $\mathbb{G}_m$ with 
\begin{align*}
    &Var(\mathbb{G}_m(x))=\\
    &\int_{a}^{x}\cdots\int_{a}^{t_{3}^{\prime}}\int_{a}^{t_{2}^{\prime}}\left(  \int_{a}^{x}\cdots\int_{a}^{t_{3}}\int_{a}^{t_{2}}E\left[
\mathbb{G}_{F}(  t_{1})  \mathbb{G}_{F}(  t_{1}^{\prime
})  \right]  \mathrm{d}t_{1}\mathrm{d}t_{2}\cdots\mathrm{d}t_{m-1}\right)  \mathrm{d}t_{1}^{\prime}\mathrm{d}t_{2}^{\prime}\cdots\mathrm{d}t_{m-1}^{\prime}
\end{align*}
for every $x\in[a,b]$.
Moreover, it follows that
\[
\sqrt{T_{n}}\left\{  \hat{c}_{m}\left(  F_{1},F_{2}\right)  -c_{m}\left(
F_{1},F_{2}\right)  \right\}  \leadsto\mathcal{F}_{\phi_{m}}^{\prime}\left(
\mathbb{G}_{m}\right),
\]
where for every $h\in\ell^{\infty}([a,b])$,
\begin{align*}
\mathcal{F}_{\phi_m}^{\prime}\left(  h\right)  =\frac{\mathcal{F}_{1\phi_m}^{\prime}\left(  h\right)  \mathcal{F}_{2}\left(
\phi_m\right)  -\mathcal{F}_{1}\left(  \phi_m\right)  \mathcal{F}_{2\phi_m}^{\prime
}\left(  h\right)  }{\left(  \mathcal{F}_{1}\left(  \phi_m\right)
+\mathcal{F}_{2}\left(  \phi_m\right)  \right)  ^{2}},
\end{align*}
\[
\mathcal{F}_{1\phi_m}^{\prime}\left(  h\right)  =\int_{C_{+}\left(  \phi_m\right)
}h\left(  x\right)  \mathrm{d}x+\int_{C_{0}\left(  \phi_m\right)  }\max\left\{
h\left(  x\right)  ,0\right\}  \mathrm{d}x,
\]
and
\[
\mathcal{F}_{2\phi_m}^{\prime}\left(  h\right)  =\int_{C_{+}\left(
-\phi_m\right)  }-h\left(  x\right)  \mathrm{d}x+\int_{C_{0}\left(
\phi_m\right)  }\max\left\{  -h\left(  x\right)  ,0\right\}  \mathrm{d}x.
\]

\end{proposition}

Let $\hat{\lambda}=n_1/(n_1+n_2)$.  With \eqref{eq.GxGx'}, for independent samples,
we estimate $E[\mathbb{G}_F(x)\mathbb{G}_F(x')]$ by 
\begin{align*}
    \hat{E}\left[  \mathbb{G}_{F}\left(  x\right)  \mathbb{G}_{F}(  x^{\prime
})  \right]  
  =&\,(  1-\hat{\lambda})\hat  F_{1}(  x\wedge x^{\prime})  -(
1-\hat{\lambda})\hat  F_{1}\left(  x\right)\hat  F_{1}(  x^{\prime})\\
 &+\hat{\lambda}\hat F_{2}(  x\wedge x^{\prime})
-\hat{\lambda}\hat F_{2}\left(  x\right) \hat F_{2}(  x^{\prime})
\end{align*}
for all $x,x'\in[a,b]$. For matched pairs, we estimate $E[\mathbb{G}_F(x)\mathbb{G}_F(x')]$ by 
\begin{align*}
    &\hat{E}\left[  \mathbb{G}_{F}\left(  x\right)  \mathbb{G}_{F}(  x^{\prime
})  \right]  \notag\\
  =&\,(  1-\hat{\lambda})\hat  F_{1}(  x\wedge x^{\prime})  -(
1-\hat{\lambda})\hat  F_{1}\left(  x\right)\hat  F_{1}(  x^{\prime})
+\hat{\lambda}\hat F_{2}(  x\wedge x^{\prime})
-\hat{\lambda}\hat F_{2}\left(  x\right) \hat F_{2}(  x^{\prime})  
  \notag\\
& -\sqrt{\hat{\lambda}(  1-\hat{\lambda})  }\hat F_{12}(  x,x^{\prime})+\sqrt{\hat{\lambda}(  1-\hat{\lambda})  }\hat F_{1}\left(  x\right)  \hat F_{2}(
x^{\prime}) \notag\\
  &-\sqrt{\hat{\lambda}(  1-\hat{\lambda})  }\hat F_{12}(  x^{\prime},x)
+\sqrt{\hat{\lambda}(  1-\hat{\lambda})  }\hat F_{1}(  x^{\prime})
\hat F_{2}\left(  x\right)
\end{align*}
for all $x,x'\in[a,b]$. We then estimate the variance $Var(\mathbb{G}_m(x))$ by
\begin{align*}
&\hat{\sigma}_1(x)^2=\hat{E}\left[
\mathbb{G}_{F}( x)^2  \right]  \text{ for every }x\in[a,b], \text{ and for }m\ge2,\\
    &\hat{\sigma}_m(x)^2=\\
&\int_{a}^{x}\cdots\int_{a}^{t_{3}^{\prime}}\int_{a}^{t_{2}^{\prime}}\left(  \int_{a}^{x}\cdots\int_{a}^{t_{3}}\int_{a}^{t_{2}}\hat{E}\left[
\mathbb{G}_{F}(  t_{1})  \mathbb{G}_{F}(  t_{1}^{\prime
})  \right]  \mathrm{d}t_{1}\mathrm{d}t_{2}\cdots\mathrm{d}t_{m-1}\right)  \mathrm{d}t_{1}^{\prime}\mathrm{d}t_{2}^{\prime}\cdots\mathrm{d}t_{m-1}^{\prime},
\end{align*}
for every $x\in[a,b]$.

\subsubsection{Bootstrap Confidence Interval}

Following Section \ref{sec.bootstrap CI Lorenz}, we construct the estimators of $C_{+}\left(  \phi_m\right)  $, $C_{+}\left(
-\phi_m\right)  $, and $C_{0}\left(  \phi_m\right)  $ by
\begin{align*}
&\widehat{C_{+}\left(  \phi_m\right)  } =\left\{  x\in\left[  a,b\right]
:\frac{\sqrt{T_{n}}\hat{\phi}_m\left(  x\right)  }{\xi_{0}\vee\hat{\sigma
}_m\left(  x\right)  }>t_{n}\right\}  ,\widehat{C_{+}\left(  -\phi_m\right)
}=\left\{  x\in\left[  a,b\right]  :\frac{\sqrt{T_{n}}\hat{\phi}_m\left(
x\right)  }{\xi_{0}\vee\hat{\sigma}_m\left(  x\right)  }<-t_{n}\right\}  ,\\
&\text{and }\widehat{C_{0}\left(  \phi_m\right)  } =\left\{  x\in\left[
a,b\right]  :\left\vert \frac{\sqrt{T_{n}}\hat{\phi}_m\left(  x\right)  }
{\xi_{0}\vee\hat{\sigma}_m\left(  x\right)  }\right\vert \leq t_{n}\right\}  ,
\end{align*}
where $t_{n}\rightarrow\infty$ and $t_{n}/\sqrt{T_{n}}\rightarrow0$ as
$n\rightarrow\infty$. We then construct the estimator of $\mathcal{F}_{1\phi_m
}^{\prime}$ and $\mathcal{F}_{2\phi_m}^{\prime}$ by%
\begin{align*}
&\mathcal{\hat{F}}_{1\phi_m}^{\prime}\left(  h\right)  =\int_{\widehat
{C_{+}\left(  \phi_m\right)  }}h\left(  p\right)  \mathrm{d}p+\int
_{\widehat{C_{0}\left(  \phi_m\right)  }}\max\left\{  h\left(  p\right)
,0\right\}  \mathrm{d}p\\
&\text{and } \mathcal{\hat{F}}_{2\phi_m}^{\prime}\left(  h\right)
=\int_{\widehat{C_{+}\left(  -\phi_m\right)  }}-h\left(  p\right)
\mathrm{d}p+\int_{\widehat{C_{0}\left(  \phi_m\right)  }}\max\left\{  -h\left(
p\right)  ,0\right\}  \mathrm{d}p
\end{align*}
for every $h$ $\in\ell^{\infty}\left(  \left[  0,1\right]  \right)  $. The
estimator of $\mathcal{F}_{\phi_m}^{\prime}$ is defined by
\[
\mathcal{\hat{F}}_{\phi_m}^{\prime}\left(  h\right)  =\frac{\mathcal{\hat{F}%
}_{1\phi_m}^{\prime}\left(  h\right)  \mathcal{F}_{2}(  \hat{\phi}_m)
-\mathcal{F}_{1}(  \hat{\phi}_m)  \mathcal{\hat{F}}_{2\phi_m}^{\prime
}\left(  h\right)  }{\left(  \mathcal{F}_{1}(  \hat{\phi}_m)
+\mathcal{F}_{2}(  \hat{\phi}_m)  \right)  ^{2}}
\]
for every $h$ $\in\ell^{\infty}\left(  \left[  0,1\right]  \right)  $.

As discussed in Section \ref{sec.bootstrap CI Lorenz}, for independent samples, we draw a bootstrap sample $\{\hat{X}_{i}^{j}%
\}_{i=1}^{n_j}$ identically and independently from $\{X_{i}^{j}\}_{i=1}^{n_j}$ for
$j=1,2$, where $\{\hat{X}_{i}^{1}\}_{i=1}^{n_1}$ is jointly independent of
$\{\hat{X}_{i}^{2}\}_{i=1}^{n_2}$. For matched pairs, we draw a bootstrap sample
$\{(\hat{X}_{i}^{1},\hat{X}_{i}^{2})\}_{i=1}^{n}$ identically and
independently from $\{(X_{i}^{1},X_{i}^{2})\}_{i=1}^{n}$. Let $\hat{F}_j^*$ be defined as in Section \ref{sec.bootstrap CI Lorenz} for $x\in[a,b]$.
Then we define the bootstrap version of $\phi_m$ by
\begin{align*}
	\hat{\phi}_m^*(x)=\mathcal{I}_m(\hat{F}_1^*-\hat{F}_2^*)(x),\quad x\in[a,b].
\end{align*}
We define the bootstrap approximation of $c_m(F_1,F_2)$ by 
\begin{align}
    \hat{c}_m^*(F_1,F_2)=\mathcal{\hat{F}}_{\phi_m}^{\prime}(\sqrt{T_n}(\hat{\phi}_m^*-\hat{\phi}_m)).
\end{align}
For every $\beta\in(0,1)$, let $c_{m,\beta}$ denote the $\beta$ quantile of the distribution of $\mathcal{F}'_{\phi_m}(\mathbb{G}_m)$.
We construct the bootstrap approximation of $c_{m,\beta}$ by 
\begin{align}
    \hat{c}_{m,\beta}=\inf\left\{c: \mathbb{P}\left(\hat{c}_m^*(F_1,F_2)\le c|\{X^1_i\}_{i=1}^{n_1},\{X^2_i\}_{i=1}^{n_2}\right)\ge \beta\right\}.
\end{align}
Empirically, we approximate $\hat{c}_{m,\beta}$ by computing the $\beta$ quantile of the $n_B$ independently generated $\hat{c}_m^*(F_1,F_2)$, where $n_B$ is chosen as large as is computationally convenient. 

For a nominal significance level $\alpha\in(0,1/2)$, we construct the $1-\alpha$ confidence interval by
\begin{align}
    \mathrm{CI}_{m,1-\alpha}=[\hat{c}_m(F_1,F_2)-T_n^{-1/2}\hat{c}_{m,1-\alpha/2},\hat{c}_m(F_1,F_2)-T_n^{-1/2}\hat{c}_{m,\alpha/2}].
\end{align}

\begin{proposition}\label{prop.confidence interval SDC}
Suppose that Assumption \ref{ass.data} holds. If $c_m(F_1,F_2)\in(0,1)$ and the CDF of $\mathcal{F}'_{\phi_m}(\mathbb{G}_m)$ is continuous and increasing at $c_{\alpha}$ and $c_{1-\alpha}$, then it follows that
\begin{align}
    \lim_{n\rightarrow\infty}\mathbb{P}(c_m(F_1,F_2)\in\mathrm{CI}_{m,1-\alpha})= 1-\alpha.
\end{align}

\end{proposition}

\section{Additional Simulation Evidence}\label{sec.additional simulation}

In this section, we provide additional simulation results for ISDCs and SDCs. 

\subsection{Inverse Stochastic Dominance Coefficient}

We consider simulations for $3$UISDC. In each iteration of the simulations for $3$UISDC, the data $\{X_i^1\}_{i=1}^{n_1}$ are generated independently
	from $\mathrm{dP}(2.1,1.5)$, and the data $\{X_i^2\}_{i=1}^{n_2}$ are generated independently
	from $\mathrm{dP}(200,\beta)$, whose law is
	parametrized by $\beta$. We let $\beta\in\{2.2,2.3,2.4,2.5\}$. Figure \ref{fig:DGPs UISDC} displays the curves $\Lambda_1^3$ and $\Lambda_{2(\beta)}^3$ corresponding to the above DGPs. The $3$UISDC $c_3^u(F_1,F_2)=$ $0.06229$, $0.14052$, $0.26581$, and $0.42840$, respectively, for $\beta=$ $2.2$, $2.3$, $2.4$, and $2.5$. We choose the tuning parameter from $t_n\in (0,20]$. For independent samples, we let $$(n_1,n_2)\in\{(200,200),(200,500),(200,1000),(1000,2000),(10000,10000)\}.$$
For matched pairs, we let $$(n_1,n_2)\in\{(200,200),(500,500),(1000,1000),(2000,2000),(10000,10000)\}.$$

Tables \ref{tab:CR UISDC IS} and \ref{tab:CR UISDC MP} show the simulation results for independent samples and matched pairs. For all the DGPs, as $n_1$ and $n_2$ increase, Mean gets close to $c_3^u(F_1,F_2)$, Bias decreases to $0$, and both SE and RMSE decrease; under appropriate choices of $t_n$, CR approaches $95\%$.

\begin{figure} [h]
\caption{$\Lambda_j^3$ Curves for Four DGPs}
\label{fig:DGPs UISDC}
\centering
\begin{subfigure}[b]{0.45\textwidth}
	\centering
\scalebox{0.2}{
\includegraphics{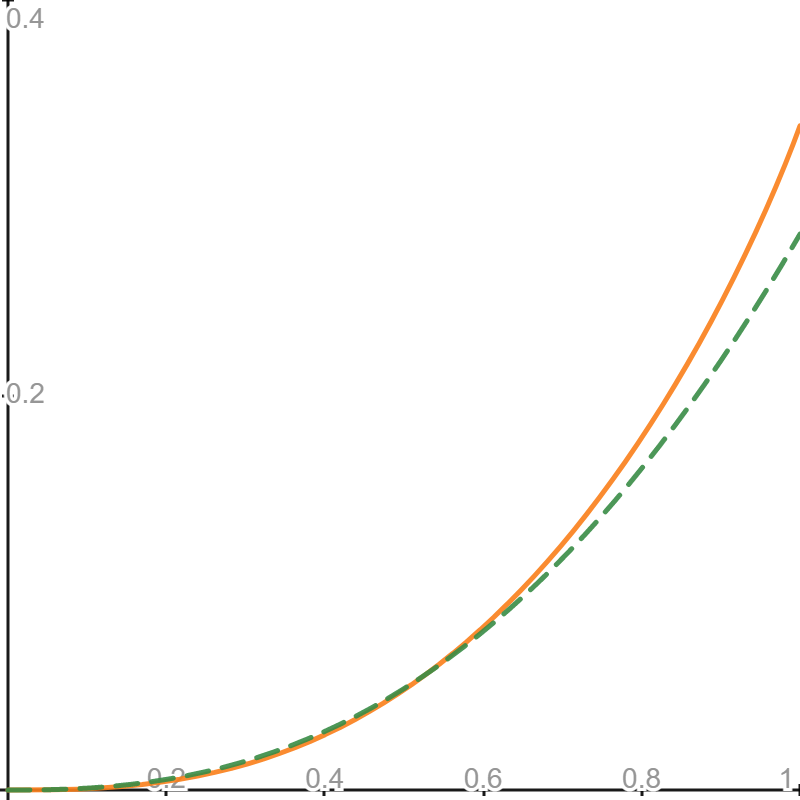}
}
\subcaption{$\Lambda_1^3$ (solid) and $\Lambda^3_{2(2.2)}$ (dashed)}

\end{subfigure}
\begin{subfigure}[b]{0.45\textwidth}
	\centering
\scalebox{0.2}{
\includegraphics{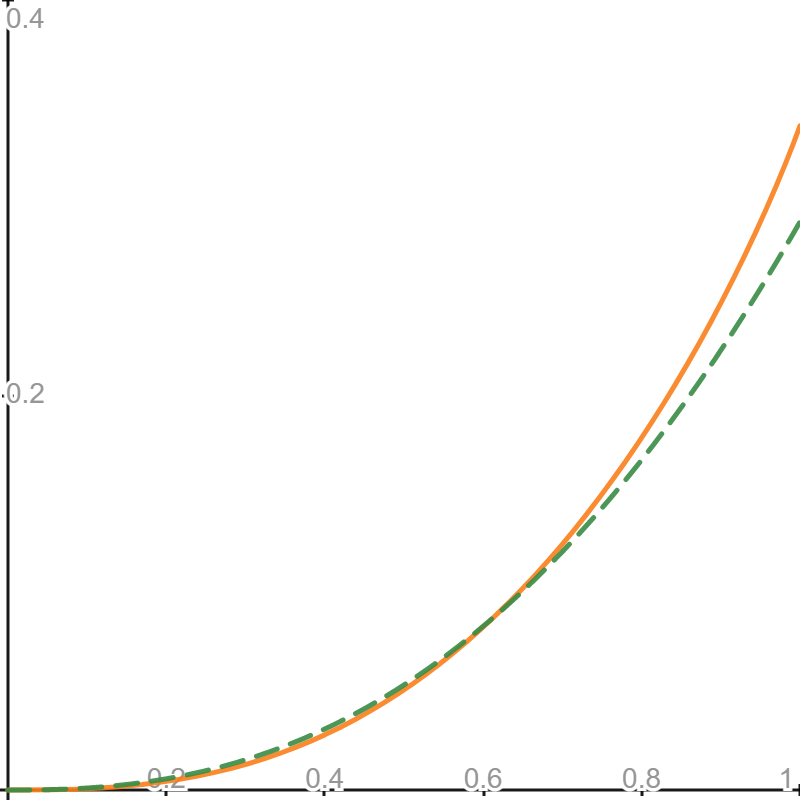}
}
\subcaption{$\Lambda_1^3$ (solid) and $\Lambda_{2(2.3)}^3$ (dashed)}

\end{subfigure}
\begin{subfigure}[b]{0.45\textwidth}
	\centering
\scalebox{0.2}{
\includegraphics{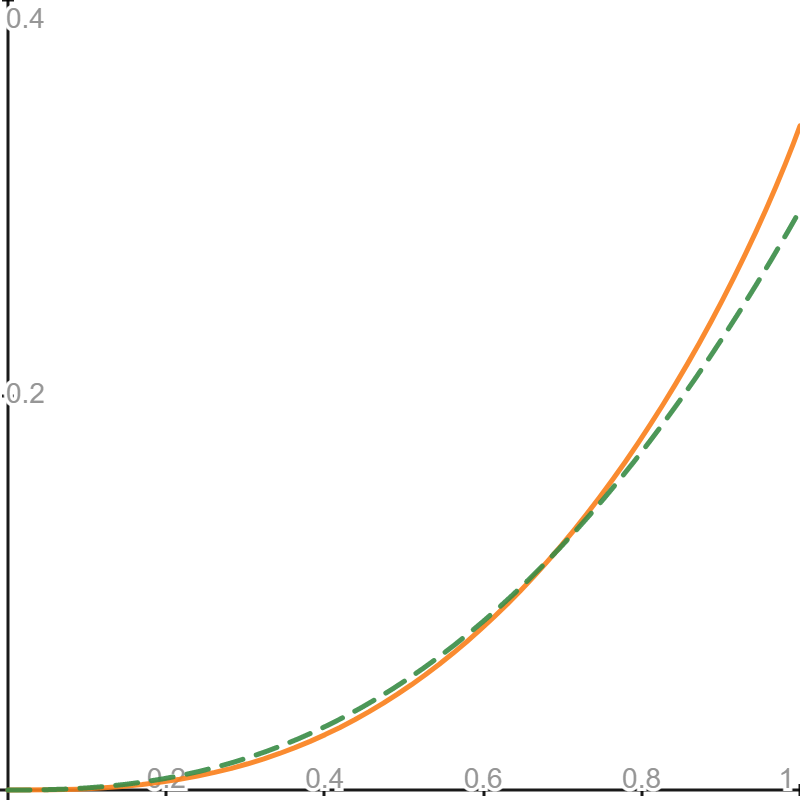}
}
\subcaption{$\Lambda_1^3$ (solid) and $\Lambda_{2(2.4)}^3$ (dashed)}

\end{subfigure}
\begin{subfigure}[b]{0.45\textwidth}
	\centering
\scalebox{0.2}{
\includegraphics{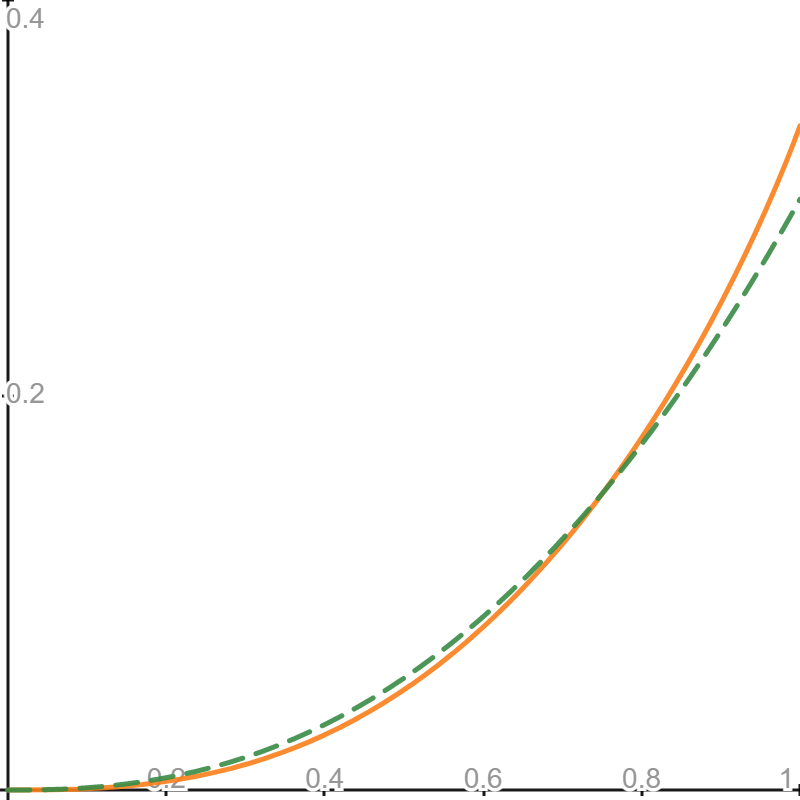}
}
\subcaption{$\Lambda_1^3$ (solid) and $\Lambda_{2(2.5)}^3$ (dashed)}

\end{subfigure}
\end{figure}

\begin{table}[h]
\centering
	\caption{Estimations and Coverage Rates for $3$UISDC (Independent Samples)}
		
		\scalebox{0.95}{
			\begin{tabular}{cccccccccc}
				\hline\hline
			{DGP} & 
			 $c_3^u(F_1,F_2)$	&	{$n_1 $}	& {$n_2 $}		 & Mean & Bias & SE & RMSE & $t_n$ & CR  \\
			 \hline
			 \multirow{5}{*}{(a)}	 & \multirow{5}{*}{0.06229} 
		               & 200 & 200   & 0.1736 & 0.1113 & 0.2474 & 0.2713 & 0.001 & 0.6980 \\
			 &         & 200 & 500   & 0.1339 & 0.0716 & 0.2040 & 0.2162 & 0.001 & 0.7190 \\
			 &         & 200 & 1000  & 0.1127 & 0.0504 & 0.1797 & 0.1866 & 0.001 & 0.6960 \\
			 &         & 1000 & 2000 & 0.0777 & 0.0154 & 0.0762 & 0.0778 & 0.001 & 0.8430 \\
			 &         & 10000& 10000& 0.0633 & 0.0010 & 0.0219 & 0.0219 & 0.001 & 0.9220\\	 
			 \hline
			 \multirow{5}{*}{(b)}	 & \multirow{5}{*}{0.14052} 
		               & 200 & 200   & 0.2546 & 0.1141 & 0.2916 & 0.3131 & 0.001 & 0.6850 \\
			 &         & 200 & 500   & 0.2130 & 0.0725 & 0.2566 & 0.2666 & 0.001 & 0.7210 \\
			 &         & 200 & 1000  & 0.1882 & 0.0476 & 0.2351 & 0.2398 & 0.001 & 0.7080 \\
			 &         & 1000 & 2000 & 0.1577 & 0.0172 & 0.1206 & 0.1218 & 0.001 & 0.8510 \\
			 &         & 10000& 10000& 0.1404 & -0.0001 & 0.0406 & 0.0406 & 0.001 & 0.9190\\	 	 
			 \hline
			 \multirow{5}{*}{(c)}	 & \multirow{5}{*}{0.26581} 
		               & 200 & 200   & 0.3461 & 0.0803 & 0.3231 & 0.3329 & 0.001 & 0.6680 \\
			 &         & 200 & 500   & 0.3077 & 0.0419 & 0.2961 & 0.2990 & 0.001 & 0.7140 \\
			 &         & 200 & 1000  & 0.2819 & 0.0161 & 0.2801 & 0.2806 & 0.001 & 0.7030 \\
			 &         & 1000 & 2000 & 0.2735 & 0.0076 & 0.1630 & 0.1631 & 0.001 & 0.8620 \\
			 &         & 10000& 10000& 0.2617 & -0.0041 & 0.0616 & 0.0618 & 0.001 & 0.9200\\	 
			 \hline
			 \multirow{5}{*}{(d)}	 & \multirow{5}{*}{0.42840} 
		               & 200 & 200   & 0.4418 & 0.0134 & 0.3389 & 0.3392 & 0.001 & 0.6690 \\
			 &         & 200 & 500   & 0.4117 & -0.0167 & 0.3174 & 0.3178 & 0.001 & 0.7010 \\
			 &         & 200 & 1000  & 0.3868 & -0.0416 & 0.3086 & 0.3114 & 0.001 & 0.6970 \\
			 &         & 1000 & 2000 & 0.4139 & -0.0145 & 0.1899 & 0.1905 & 0.001 & 0.8660 \\
			 &         & 10000& 10000& 0.4177 & -0.0107 & 0.0761 & 0.0768 & 0.001 & 0.9210\\	 
				\hline\hline                         	
			\end{tabular}
		}

		\label{tab:CR UISDC IS}

\end{table}

\begin{table}[h]
\centering
	\caption{Estimations and Coverage Rates for $3$UISDC (Matched Pairs)}
		
		\scalebox{0.95}{
			\begin{tabular}{cccccccccc}
				\hline\hline
			{DGP} & 
			 $c_3^u(F_1,F_2)$	&	{$n_1 $}	& {$n_2 $}		 & Mean & Bias & SE & RMSE & $t_n$ & CR  \\
			 \hline
			 \multirow{5}{*}{(a)}	 & \multirow{5}{*}{0.06229} 
		               & 200 & 200 & 0.1212 & 0.0589 & 0.1711 & 0.1809 & 0.001 & 0.7690
\\ 
        &&500 & 500 & 0.1386 & 0.0764 & 0.1308 & 0.1515 & 0.001 & 0.9320
\\ 
        &&1000 & 1000 & 0.0912 & 0.0289 & 0.0732 & 0.0787 & 0.001 & 0.9430
 \\ 
        &&2000 & 2000 & 0.0709 & 0.0086 & 0.0422 & 0.0431 & 0.001
 & 0.9350
\\ 
        &&10000 & 10000 & 0.0642 & 0.0019 & 0.0166 & 0.0167 & 2.4 & 0.9510

\\  
			 \hline
			 \multirow{5}{*}{(b)}	 & \multirow{5}{*}{0.14052} 
		               & 200 & 200 & 0.1998 & 0.0593 & 0.2236 & 0.2313 & 0.001
 & 0.7930\\
        &&500 & 500 & 0.2453 & 0.1048 & 0.1785 & 0.2070 & 0.001
 & 0.9400
\\ 
        &&1000 & 1000 & 0.1863 & 0.0458 & 0.1180 & 0.1266 & 0.001
 & 0.9490
\\ 
        &&2000 & 2000 & 0.1519 & 0.0114 & 0.0729 & 0.0738 & 0.001 & 0.9510

\\ 
        &&10000 & 10000 & 0.1423 & 0.0018 & 0.0305 & 0.0306 & 3.7
 & 0.9510
\\  	 
			 \hline
			 \multirow{5}{*}{(c)}	 & \multirow{5}{*}{0.26581} 
		               & 200 & 200 & 0.2972 & 0.0314 & 0.2656 & 0.2674 & 0.001
& 0.8020\\ 
        &&500 & 500 & 0.3783 & 0.1125 & 0.2106 & 0.2388 & 0.001
& 0.9180
\\ 
        &&1000 & 1000 & 0.3221 & 0.0563 & 0.1576 & 0.1674 & 0.001
& 0.9480
\\ 
        &&2000 & 2000 & 0.2756 & 0.0098 & 0.1046 & 0.1050 & 0.601
& 0.9490
\\ 
        &&10000 & 10000 & 0.2655 & -0.0004 & 0.0462 & 0.0462 & 6.7
& 0.9500
\\  
			 \hline
			 \multirow{5}{*}{(d)}	 & \multirow{5}{*}{0.42840} 
		               & 200 & 200 & 0.4061 & -0.0223 & 0.2900 & 0.2909 & 0.001
 & 0.8030 
\\ 
        &&500 & 500 & 0.5203 & 0.0919 & 0.2186 & 0.2371 & 0.001
 & 0.8990 
\\ 
        &&1000 & 1000 & 0.4810 & 0.0526 & 0.1761 & 0.1838 & 0.001 & 0.9320 
\\ 
        &&2000 & 2000 & 0.4305 & 0.0021 & 0.1244 & 0.1244 & 1.5 & 0.9500 \\ 
        &&10000 & 10000 & 0.4236 & -0.0048 & 0.0567 & 0.0569 & 15.6 & 0.9550
\\  
				\hline\hline                         	
			\end{tabular}
		}

		\label{tab:CR UISDC MP}

\end{table}

\subsection{Stochastic Dominance Coefficient}
We let $X^1$ and $X^2$ be random variables such that
\begin{align*}
X^1=\begin{cases}0.25&\text{ with probability }1/\beta,\\
1&\text{ with probability }1-1/\beta,\end{cases}
\text{ and }
X^2=\begin{cases}0.5&\text{ with probability }2/3,\\
0.75&\text{ with probability }1/3,\end{cases}
\end{align*}
with $\beta\in\{2,4,6,8\}$.
Figure \ref{fig:DGPs ASD} displays the CDF curves $F_{1(\beta)}$ and $F_{2}$ corresponding to the above DGPs. The SDC $c_1(F_1,F_2)=$ $0.081081$, $0.11111$, $0.17647$, and $0.42857$, respectively for $\beta=$ $8$, $6$, $4$, and $2$. 
We choose the tuning parameter from $t_n\in (0,20]$. For independent samples, we let $$(n_1,n_2)\in\{(100,100),(100,200),(100,500),(200,500),(1000,1000)\}.$$
For matched pairs, we let $$(n_1,n_2)\in\{(100,100),(200,200),(300,300),(500,500),(1000,1000)\}.$$
Tables \ref{tab:CR SDC IS} and \ref{tab:CR SDC MP} show the simulation results for independent samples and matched pairs, which are similar to those in Tables \ref{tab:CR IS} and \ref{tab:CR MP}. For all DGPs, as $n_1$ and $n_2$ increase, Mean gets close to $c_1(F_1,F_2)$, Bias decreases to $0$, and both SE and RMSE decrease; under appropriate choices of $t_n$, CR approaches $95\%$. 

\begin{figure} [h]
\caption{Distribution Curves for Four DGPs}
\label{fig:DGPs ASD}
\centering
\begin{subfigure}[b]{0.45\textwidth}
	\centering
\scalebox{0.2}{
\includegraphics{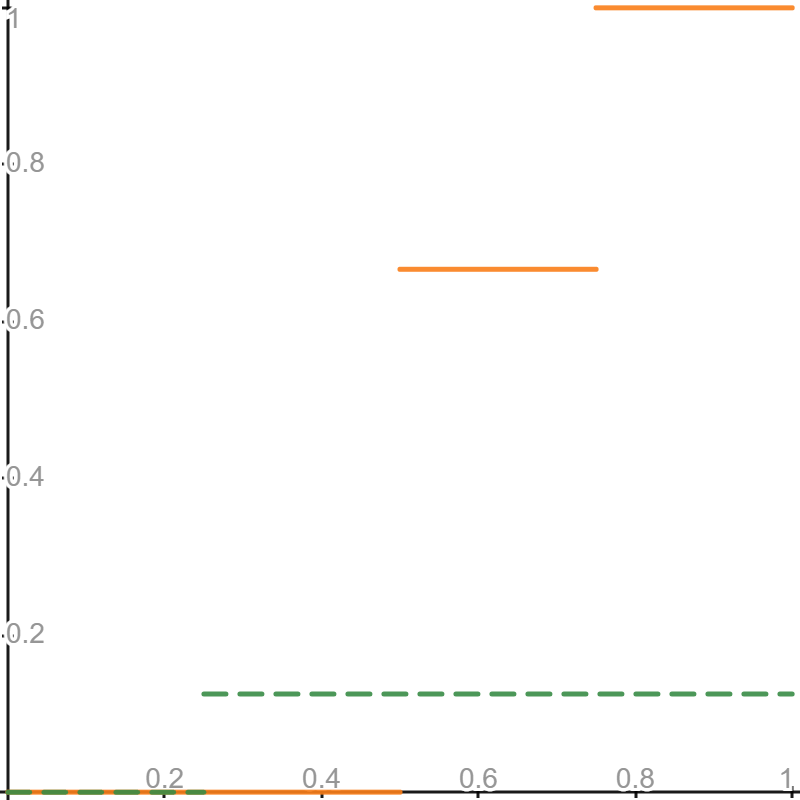}
}
\subcaption{$F_{1(8)}$ (dashed) and $F_{2}$ (solid)}

\end{subfigure}
\begin{subfigure}[b]{0.45\textwidth}
	\centering
\scalebox{0.2}{
\includegraphics{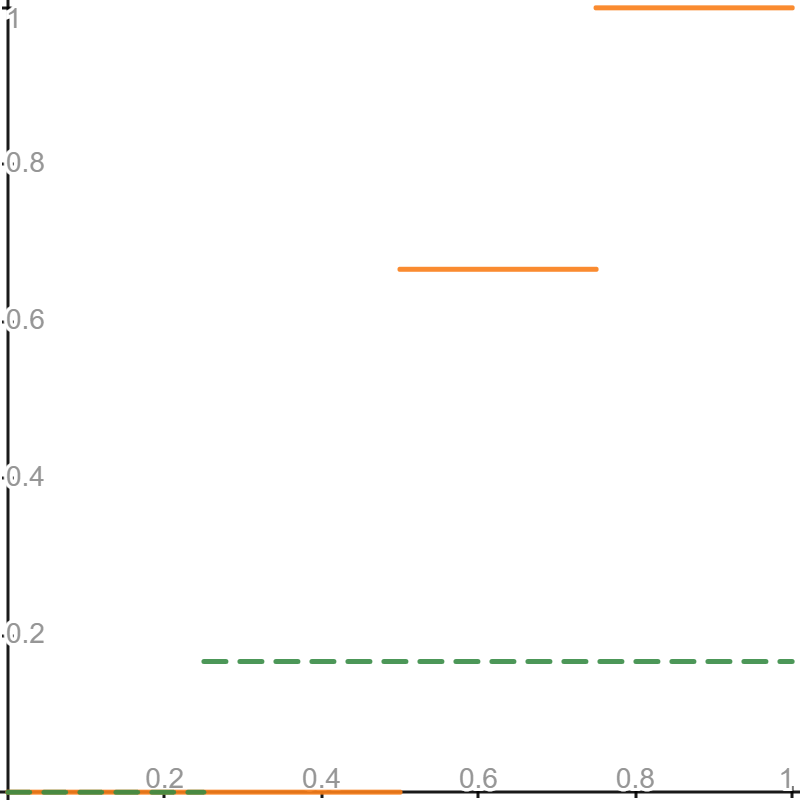}
}
\subcaption{$F_{1(6)}$ (dashed) and $F_{2}$ (solid)}

\end{subfigure}
\begin{subfigure}[b]{0.45\textwidth}
	\centering
\scalebox{0.2}{
\includegraphics{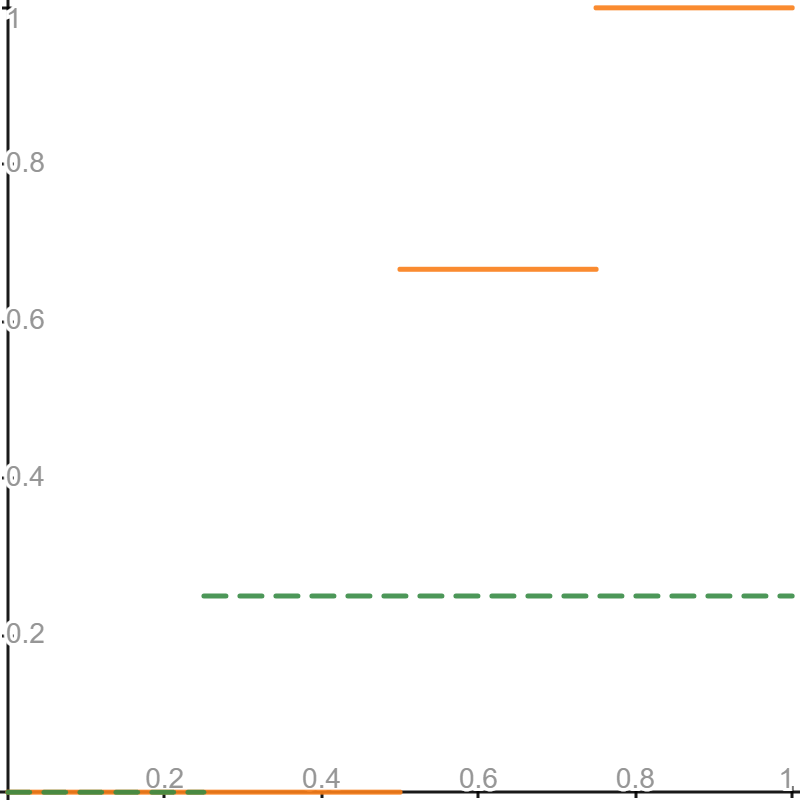}
}
\subcaption{$F_{1(4)}$ (dashed) and $F_{2}$ (solid)}

\end{subfigure}
\begin{subfigure}[b]{0.45\textwidth}
	\centering
\scalebox{0.2}{
\includegraphics{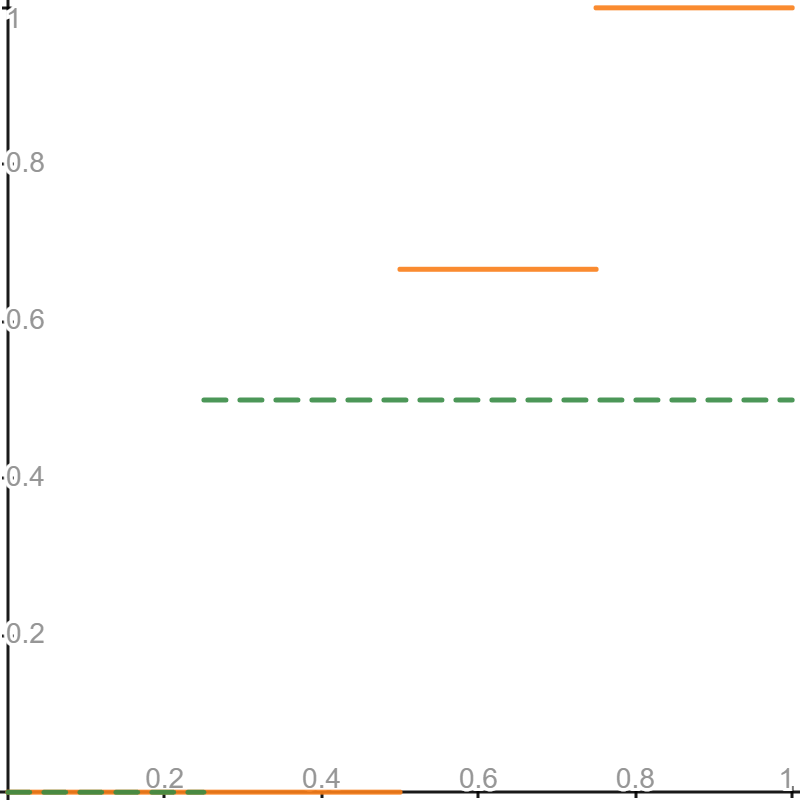}
}
\subcaption{$F_{1(2)}$ (dashed) and $F_{2}$ (solid)}

\end{subfigure}
\end{figure}

\begin{table}[h]
\centering
	\caption{Estimations and Coverage Rates for SDC (Independent Samples)}
		
		\scalebox{0.95}{
			\begin{tabular}{cccccccccc}
				\hline\hline
			{DGP} & 
			 $c_1(F_1,F_2)$	&	{$n_1 $}	& {$n_2 $}		 & Mean & Bias & SE & RMSE & $t_n$ & CR  \\
			 \hline
			 \multirow{5}{*}{(a)}	 & \multirow{5}{*}{0.08108} 
		               & 100 & 100 & 0.0826 & 0.0016 & 0.0241 & 0.0241 & 0.001 & 0.9380 \\
			 &         & 100 & 200 & 0.0823 & 0.0012 & 0.0236 & 0.0236 & 0.001 & 0.9350 \\
			 &         & 100 & 500 & 0.0812 & 0.0002 & 0.0229 & 0.0229 & 0.001 & 0.9330 \\
			 &         & 200 & 500 & 0.0819 & 0.0009 & 0.0171 & 0.0171 & 0.001 & 0.9270 \\
			 &         & 1000& 1000& 0.0812 & 0.0001 & 0.0073 & 0.0073 & 0.001 & 0.9490\\	 
			 \hline
			 \multirow{5}{*}{(b)}	 & \multirow{5}{*}{0.11111} 
		               & 100 & 100 & 0.1130 & 0.0019 & 0.0290 & 0.0290 & 0.001 & 0.9410 \\
			 &         & 100 & 200 & 0.1114 & 0.0003 & 0.0276 & 0.0276 & 0.001 & 0.9450 \\
			 &         & 100 & 500 & 0.1119 & 0.0008 & 0.0269 & 0.0269 & 0.001 & 0.9540 \\
			 &         & 200 & 500 & 0.1115 & 0.0004 & 0.0196 & 0.0196 & 0.001 & 0.9440 \\
			 &         & 1000& 1000& 0.1111 & 0.0000 & 0.0086 & 0.0086 & 0.001 & 0.9500\\	 
			 \hline
		\multirow{5}{*}{(c)}	 & \multirow{5}{*}{0.17647} 
		               & 100 & 100 & 0.1804 & 0.0039 & 0.0386 & 0.0388 & 2.6 & 0.9470 \\
			 &         & 100 & 200 & 0.1781 & 0.0016 & 0.0370 & 0.0371 & 3.0 & 0.9390 \\
			 &         & 100 & 500 & 0.1766 & 0.0001 & 0.0361 & 0.0361 & 3.0 & 0.9480 \\
			 &         & 200 & 500 & 0.1774 & 0.0010 & 0.0251 & 0.0251 & 4.9 & 0.9530 \\
			 &         & 1000& 1000& 0.1767 & 0.0002 & 0.0113 & 0.0113 & 12.0 & 0.9510\\	 
			 \hline
		\multirow{5}{*}{(d)}	 & \multirow{5}{*}{0.42857} 
		               & 100 & 100 & 0.4326 & 0.0040 & 0.0652 & 0.0653 & 0.001 & 0.9470 \\
			 &         & 100 & 200 & 0.4355 & 0.0069 & 0.0657 & 0.0660 & 0.001 & 0.9390 \\
			 &         & 100 & 500 & 0.4317 & 0.0032 & 0.0634 & 0.0635 & 0.001 & 0.9410 \\
			 &         & 200 & 500 & 0.4302 & 0.0016 & 0.0443 & 0.0443 & 0.001 & 0.9450 \\
			 &         & 1000& 1000& 0.4284 & -0.0002& 0.0209 & 0.0209 & 5.0 & 0.9500\\
				\hline\hline                         	
			\end{tabular}
		}

		\label{tab:CR SDC IS}

\end{table}

\begin{table}[h]
\centering
	\caption{Estimations and Coverage Rates for SDC (Matched Pairs)}
		
		\scalebox{0.95}{
			\begin{tabular}{ccccccccccccc}
				\hline\hline
			{DGP} & 
			 $c_1(F_1,F_2)$	&	{$n_1 $}	& {$n_2 $}		 & Mean & Bias & SE & RMSE & $t_n$ & CR  \\
			 \hline
			 \multirow{5}{*}{(a)}	 & \multirow{5}{*}{0.08108} 
		               & 100 & 100 & 0.0824 & 0.0013 & 0.0233 & 0.0234 & 0.001 & 0.9450 \\
			 &         & 200 & 200 & 0.0817 & 0.0006 & 0.0164 & 0.0164 & 0.001 & 0.9450 \\
			 &         & 300 & 300 & 0.0816 & 0.0006 & 0.0138 & 0.0138 & 0.001 & 0.9410 \\
			 &         & 500 & 500 & 0.0815 & 0.0004 & 0.0101 & 0.0101 & 7.3 & 0.9520 \\
			 &         & 1000& 1000& 0.0815 & 0.0004 & 0.0077 & 0.0077 & 7.3 & 0.9390\\	 
			 \hline
			 \multirow{5}{*}{(b)}	 & \multirow{5}{*}{0.11111} 
		              & 100 & 100 & 0.1126 & 0.0015 & 0.0278 & 0.0278 & 3.2 & 0.9530 \\
			 &         & 200 & 200 & 0.1113 & 0.0002 & 0.0190 & 0.0190 & 3.2 & 0.9470 \\
			 &         & 300 & 300 & 0.1111 & 0.0000 & 0.0158 & 0.0158 & 6.4 & 0.9510 \\
			 &         & 500 & 500 & 0.1116 & 0.0005 & 0.0114 & 0.0115 & 9.0 & 0.9510 \\
			 &         & 1000& 1000& 0.1117 & 0.0006 & 0.0089 & 0.0089 & 9.0 & 0.9490\\	 
			 \hline
		\multirow{5}{*}{(c)}	 & \multirow{5}{*}{0.17647} 
		              & 100 & 100 & 0.1796 & 0.0031 & 0.0359 & 0.0361 & 4.6 & 0.9480 \\
			 &         & 200 & 200 & 0.1768 & 0.0003 & 0.0248 & 0.0248 & 6.9 & 0.9490 \\
			 &         & 300 & 300 & 0.1766 & 0.0001 & 0.0203 & 0.0203 & 8.7 & 0.9530 \\
			 &         & 500 & 500 & 0.1768 & 0.0003 & 0.0152 & 0.0152 & 11.7 & 0.9510 \\
			 &         & 1000& 1000& 0.1769 & 0.0004 & 0.0110 & 0.0110 & 16.9 & 0.9500\\	 
			 \hline
		\multirow{5}{*}{(d)}	 & \multirow{5}{*}{0.42857} 
		               & 100 & 100 & 0.4298 & 0.0012 & 0.0519 & 0.0519 & 4.9 & 0.9490 \\
			 &         & 200 & 200 & 0.4299 & 0.0013 & 0.0371 & 0.0371 & 6.5 & 0.9550 \\
			 &         & 300 & 300 & 0.4294 & 0.0009 & 0.0292 & 0.0292 & 7.9 & 0.9460 \\
			 &         & 500 & 500 & 0.4289 & 0.0004 & 0.0229 & 0.0229 & 10.0 & 0.9480 \\
			 &         & 1000& 1000& 0.4299 & 0.0014 & 0.0156 & 0.0157 & 13.9 & 0.9510\\
				\hline\hline                         	
			\end{tabular}
		}

		\label{tab:CR SDC MP}

\end{table}

\subsection{Tuning Parameter Selection}\label{sec.tuning parameter selection}
In the above simulations, we compute the confidence intervals for all $t_n$ from some prespecified set and display the values that yield the best results. We now propose an empirical way of selecting $t_n$ in practice. Let $t_n$ be selected from a set $S_t$ that is sufficiently large. Suppose that we observe the data $\{X_i^1\}_{i=1}^{n_1}$ and $\{X_i^2\}_{i=1}^{n_2}$. We then take the empirical distributions of $\{X_i^1\}_{i=1}^{n_1}$ and $\{X_i^2\}_{i=1}^{n_2}$ as the DGP to generate the data in the simulations and compute the dominance coefficient $c$ based on the empirical distributions of $\{X_i^1\}_{i=1}^{n_1}$ and $\{X_i^2\}_{i=1}^{n_2}$. We take this $c$ as the true value of the coefficient we are interested in. 
Then we follow the previous simulation procedure to compute the bootstrap confidence intervals for every value of $t_n$. Finally, we select the value $t_n$ that yields the coverage rate closest to $1-\alpha$, and use this value to construct the bootstrap confidence intervals in the application.

\section{Proofs} \label{sec.proofs}

\subsection{Proofs for Section \ref{sec.LDC}}

\begin{proof}[Proof of Lemma \ref{lemma.ald equivalence}]
    If \eqref{eq.ALD2} holds with some $c\in[0,\varepsilon]$ for some $\varepsilon\in[0,1/2)$, then clearly $F_1$ $\varepsilon$-ALD $F_2$ by definition. If $F_1$ $\varepsilon$-ALD $F_2$ for some $\varepsilon\in[0,1/2)$, then we can find $$c=\frac{\int_{S(F_{1},F_{2})}\left(  L_{2}\left(  p\right)  -L_{1}\left(  p\right)
\right)  \mathrm{d}p}{\int_{0}^{1}\left\vert
L_{2}\left(  p\right)  -L_{1}\left(  p\right)  \right\vert \mathrm{d}p}$$
that satisfies \eqref{eq.ALD2}. 
\end{proof}

\begin{proof}[Proof of Lemma \ref{lemma.c properties}]
Suppose that
\[
\int_{S(F_{1},F_{2})}\left(  L_{2}\left(  p\right)  -L_{1}\left(  p\right)
\right)  \mathrm{d}p>c\left(  L_{1},L_{2}\right)  \int_{0}^{1}\left\vert
L_{2}\left(  p\right)  -L_{1}\left(  p\right)  \right\vert \mathrm{d}p,
\]
then there is some $\delta>0$ such that
\[
\int_{S(F_{1},F_{2})}\left(  L_{2}\left(  p\right)  -L_{1}\left(  p\right)
\right)  \mathrm{d}p=\left\{  c\left(  L_{1},L_{2}\right)  +\delta\right\}
\int_{0}^{1}\left\vert L_{2}\left(  p\right)  -L_{1}\left(  p\right)
\right\vert \mathrm{d}p.
\]
For all $\varepsilon\in\left[  0,1\right]$ with
\[
\int_{S(F_{1},F_{2})}\left(  L_{2}\left(  p\right)  -L_{1}\left(  p\right)
\right)  \mathrm{d}p\leq\varepsilon\int_{0}^{1}\left\vert L_{2}\left(
p\right)  -L_{1}\left(  p\right)  \right\vert \mathrm{d}p  ,
\]
we have that $\varepsilon\geq c\left(  L_{1},L_{2}\right)  +\delta$.
It then follows that
\[
c\left(  L_{1},L_{2}\right)  <c\left(  L_{1},L_{2}\right)  +\frac{\delta}%
{2}<\varepsilon
\]
for all $\varepsilon$ such that $\int_{S(F_{1},F_{2})}\left(  L_{2}\left(
p\right)  -L_{1}\left(  p\right)  \right)  \mathrm{d}p\leq\varepsilon\int
_{0}^{1}\left\vert L_{2}\left(  p\right)  -L_{1}\left(  p\right)  \right\vert
\mathrm{d}p$, which contradicts the definition of $c\left(  L_{1}%
,L_{2}\right)  $. Thus, we have
\[
\int_{S(F_{1},F_{2})}\left(  L_{2}\left(  p\right)  -L_{1}\left(  p\right)
\right)  \mathrm{d}p\leq c\left(  L_{1},L_{2}\right)  \int_{0}^{1}\left\vert
L_{2}\left(  p\right)  -L_{1}\left(  p\right)  \right\vert \mathrm{d}p.
\]
If there is some other $c$ such that $c<c\left(  L_{1},L_{2}\right)  $ and
$$\int_{S(F_{1},F_{2})}\left(  L_{2}\left(  p\right)  -L_{1}\left(  p\right)
\right)  \mathrm{d}p\leq c\int_{0}^{1}\left\vert L_{2}\left(  p\right)
-L_{1}\left(  p\right)  \right\vert \mathrm{d}p,$$ then the definition of
$c\left(  L_{1},L_{2}\right)  $ is contradicted. Thus, $c(L_1,L_2)$ is the smallest $\varepsilon$ such that \eqref{eq.ALD inequality} holds. If $c\left(  L_{1},L_{2}\right)  =0$, then $\int_{S(F_{1},F_{2})}\left(
L_{2}\left(  p\right)  -L_{1}\left(  p\right)  \right)  \mathrm{d}p=0$. By \eqref{eq.0timesinf}, \eqref{eq.LDC} holds. If $c\left(  L_{1},L_{2}\right)  \neq0$, then $\int_{S(F_{1},F_{2})}\left(
L_{2}\left(  p\right)  -L_{1}\left(  p\right)  \right)  \mathrm{d}p\neq0$ and clearly $\int_0^1|L_2(p)-L_1(p)|\mathrm{d}p\neq0$. Then \eqref{eq.LDC} holds since $c(L_1,L_2)$ is the smallest $\varepsilon$ such that \eqref{eq.ALD inequality} holds.

If $c\left(  L_{1},L_{2}\right)  =0$, then $\int_{S(F_{1},F_{2})}\left(
L_{2}\left(  p\right)  -L_{1}\left(  p\right)  \right)  \mathrm{d}p=0$ which
implies that $F_{1}$ Lorenz dominates $F_{2}$. If $F_{1}$ Lorenz dominates
$F_{2}$, then $\int_{S(F_{1},F_{2})}\left(  L_{2}\left(  p\right)
-L_{1}\left(  p\right)  \right)  \mathrm{d}p=0$ and thus $c\left(  L_{1}%
,L_{2}\right)  =0$. If $c\left(  L_{1},L_{2}\right)  =1$, then $\int_{S(F_{2},F_{1})}\left(
L_{1}\left(  p\right)  -L_{2}\left(  p\right)  \right)  \mathrm{d}p=0$ which
implies that $F_{2}$ Lorenz dominates $F_{1}$. 

It is clear that
\begin{align*}
\int_{0}^{1}\left\vert L_{2}\left(  p\right)  -L_{1}\left(  p\right)
\right\vert \mathrm{d}p  &  =\int_{0}^{1}\phi\left(  p\right)  1\left\{
\phi\left(  p\right)  \geq0\right\}  \mathrm{d}p+\int_{0}^{1}-\phi\left(
p\right)  1\left\{  \phi\left(  p\right)  <0\right\}  \mathrm{d}p\\
&  =\int_{0}^{1}\max\left\{  \phi\left(  p\right)  ,0\right\}  \mathrm{d}%
p+\int_{0}^{1}\max\left\{  -\phi\left(  p\right)  ,0\right\}  \mathrm{d}p.
\end{align*}
So by \eqref{eq.0timesinf}, \eqref{eq.LDC} holds.
If $c\left(  L_{1},L_{2}\right)  \in(0,1]  $, then by \eqref{eq.LDC},
\[
c\left(  L_{2},L_{1}\right)  =\frac{\int_{0}^{1}\max\left\{  -\phi\left(
p\right)  ,0\right\}  \mathrm{d}p}{\int_{0}^{1}\max\left\{  \phi\left(
p\right)  ,0\right\}  \mathrm{d}p+\int_{0}^{1}\max\left\{  -\phi\left(
p\right)  ,0\right\}  \mathrm{d}p}=1-c\left(  L_{1},L_{2}\right)  .
\]
\end{proof}

\begin{proof}[Proof of Proposition \ref{prop.c vs Gini}]
For every $c\left(
L_{1},L_{2}\right)  \in(0,1/2)$, by Lemma \ref{lemma.c properties} and Proposition 1 of \citet{zheng2018almost}, $F_1$ $\varepsilon$-ALD $F_2$
for every $\varepsilon\in\lbrack c\left(  L_{1},L_{2}\right)  ,1/2)$ and
$I\left(  F_1,\theta\right)  \leq I\left(  F_2,\theta\right)  $ for all $I\left(
\cdot,\theta\right)  \in\mathscr{B}^{\ast}\left(  \varepsilon\right)  $.
If $c\left(  L_{1},L_{2}\right)  =0$, then by Lemma \ref{lemma.c properties}, $F_{1}$ Lorenz
dominates $F_{2}$ and the claim is clearly true. 
\end{proof}

\begin{proof}[Proof of Proposition \ref{prop.AmULD}]
The proof closely follows the strategies of the proofs of Theorem 1 of \citet{leshno2002preferred} and Theorem 3.1A of \citet{aaberge2009ranking}.
 Using integration by parts, we can show that for every $P\in\mathrm{P}$,
\begin{align*}
&  J_{P}\left(  L_{2}\right)  -J_{P}\left(  L_{1}\right)  =\int_{0}%
^{1}P^{\prime}\left(  t\right)  \mathrm{d}L_{1}\left(  t\right)  -\int_{0}%
^{1}P^{\prime}\left(  t\right)  \mathrm{d}L_{2}\left(  t\right) \\
=&\,-\int_{0}^{1}P^{\left(  2\right)  }\left(  t\right)  \mathrm{d}L_{1}%
^{2}\left(  t\right)  +\int_{0}^{1}P^{\left(  2\right)  }\left(  t\right)
\mathrm{d}L_{2}^{2}\left(  t\right) \\
=&\,-P^{\left(  2\right)  }\left(  t\right)  L_{1}^{2}\left(  t\right)
|_{0}^{1}+P^{\left(  2\right)  }\left(  t\right)  L_{2}^{2}\left(  t\right)
|_{0}^{1}+\int_{0}^{1}P^{\left(  3\right)  }\left(  t\right)  \mathrm{d}%
L_{1}^{3}\left(  t\right)  -\int_{0}^{1}P^{\left(  3\right)  }\left(
t\right)  \mathrm{d}L_{2}^{3}\left(  t\right) \\
=&\,\sum_{j=2}^{m}\left(  -1\right)  ^{j-1}P^{\left(  j\right)  }\left(
1\right)  \left(  L_{1}^{j}\left(  1\right)  -L_{2}^{j}\left(  1\right)
\right)  +\int_{0}^{1}\left(  -1\right)  ^{m}P^{\left(  m+1\right)  }\left(
t\right)  \left(  L_{1}^{m}\left(  t\right)  -L_{2}^{m}\left(  t\right)
\right)  \mathrm{d}t.
\end{align*}
If $P\in\mathrm{P}_{m}\left(  \varepsilon_{m}\right)  $, then $\sum_{j=2}%
^{m}\left(  -1\right)  ^{j-1}P^{\left(  j\right)  }\left(  1\right)
(L_{1}^{j}\left(  1\right)  -L_{2}^{j}\left(  1\right)  )=0$. For every $m\ge2$, let $r_{U}
=\sup_{t}\{\left(  -1\right)  ^{m}P^{\left(  m+1\right)  }\left(  t\right)
\}$ and $r_{L}=\inf_{t}\{\left(  -1\right)  ^{m}P^{\left(  m+1\right)
}\left(  t\right)  \}$. Also, let $S=\{t\in\left[  0,1\right]  :L_{1}%
^{m}\left(  t\right)  -L_{2}^{m}\left(  t\right)  <0\}$. It then follows that
\begin{align*}
&  \int_{0}^{1}\left(  -1\right)  ^{m}P^{\left(  m+1\right)  }\left(
t\right)  \left(  L_{1}^{m}\left(  t\right)  -L_{2}^{m}\left(  t\right)
\right)  \mathrm{d}t\\
\geq&\, r_{U}\int_{S}\left(  L_{1}^{m}\left(  t\right)  -L_{2}^{m}\left(
t\right)  \right)  \mathrm{d}t+r_{L}\int_{S^{c}}\left(  L_{1}^{m}\left(
t\right)  -L_{2}^{m}\left(  t\right)  \right)  \mathrm{d}t\\
=&\,r_{U}\int_{S}\left(  L_{1}^{m}\left(  t\right)  -L_{2}^{m}\left(
t\right)  \right)  \mathrm{d}t-r_{L}\int_{S}\left(  L_{1}^{m}\left(  t\right)
-L_{2}^{m}\left(  t\right)  \right)  \mathrm{d}t\\
&  +r_{L}\int_{S}\left(  L_{1}^{m}\left(  t\right)  -L_{2}^{m}\left(
t\right)  \right)  \mathrm{d}t+r_{L}\int_{S^{c}}\left(  L_{1}^{m}\left(
t\right)  -L_{2}^{m}\left(  t\right)  \right)  \mathrm{d}t\\
=&\,\left(  r_{U}+r_{L}\right)  \int_{S}\left(  L_{1}^{m}\left(  t\right)
-L_{2}^{m}\left(  t\right)  \right)  \mathrm{d}t+r_{L}\int_{0}^{1}\left\vert
L_{1}^{m}\left(  t\right)  -L_{2}^{m}\left(  t\right)  \right\vert
\mathrm{d}t\\
=&\,-\left(  r_{U}+r_{L}\right)  \int_{0}^{1}\max\left(  L_{2}^{m}\left(
t\right)  -L_{1}^{m}\left(  t\right)  ,0\right)  \mathrm{d}t+r_{L}\int_{0}%
^{1}\left\vert L_{1}^{m}\left(  t\right)  -L_{2}^{m}\left(  t\right)
\right\vert \mathrm{d}t\\
\geq&\,-\left(  r_{U}+r_{L}\right)  \varepsilon_{m}\int_{0}^{1}\left\vert
L_{1}^{m}\left(  t\right)  -L_{2}^{m}\left(  t\right)  \right\vert
\mathrm{d}t+r_{L}\int_{0}^{1}\left\vert L_{1}^{m}\left(  t\right)  -L_{2}%
^{m}\left(  t\right)  \right\vert \mathrm{d}t.
\end{align*}
With $P\in\mathrm{P}_{m}\left(  \varepsilon_{m}\right)  $, we have
\[
r_{U}\leq r_{L}\left(  \frac{1}{\varepsilon_{m}}-1\right)  \Leftrightarrow
\left(  r_{U}+r_{L}\right)  \varepsilon_{m}\leq r_{L},
\]
which implies $J_{P}\left(  L_{2}\right)  -J_{P}\left(  L_{1}\right)  \geq0$.

If $L_{1}^{j}\left(  1\right)  -L_{2}^{j}\left(  1\right)  \geq0$ for all
$j\leq m$, then for every $P\in\mathrm{P}_{m}^{\prime}\left(  \varepsilon
_{m}\right)  $, we have that \linebreak $\sum_{j=2}^{m}\left(  -1\right)  ^{j-1}P^{\left(  j\right)
}\left(  1\right)  (L_{1}^{j}\left(  1\right)  -L_{2}^{j}\left(  1\right)
)\geq0$ and the result follows.
\end{proof}

\begin{proof}[Proof of Lemma \ref{lemma.cmu L properties}]
The proof is similar to that of Lemma \ref{lemma.c properties}.
\end{proof}

\begin{proof}[Proof of Proposition \ref{prop.ULDC properties}]
The results directly follow from Lemma \ref{lemma.cmu L properties} and Proposition \ref{prop.AmULD}.  
\end{proof}

\begin{proof}[Proof of Proposition \ref{prop.AmDLD}]
The proof closely follows the strategies of the proofs of Theorem 1 of \citet{leshno2002preferred} and Theorem 3.1B of \citet{aaberge2009ranking}.
Clearly, for $m\geq3$, we have that
${\mathrm{d}\tilde{L}_{j}^{m}\left(  p\right)  }/{\mathrm{d}p}=-\tilde
{L}_{j}^{m-1}\left(  p\right)$.
Also,
\begin{align*}
\tilde{L}_{j}^{2}\left(  p\right)  =\int_{p}^{1}\left(  1-L_{j}\left(
t\right)  \right)  \mathrm{d}t  =\left(  1-p\right)  -\int_{0}^{1}L_{j}\left(  t\right)  \mathrm{d}%
t+L_{j}^{2}\left(  p\right) ,
\end{align*}
and it follows that
${\mathrm{d}\tilde{L}_{j}^{2}\left(  p\right)  }/{\mathrm{d}p}
=-1+{\mathrm{d}L_{j}^{2}\left(  p\right)  }/{\mathrm{d}p}$.
Then we have that for every $P\in\mathrm{P}$,
\begin{align*}
&  J_{P}\left(  L_{2}\right)  -J_{P}\left(  L_{1}\right)  =\int_{0}%
^{1}P^{\prime}\left(  t\right)  \mathrm{d}L_{1}\left(  t\right)  -\int_{0}%
^{1}P^{\prime}\left(  t\right)  \mathrm{d}L_{2}\left(  t\right)  \\
=&\,-\int_{0}^{1}P^{\left(  2\right)  }\left(  t\right)  \mathrm{d}L_{1}%
^{2}\left(  t\right)  +\int_{0}^{1}P^{\left(  2\right)  }\left(  t\right)
\mathrm{d}L_{2}^{2}\left(  t\right)  =-\int_{0}^{1}P^{\left(  2\right)  }\left(  t\right)  \mathrm{d}\tilde
{L}_{1}^{2}\left(  t\right)  +\int_{0}^{1}P^{\left(  2\right)  }\left(
t\right)  \mathrm{d}\tilde{L}_{2}^{2}\left(  t\right)  \\
=&\,-P^{\left(  2\right)  }\left(  t\right)  \tilde{L}_{1}^{2}\left(
t\right)  |_{0}^{1}+P^{\left(  2\right)  }\left(  t\right)  \tilde{L}_{2}%
^{2}\left(  t\right)  |_{0}^{1}-\int_{0}^{1}P^{\left(  3\right)  }\left(
t\right)  \mathrm{d}\tilde{L}_{1}^{3}\left(  t\right)  +\int_{0}^{1}P^{\left(
3\right)  }\left(  t\right)  \mathrm{d}\tilde{L}_{2}^{3}\left(  t\right)  \\
=&\,\sum_{j=2}^{m}P^{\left(  j\right)  }\left(  0\right)  \left(  \tilde
{L}_{1}^{j}\left(  0\right)  -\tilde{L}_{2}^{j}\left(  0\right)  \right)
+\int_{0}^{1}-P^{\left(  m+1\right)  }\left(  t\right)  \left(  \tilde{L}%
_{2}^{m}\left(  t\right)  -\tilde{L}_{1}^{m}\left(  t\right)  \right)
\mathrm{d}t.
\end{align*}
If $P\in\tilde{\mathrm{P}}_{m}\left(  \varepsilon_{m}\right)  $, then $\sum_{j=2}%
^{m}P^{\left(  j\right)  }\left(  0\right)  (\tilde{L}_{1}^{j}\left(
0\right)  -L_{2}^{j}\left(  0\right)  )=0$. Let $r_{U}=\sup_{t}\{-P^{\left(
m+1\right)  }\left(  t\right)  \}$ and $r_{L}=\inf_{t}\{-P^{\left(
m+1\right)  }\left(  t\right)  \}$. Also, let $S=\{t\in\left[  0,1\right]
:\tilde{L}_{1}^{m}\left(  t\right)  -\tilde{L}_{2}^{m}\left(  t\right)  >0\}$.
It then follows that
\begin{align*}
&  \int_{0}^{1}-P^{\left(  m+1\right)  }\left(  t\right)  \left(  \tilde
{L}_{2}^{m}\left(  t\right)  -\tilde{L}_{1}^{m}\left(  t\right)  \right)
\mathrm{d}t\\
\geq&\, r_{U}\int_{S}\left(  \tilde{L}_{2}^{m}\left(  t\right)  -\tilde{L}%
_{1}^{m}\left(  t\right)  \right)  \mathrm{d}t+r_{L}\int_{S^{c}}\left(
\tilde{L}_{2}^{m}\left(  t\right)  -\tilde{L}_{1}^{m}\left(  t\right)
\right)  \mathrm{d}t\\
=&\,r_{U}\int_{S}\left(  \tilde{L}_{2}^{m}\left(  t\right)  -\tilde{L}_{1}%
^{m}\left(  t\right)  \right)  \mathrm{d}t-r_{L}\int_{S}\left(  \tilde{L}%
_{2}^{m}\left(  t\right)  -\tilde{L}_{1}^{m}\left(  t\right)  \right)
\mathrm{d}t\\
&  +r_{L}\int_{S}\left(  \tilde{L}_{2}^{m}\left(  t\right)  -\tilde{L}_{1}%
^{m}\left(  t\right)  \right)  \mathrm{d}t+r_{L}\int_{S^{c}}\left(  \tilde
{L}_{2}^{m}\left(  t\right)  -\tilde{L}_{1}^{m}\left(  t\right)  \right)
\mathrm{d}t\\
=&\,\left(  r_{U}+r_{L}\right)  \int_{S}\left(  \tilde{L}_{2}^{m}\left(
t\right)  -\tilde{L}_{1}^{m}\left(  t\right)  \right)  \mathrm{d}t+r_{L}%
\int_{0}^{1}\left\vert \tilde{L}_{1}^{m}\left(  t\right)  -\tilde{L}_{2}%
^{m}\left(  t\right)  \right\vert \mathrm{d}t\\
=&\,-\left(  r_{U}+r_{L}\right)  \int_{0}^{1}\max\left(  \tilde{L}_{1}%
^{m}\left(  t\right)  -\tilde{L}_{2}^{m}\left(  t\right)  ,0\right)
\mathrm{d}t+r_{L}\int_{0}^{1}\left\vert \tilde{L}_{1}^{m}\left(  t\right)
-\tilde{L}_{2}^{m}\left(  t\right)  \right\vert \mathrm{d}t\\
\geq&\,-\left(  r_{U}+r_{L}\right)  \varepsilon_{m}\int_{0}^{1}\left\vert
\tilde{L}_{1}^{m}\left(  t\right)  -\tilde{L}_{2}^{m}\left(  t\right)
\right\vert \mathrm{d}t+r_{L}\int_{0}^{1}\left\vert \tilde{L}_{1}^{m}\left(
t\right)  -\tilde{L}_{2}^{m}\left(  t\right)  \right\vert \mathrm{d}t.
\end{align*}
With $P\in\tilde{\mathrm{P}}_{m}\left(  \varepsilon_{m}\right)  $, we have
\[
r_{U}\leq r_{L}\left(  \frac{1}{\varepsilon_{m}}-1\right)  \Leftrightarrow
\left(  r_{U}+r_{L}\right)  \varepsilon_{m}\leq r_{L},
\]
which implies $J_{P}\left(  L_{2}\right)  -J_{P}\left(  L_{1}\right)  \geq0$.

If $\tilde{L}_{1}^{j}\left(  0\right)  -\tilde{L}_{2}^{j}\left(  0\right)
\leq0$ for all $j\leq m$, then for every $P\in\tilde{\mathrm{P}}_{m}^{\prime}\left(
\varepsilon_{m}\right)  $, $\sum_{j=2}^{m}P^{\left(  j\right)  }\left(
0\right)  (\tilde{L}_{1}^{j}\left(  0\right)  -\tilde{L}_{2}^{j}\left(
0\right)  )\geq0$ and the result follows.
\end{proof}

\begin{proof}[Proof of Lemma \ref{lemma.cmd L properties}]
The proof is similar to that of Lemma \ref{lemma.c properties}.
\end{proof}

\begin{proof}[Proof of Proposition \ref{prop.DLDC properties}]
The results directly follow from Proposition \ref{prop.AmDLD}.  
\end{proof}

\begin{lemma}\label{lemma.directional derivative}
Let $\mathbb{D}$ be a normed space. Suppose that $\mathcal{F}_{1}%
:  \mathbb{D}  \rightarrow\mathbb{R}$ and
$\mathcal{F}_{2}:  \mathbb{D}  \rightarrow
\mathbb{R}$ such that $\mathcal{F}_{1}$ and $\mathcal{F}_{2}$ are both
Hadamard directionally differentiable at $\phi$ tangentially to some
$\mathbb{D}_{0}\subset\mathbb{D}$ with the derivatives $\mathcal{F}_{1\phi
}^{\prime}$ and $\mathcal{F}_{2\phi}^{\prime}$. Then we have the following results:
\begin{enumerate}[label=(\roman*)]
\item The summation $\mathcal{F}_1+\mathcal{F}_2$ is Hadamard directionally differentiable at $\phi$ such that
\[
\left(  \mathcal{F}_{1}+\mathcal{F}_{2}\right)  _{\phi}^{\prime}\left(
h\right)  =\mathcal{F}_{1\phi}^{\prime}\left(  h\right)  +\mathcal{F}_{2\phi
}^{\prime}\left(  h\right)  ,\quad h\in\mathbb{D}_{0}.
\]

\item The multiplication $\mathcal{F}_1\mathcal{F}_2$ is Hadamard directionally differentiable at $\phi$ such that
\[
\left(  \mathcal{F}_{1}\mathcal{F}_{2}\right)  _{\phi}^{\prime}\left(
h\right)  =\mathcal{F}_{1\phi}^{\prime}\left(  h\right)  \mathcal{F}%
_{2}\left(  \phi\right)  +\mathcal{F}_{1}\left(  \phi\right)  \mathcal{F}%
_{2\phi}^{\prime}\left(  h\right)  ,\quad h\in\mathbb{D}_{0}.
\]

\item If $\mathcal{F}_1\left(  \phi\right)  \neq0$, then
the inverse $1/\mathcal{F}_1$ is Hadamard directionally differentiable at $\phi$ such that
\[
\left(  \frac{1}{\mathcal{F}_{1}}\right)  _{\phi}^{\prime}\left(  h\right)
=-\frac{\mathcal{F}_{1\phi}^{\prime}\left(  h\right)  }{\mathcal{F}_1\left(
\phi\right)  ^{2}},\quad h\in\mathbb{D}_{0}.
\]

\end{enumerate}
\end{lemma}

\begin{proof}[Proof of Lemma \ref{lemma.directional derivative}]
Let $t_n \downarrow 0$ and $h_n\to h\in\mathbb{D}_0$ such that $\phi+t_n h_n\in\mathbb{D}$.

(i). By definition, it is easy to show that
\[
\frac{\mathcal{F}_{1}\left(  \phi+t_{n}h_{n}\right)  +\mathcal{F}_{2}\left(
\phi+t_{n}h_{n}\right)  -\mathcal{F}_{1}\left(  \phi\right)  -\mathcal{F}%
_{2}\left(  \phi\right)  }{t_{n}}-\mathcal{F}_{1\phi}^{\prime}\left(
h\right)  -\mathcal{F}_{2\phi}^{\prime}\left(  h\right)  \rightarrow0.
\]

(ii). First, we have that
\begin{align*}
&  \frac{\mathcal{F}_{1}\left(  \phi+t_{n}h_{n}\right)  \mathcal{F}_{2}\left(
\phi+t_{n}h_{n}\right)  -\mathcal{F}_{1}\left(  \phi\right)  \mathcal{F}%
_{2}\left(  \phi\right)  }{t_{n}}\\
=&\,\frac{\mathcal{F}_{1}\left(  \phi+t_{n}h_{n}\right)  \mathcal{F}%
_{2}\left(  \phi+t_{n}h_{n}\right)  -\mathcal{F}_{1}\left(  \phi\right)
\mathcal{F}_{2}\left(  \phi+t_{n}h_{n}\right)  }{t_{n}}\\
&  +\frac{\mathcal{F}_{1}\left(  \phi\right)  \mathcal{F}_{2}\left(
\phi+t_{n}h_{n}\right)  -\mathcal{F}_{1}\left(  \phi\right)  \mathcal{F}%
_{2}\left(  \phi\right)  }{t_{n}}.
\end{align*}
Then, by the continuity of $\mathcal{F}_2$ ($\mathcal{F}_2$ is directionally differentiable), we can show that
\begin{align*}
&  \frac{\mathcal{F}_{1}\left(  \phi+t_{n}h_{n}\right)  \mathcal{F}_{2}\left(
\phi+t_{n}h_{n}\right)  -\mathcal{F}_{1}\left(  \phi\right)  \mathcal{F}%
_{2}\left(  \phi+t_{n}h_{n}\right)  }{t_{n}}-\mathcal{F}_{1\phi}^{\prime
}\left(  h\right)  \mathcal{F}_{2}\left(  \phi+t_{n}h_{n}\right)  \\
&+\mathcal{F}_{1\phi}^{\prime}\left(  h\right)  \mathcal{F}_{2}\left(
\phi+t_{n}h_{n}\right)  -\mathcal{F}_{1\phi}^{\prime}\left(  h\right)
\mathcal{F}_{2}\left(  \phi\right)  \rightarrow0
\end{align*}
and
\[
\frac{\mathcal{F}_{1}\left(  \phi\right)  \mathcal{F}_{2}\left(  \phi
+t_{n}h_{n}\right)  -\mathcal{F}_{1}\left(  \phi\right)  \mathcal{F}%
_{2}\left(  \phi\right)  }{t_{n}}-\mathcal{F}_{1}\left(  \phi\right)
\mathcal{F}_{2\phi}^{\prime}\left(  h\right)  \rightarrow0.
\]

(iii). 
We first write 
\[
\frac{1/\mathcal{F}_{1}\left(  \phi+t_{n}h_{n}\right)  -1/\mathcal{F}%
_{1}\left(  \phi\right)  }{t_{n}}=\frac{1}{t_{n}}\frac{\mathcal{F}_{1}\left(
\phi\right)  -\mathcal{F}_{1}\left(  \phi+t_{n}h_{n}\right)  }{\mathcal{F}%
_{1}\left(  \phi+t_{n}h_{n}\right)  \mathcal{F}_{1}\left(  \phi\right)  }.
\]
Then, by the continuity of $\mathcal{F}_1$ ($\mathcal{F}_1$ is directionally differentiable), we can show that
\begin{align*}
\frac{1}{t_{n}}\frac{\mathcal{F}_{1}\left(  \phi\right)  -\mathcal{F}%
_{1}\left(  \phi+t_{n}h_{n}\right)  }{\mathcal{F}_{1}\left(  \phi+t_{n}%
h_{n}\right)  \mathcal{F}_{1}\left(  \phi\right)  }&-\left(  -\frac
{\mathcal{F}_{1\phi}^{\prime}\left(  h\right)  }{\mathcal{F}_{1}\left(
\phi+t_{n}h_{n}\right)  \mathcal{F}_{1}\left(  \phi\right)  }\right)\\
&-\frac{\mathcal{F}_{1\phi}^{\prime}\left(  h\right)  }{\mathcal{F}_{1}\left(
\phi+t_{n}h_{n}\right)  \mathcal{F}_{1}\left(  \phi\right)  }-\left(
-\frac{\mathcal{F}_{1\phi}^{\prime}\left(  h\right)  }{\mathcal{F}_{1}\left(
\phi\right)  ^{2}}\right)  \rightarrow0.
\end{align*}

\end{proof}

\begin{proof}[Proof of Lemma \ref{lemma.G kernel}]
We closely follow the proofs of Lemma A.1 of \citet{Beare2017improved} and Lemma 3.1 of \citet{jiang2023nonparametric}.
First, we have that
\begin{align*}
& E\left[  \mathbb{G}\left(  p\right)  \mathbb{G}\left(  p^{\prime}\right)
\right]  =E\left[  \left(  \mathbb{\lambda}^{1/2}\mathcal{L}_{2}\left(  p\right)
-\left(  1-\lambda\right)  ^{1/2}\mathcal{L}_{1}\left(  p\right)  \right)
\left(  \mathbb{\lambda}^{1/2}\mathcal{L}_{2}\left(  p^{\prime}\right)
-\left(  1-\lambda\right)  ^{1/2}\mathcal{L}_{1}\left(  p^{\prime}\right)
\right)  \right]  \\
=&\,\lambda E\left[  \mathcal{L}_{2}\left(  p\right)  \mathcal{L}_{2}\left(
p^{\prime}\right)  \right]  -\mathbb{\lambda}^{1/2}\left(  1-\lambda\right)
^{1/2}E\left[  \mathcal{L}_{2}\left(  p\right)  \mathcal{L}_{1}\left(
p^{\prime}\right)  \right]  -\mathbb{\lambda}^{1/2}\left(  1-\lambda\right)
^{1/2}E\left[  \mathcal{L}_{1}\left(  p\right)  \mathcal{L}_{2}\left(
p^{\prime}\right)  \right]  \\
&  +\left(  1-\lambda\right)  E\left[  \mathcal{L}_{1}\left(  p\right)
\mathcal{L}_{1}\left(  p^{\prime}\right)  \right].
\end{align*}

For $j=1,2$, $p\in\left[  0,1\right]  $, and $t\in\left[  0,1\right]
$, define%
\[
h_{j,p}\left(  t\right)  =\frac{1}{\mu_{j}}\left\{  L_{j}\left(  p\right)
-1\left(  t\leq p\right)  \right\}  Q_{j}^{\prime}\left(  t\right)  .
\]
Note that the almost sure integrability of $h_{j,p}\mathcal{B}_{j}$ follows from
the weak convergence of $n_{j}^{1/2}(\hat{Q}_{j}-Q_{j})\leadsto-Q_{j}^{\prime
}\mathcal{B}_{j}$ in $L^{1}\left(  \left[  0,1\right]  \right)  $ established
by \citet{K17}. Then we can show that $\mathcal{L}_j$ defined in \eqref{eq.L} satisfies that
\[
\mathcal{L}_{j}\left(  p\right)  =\int_{0}^{1}h_{j,p}\left(
t\right)  \mathcal{B}_{j}\left(  t\right)  \mathrm{d}t,\quad p\in\left[
0,1\right].
\]
Define
\[
H_{j,p}\left(  u\right)  =\int_{0}^{u}h_{j,p}\left(  t\right)  \mathrm{d}t,\quad p\in\left[
0,1\right].
\]
For $j,j^{\prime}\in\left\{  1,2\right\}  $ and $p,p^{\prime}\in\left[
0,1\right]  $, by Fubini's theorem,
\[
Cov\left(  \mathcal{L}_{j}\left(  p\right)  ,\mathcal{L}_{j^{\prime}}\left(
p^{\prime}\right)  \right)  =\int_{0}^{1}\int_{0}^{1}h_{j,p}\left(
t_{1}\right)  h_{j^{\prime},p^{\prime}}\left(  t_{2}\right)  E\left[
\mathcal{B}_{j}\left(  t_{1}\right)  \mathcal{B}_{j^{\prime}}\left(
t_{2}\right)  \right]  \mathrm{d}t_{1}\mathrm{d}t_{2}.
\]
Let $\left(  U,V\right)  $ be a pair of random variables with joint CDF given
by the copula $\mathbf C$ in Assumption \ref{ass.data}(ii). For $j=1$ and $j^{\prime}=2$, $E\left[  \mathcal{B}%
_{1}\left(  t_{1}\right)  \mathcal{B}_{2}\left(  t_{2}\right)  \right]
=\mathbf C\left(  t_{1},t_{2}\right)  -t_{1}t_{2}$. By Theorem 3.1 of \citet{lo2017functional} (see also \citet{cuadras2002covariance} and \citet{beare2009generalization}),
\begin{align*}
Cov\left(  H_{j,p}\left(  U\right)  ,H_{j^{\prime},p^{\prime}}\left(
V\right)  \right)  =\int_{0}^{1}\int_{0}^{1}h_{j,p}\left(  t_{1}\right)
h_{j^{\prime},p^{\prime}}\left(  t_{2}\right)  \left( \mathbf C\left(  t_{1}%
,t_{2}\right)  -t_{1}t_{2}\right)  \mathrm{d}t_{1}\mathrm{d}t_{2}.    
\end{align*}
Then it follows that
\begin{align*}
&E\left[  \mathcal{L}_{1}\left(  p\right)  \mathcal{L}_{2}\left(  p^{\prime
}\right)  \right]  =Cov\left(  \mathcal{L}_{1}\left(  p\right)
,\mathcal{L}_{2}\left(  p^{\prime}\right)  \right)  =Cov\left(
H_{1,p}\left(  F_{1}\left(  X_i^{1}\right)  \right)  ,H_{2,p^{\prime}}\left(
F_{2}\left(  X_i^{2}\right)  \right)  \right)  \\
& =Cov\left(  \frac{1}{\mu_{1}}\left\{  L_{1}\left(  p\right)  X_i^{1}%
-Q_{1}\left(  p\right)  \wedge X_i^{1}\right\}  ,\frac{1}{\mu_{2}}\left\{
L_{2}\left(  p^{\prime}\right)  X_i^{2}-Q_{2}\left(  p^{\prime}\right)  \wedge
X_i^{2}\right\}  \right)  .
\end{align*}
For $j=j^{\prime}$, $E\left[  \mathcal{B}_{j}\left(  t_{1}\right)
\mathcal{B}_{j}\left(  t_{2}\right)  \right]  =t_{1}\wedge t_{2}-t_{1}t_{2}$.
By Theorem 3.1 of \citet{lo2017functional} again,
\[
Cov\left(  H_{j,p}\left(  U\right)  ,H_{j,p^{\prime}}\left(  U\right)
\right)  =\int_{0}^{1}\int_{0}^{1}h_{j,p}\left(  t_{1}\right)  h_{j,p'}\left(
t_{2}\right)  \left(  t_{1}\wedge t_{2}-t_{1}t_{2}\right)  \mathrm{d}%
t_{1}\mathrm{d}t_{2}.
\]
Then it follows that
\begin{align*}
&E\left[  \mathcal{L}_{j}\left(  p\right)  \mathcal{L}_{j}\left(  p^{\prime
}\right)  \right] =Cov\left(  \mathcal{L}_{j}\left(  p\right)
,\mathcal{L}_{j}\left(  p^{\prime}\right)  \right)  =Cov\left(  H_{j,p}\left(
F_{j}\left(  X_i^{j}\right)  \right)  ,H_{j,p^{\prime}}\left(  F_{j}\left(
X_i^{j}\right)  \right)  \right)  \\
& =Cov\left(  \frac{1}{\mu_{j}}\left\{  L_{j}\left(  p\right)  X_i^{j}%
-Q_{j}\left(  p\right)  \wedge X_i^{j}\right\}  ,\frac{1}{\mu_{j}}\left\{
L_{j}\left(  p^{\prime}\right)  X_i^{j}-Q_{j}\left(  p^{\prime}\right)  \wedge
X_i^{j}\right\}  \right)  .
\end{align*}
\end{proof}

\begin{proof}[Proof of Lemma \ref{lemma.phi weak convergence}]
The results follow from the continuous mapping theorem and Fubini's theorem (see, for example, Theorem 2.37 of \citet{folland2013real}). 
\end{proof}

\begin{proof}[Proof of Proposition \ref{prop.c asymptotic limit}]
For simplicity of notation, we focus on the case where $m=1$. The proof of general results for $m\ge2$ and $w\in\{u,d\}$ is similar. 
We first show the Hadamard differentiability of $\mathcal{F}$. Because $L_1$ and $L_2$ are continuous in $p$, if $L_1(p)\neq L_2(p)$ for some $p\in[0,1]$, then $\mathcal{F}_1(\phi)+\mathcal{F}_2(\phi)> 0$. By Lemma \ref{lemma.directional derivative}, we have that
\begin{align}\label{eq.F HD}
&\mathcal{F}_{\phi}^{\prime}\left(  h\right)  =\left(  \frac{\mathcal{F}%
_{1}}{\mathcal{F}_{1}+\mathcal{F}_{2}}\right)  _{\phi}^{\prime}\left(
h\right)  =\frac{\mathcal{F}_{1\phi}^{\prime}\left(  h\right)  }%
{\mathcal{F}_{1}\left(  \phi\right)  +\mathcal{F}_{2}\left(  \phi\right)
}-\mathcal{F}_{1}\left(  \phi\right)  \frac{\mathcal{F}_{1\phi}^{\prime
}\left(  h\right)  +\mathcal{F}_{2\phi}^{\prime}\left(  h\right)  }{\left(
\mathcal{F}_{1}\left(  \phi\right)  +\mathcal{F}_{2}\left(  \phi\right)
\right)  ^{2}}\notag\\
=&\,\frac{\mathcal{F}_{1\phi}^{\prime}\left(  h\right)  \left[  \mathcal{F}%
_{1}\left(  \phi\right)  +\mathcal{F}_{2}\left(  \phi\right)  \right]
-\mathcal{F}_{1}\left(  \phi\right)  \left[  \mathcal{F}_{1\phi}^{\prime
}\left(  h\right)  +\mathcal{F}_{2\phi}^{\prime}\left(  h\right)  \right]
}{\left(  \mathcal{F}_{1}\left(  \phi\right)  +\mathcal{F}_{2}\left(
\phi\right)  \right)  ^{2}} =\frac{\mathcal{F}_{1\phi}^{\prime}\left(  h\right)  \mathcal{F}_{2}\left(
\phi\right)  -\mathcal{F}_{1}\left(  \phi\right)  \mathcal{F}_{2\phi}^{\prime
}\left(  h\right)  }{\left(  \mathcal{F}_{1}\left(  \phi\right)
+\mathcal{F}_{2}\left(  \phi\right)  \right)  ^{2}}.
\end{align}
This implies the continuity of $\mathcal{F}$. By Lemma 3 of \citet{BDB14}, we have $\hat{c}(L_1,L_2)\to c(L_1,L_2)$ a.s. With \eqref{eq.phi weak convergence} and \eqref{eq.F HD}, by Theorem 2.1 of \citet{fang2014inference}, \eqref{eq.c asymptotic limit} holds.
\end{proof}

\begin{proof}[Proof of Proposition \ref{prop.confidence interval}]
For simplicity of notation, we focus on the case where $m=1$. The proof of general results for $m\ge2$ and $w\in\{u,d\}$ is similar.
Since $c(L_1,L_2)\neq0$, $L_1\neq L_2$ and $\mathcal{F}_1(\phi)+\mathcal{F}_2(\phi)>0$.
Let $\varepsilon>0$ and $\mu$ denote the Lebesgue measure in the following. 
First, we have that
\begin{align*}
&\mathbb{P}\left(  \mu\left(  \widehat{B_{+}\left(  \phi\right)  }\setminus
B_{+}\left(  \phi\right)  \right)  >\varepsilon\right)     \leq
\mathbb{P}\left(  \widehat{B_{+}\left(  \phi\right)  }\setminus B_{+}\left(
\phi\right)  \neq\varnothing\right) \\
\leq&\,\mathbb{P}\left(  \sup_{p\in\widehat{B_{+}\left(  \phi\right)
}\setminus B_{+}\left(  \phi\right)  }\frac{\sqrt{T_{n}}(\hat{\phi}\left(
p\right)  -\phi\left(  p\right)  )}{\xi_{0}\vee\hat{\sigma}\left(  p\right)
}>t_{n}\right)  \leq\mathbb{P}\left(  \sup_{p\in\left[  0,1\right]  }\frac{\sqrt{T_{n}%
}(\hat{\phi}\left(  p\right)  -\phi\left(  p\right)  )}{\xi_{0}\vee\hat
{\sigma}\left(  p\right)  }>t_{n}\right)  .
\end{align*}
Since $\sqrt{T_{n}}(\hat{\phi}  -\phi
)\leadsto\mathbb{G}$, then by
Example 1.4.7 (Slutsky’s lemma) and Theorem 1.3.6 (continuous mapping) of \citet{van1996weak},
\[
\frac{1}{t_{n}}\sup_{p\in\left[  0,1\right]  }\frac{\sqrt{T_{n}}(\hat{\phi
}\left(  p\right)  -\phi\left(  p\right)  )}{\xi_{0}\vee\hat{\sigma}\left(
p\right)  }\le 
\frac{1}{t_{n}}\sup_{p\in\left[  0,1\right]  }\left\vert\frac{\sqrt{T_{n}}(\hat{\phi
}\left(  p\right)  -\phi\left(  p\right)  )}{\xi_{0}  } \right\vert
\leadsto0.
\]
Then by Theorem 1.3.4(iii) of \citet{van1996weak}, $\mathbb{P}(\mu(\widehat{B_{+}\left(  \phi\right)  }\setminus
B_{+}\left(  \phi\right)  )>\varepsilon)\rightarrow0$.

By Lemma 3 of \citet{BDB14}, $\sup_{p\in\left[  0,1\right]  }|\hat{\phi}\left(  p\right)
-\phi\left(  p\right)  |\rightarrow0$ a.s. For simplicity of notation, let $\hat{\sigma}$ denote $\hat{\sigma}_1^w$ in the following.
In the proof of Lemma A.2 of \citet{Beare2017improved}, it is shown that $\sup_{p\in[0,1]}\hat{\sigma}(p)$ is bounded by some $C_{\sigma}>0$ a.s.
Define
\[
A_{\phi}=\left\{  \omega\in\Omega:\sup_{p\in\left[  0,1\right]  }\left\vert
\hat{\phi}_{\omega}\left(  p\right)  -\phi\left(  p\right)  \right\vert
\rightarrow0\right\}  \cap \left\{ \omega\in\Omega: \sup_{p\in[0,1]}\hat{\sigma}_{\omega}(p) \le C_{\sigma} \right\},
\]
where the subscript ${\omega}$ denotes that the random elements are fixed at $\omega$.
Then $\mathbb{P}\left(  A_{\phi}\right)  =1$. Consider
\[
\mathbb{P}\left(  \mu\left(  B_{+}\left(  \phi\right)  \setminus\widehat
{B_{+}\left(  \phi\right)  }\right)  >\varepsilon\right)  =\int1\left\{
\mu\left(  B_{+}\left(  \phi\right)  \setminus\widehat{B_{+}\left(
\phi\right)  }\right)  >\varepsilon\right\}  \mathrm{d}\mathbb{P}.
\]
For every fixed $\omega\in A_{\phi}$, we have that%
\[
\mu\left(  B_{+}\left(  \phi\right)  \setminus\widehat{B_{+}\left(
\phi\right)  }\right)  =\int_{\left[  0,1\right]  }1\left\{  p\in B_{+}\left(
\phi\right)  \setminus\widehat{B_{+}\left(  \phi\right)  }\right\}
\mathrm{d}\mu\left(  p\right)  ,
\]
where%
\[
B_{+}\left(  \phi\right)  \setminus\widehat{B_{+}\left(  \phi\right)
}=\left\{  p\in\left[  0,1\right]  :\frac{\sqrt{T_{n}}\hat{\phi}_{\omega
}\left(  p\right)  }{\xi_{0}\vee\hat{\sigma}_{\omega}\left(  p\right)  }\leq t_{n}%
,\phi\left(  p\right)  >0\right\}  .
\]
For every $p\in\left[  0,1\right]  $ such that $\phi\left(  p\right)  >0$,
since $\sup_{p\in\left[  0,1\right]  }\vert \hat{\phi}_{\omega}\left(
p\right)  -\phi\left(  p\right)  \vert \rightarrow0$, then when $n$ is
large enough, $\hat{\phi}_{\omega}\left(  p\right)  >1/2\cdot\phi\left(  p\right)
$. Since $\hat{\sigma}_{\omega}$ is uniformly bounded and $t_{n}/\sqrt{T_{n}%
}\rightarrow0$, we have that ($\omega$ is fixed) for every fixed $p\in[0,1]$,
\[
1\left\{  p\in B_{+}\left(  \phi\right)  \setminus\widehat{B_{+}\left(
\phi\right)  }\right\}  \rightarrow0.
\]
Thus, by the dominated convergence theorem ($\omega$
is fixed),
\[
\mu\left(  B_{+}\left(  \phi\right)  \setminus\widehat{B_{+}\left(
\phi\right)  }\right)  =\int_{\left[  0,1\right]  }1\left\{  p\in B_{+}\left(
\phi\right)  \setminus\widehat{B_{+}\left(  \phi\right)  }\right\}
\mathrm{d}\mu\left(  p\right)  \rightarrow0.
\]
This implies that for every $\omega\in A_{\phi}$,
\[
1\left\{  \mu\left(  B_{+}\left(  \phi\right)  \setminus\widehat{B_{+}\left(
\phi\right)  }\right)  >\varepsilon\right\}  \rightarrow0.
\]
Then by the dominated convergence theorem again,
\[
\mathbb{P}\left(  \mu\left(  B_{+}\left(  \phi\right)  \setminus\widehat
{B_{+}\left(  \phi\right)  }\right)  >\varepsilon\right)  =\int1\left\{
\mu\left(  B_{+}\left(  \phi\right)  \setminus\widehat{B_{+}\left(
\phi\right)  }\right)  >\varepsilon\right\}  \mathrm{d}\mathbb{P}%
\rightarrow0.
\]
Similarly, we can show that
\[
\mathbb{P}\left(  \mu\left(  \widehat{B_{+}\left(  -\phi\right)  }\setminus
B_{+}\left(  -\phi\right)  \right)  >\varepsilon\right)  \rightarrow0\text{
and }\mathbb{P}\left(  \mu\left(  B_{+}\left(  -\phi\right)  \setminus
\widehat{B_{+}\left(  -\phi\right)  }\right)  >\varepsilon\right)
\rightarrow0.
\]
In addition,
\begin{align*}
\mathbb{P}\left(  \mu\left(  B_{0}\left(  \phi\right)  \setminus\widehat
{B_{0}\left(  \phi\right)  }\right)  >\varepsilon\right)   &  \leq
\mathbb{P}\left(  B_{0}\left(  \phi\right)  \setminus\widehat{B_{0}\left(
\phi\right)  }\neq\varnothing\right) \\
&  \leq\mathbb{P}\left(  \sup_{p\in B_{0}\left(  \phi\right)  \setminus
\widehat{B_{0}\left(  \phi\right)  }}\left\vert \frac{\sqrt{T_{n}}(\hat{\phi
}\left(  p\right)  -\phi\left(  p\right)  )}{\xi_{0}\vee\hat{\sigma}\left(
p\right)  }\right\vert >t_{n}\right)  \rightarrow0.
\end{align*}
Also, we have that
\[
\widehat{B_{0}\left(  \phi\right)  }\setminus B_{0}\left(  \phi\right)
=\left\{  p\in\left[  0,1\right]  :\frac{\sqrt{T_{n}}\left\vert \hat{\phi
}\left(  p\right)  \right\vert }{\xi_{0}\vee\hat{\sigma}\left(
p\right)  }\leq t_{n},\left\vert \phi\left(  p\right)  \right\vert >0\right\}
.
\]
Then by similar arguments, we have that
\[
\mathbb{P}\left(  \mu\left(  \widehat{B_{0}\left(  \phi\right)  }\setminus
B_{0}\left(  \phi\right)  \right)  >\varepsilon\right)  \rightarrow0.
\]

It is easy to show that $\mathcal{\hat{F}}_{1\phi}^{\prime}\left(  h\right)  $
is Lipschitz continuous, because
\begin{align*}
&\left\vert \mathcal{\hat{F}}_{1\phi}^{\prime}\left(  h_{1}\right)
-\mathcal{\hat{F}}_{1\phi}^{\prime}\left(  h_{2}\right)  \right\vert  
\leq\left\vert \int_{\widehat{B_{+}\left(  \phi\right)  }}h_{1}\left(
p\right)  \mathrm{d}p-\int_{\widehat{B_{+}\left(  \phi\right)  }}h_{2}\left(
p\right)  \mathrm{d}p\right\vert\\ 
&+\left\vert \int_{\widehat{B_{0}\left(
\phi\right)  }}\max\left\{  h_{1}\left(  p\right)  ,0\right\}  \mathrm{d}%
p-\int_{\widehat{B_{0}\left(  \phi\right)  }}\max\left\{  h_{2}\left(
p\right)  ,0\right\}  \mathrm{d}p\right\vert \leq2\left\Vert h_{1}-h_{2}\right\Vert _{\infty}.
\end{align*}
Similarly, we can show that $\mathcal{\hat{F}}_{2\phi}^{\prime}\left(
h\right)  $ is Lipschitz continuous. For every $h\in C\left(  \left[
0,1\right]  \right)  $,
\begin{align*}
\left\vert \mathcal{\hat{F}}_{1\phi}^{\prime}\left(  h\right)
-\mathcal{F}_{1\phi}^{\prime}\left(  h\right)  \right\vert 
\leq&\,\left\Vert h\right\Vert _{\infty}\bigg[ \mu\left(  B_{+}\left(
\phi\right)  \setminus\widehat{B_{+}\left(  \phi\right)  }\right)  +\mu\left(
\widehat{B_{+}\left(  -\phi\right)  }\setminus B_{+}\left(  -\phi\right)
\right) \\
&+\mu\left(  B_{0}\left(  \phi\right)  \setminus\widehat{B_{0}\left(
\phi\right)  }\right)  +\mu\left(  \widehat{B_{0}\left(  \phi\right)
}\setminus B_{0}\left(  \phi\right)  \right)  \bigg] \rightarrow_{p}0.
\end{align*}
Similarly, we can show that
\[
\left\vert \mathcal{\hat{F}}_{2\phi}^{\prime}\left(  h\right)  -\mathcal{F}%
_{2\phi}^{\prime}\left(  h\right)  \right\vert \rightarrow_{p}0.
\]
Next, since $\hat{\phi}\rightarrow\phi$ a.s. by \citet{BDB14}, we have that for some $C_n=O_p(1)$,
\begin{align*}
&\left\vert \mathcal{\hat{F}}_{\phi}^{\prime}\left(  h_{1}\right)
-\mathcal{\hat{F}}_{\phi}^{\prime}\left(  h_{2}\right)  \right\vert  \\
=&\,\left\vert \frac{\mathcal{\hat{F}}_{1\phi}^{\prime}\left(  h_{1}\right)
\mathcal{F}_{2}\left(  \hat{\phi}\right)  -\mathcal{F}_{1}\left(  \hat{\phi
}\right)  \mathcal{\hat{F}}_{2\phi}^{\prime}\left(  h_{1}\right)  }{\left(
\mathcal{F}_{1}\left(  \hat{\phi}\right)  +\mathcal{F}_{2}\left(  \hat{\phi
}\right)  \right)  ^{2}}-\frac{\mathcal{\hat{F}}_{1\phi}^{\prime}\left(
h_{2}\right)  \mathcal{F}_{2}\left(  \hat{\phi}\right)  -\mathcal{F}%
_{1}\left(  \hat{\phi}\right)  \mathcal{\hat{F}}_{2\phi}^{\prime}\left(
h_{2}\right)  }{\left(  \mathcal{F}_{1}\left(  \hat{\phi}\right)
+\mathcal{F}_{2}\left(  \hat{\phi}\right)  \right)  ^{2}}\right\vert \\
\leq&\,\frac{\left[  \left\vert \mathcal{F}_{2}\left(  \hat{\phi}\right)
\right\vert \left\vert \mathcal{\hat{F}}_{1\phi}^{\prime}\left(  h_{1}\right)
-\mathcal{\hat{F}}_{1\phi}^{\prime}\left(  h_{2}\right)  \right\vert
+\left\vert \mathcal{F}_{1}\left(  \hat{\phi}\right)  \right\vert \left\vert
\mathcal{\hat{F}}_{2\phi}^{\prime}\left(  h_{1}\right)  -\mathcal{\hat{F}%
}_{2\phi}^{\prime}\left(  h_{2}\right)  \right\vert \right]  }{\left(
\mathcal{F}_{1}\left(  \hat{\phi}\right)  +\mathcal{F}_{2}\left(  \hat{\phi
}\right)  \right)  ^{2}}\leq C_n\left\Vert h_{1}-h_{2}\right\Vert _{\infty}.
\end{align*}
Also, since  $\mathcal{\hat{F}}_{1\phi
}^{\prime}\left(  h\right)  \rightarrow_{p}\mathcal{F}_{1\phi}^{\prime}\left(
h\right)  $ and $\mathcal{\hat{F}}_{2\phi}^{\prime}\left(  h\right)
\rightarrow_{p}\mathcal{F}_{2\phi}^{\prime}\left(  h\right)  $ for every $h\in C\left(  \left[  0,1\right]  \right)  $, then we have
$\mathcal{\hat
{F}}_{\phi}^{\prime}\left(  h\right)  \rightarrow_{p}\mathcal{F}_{\phi
}^{\prime}\left(  h\right)  $.
By Remark 3.4 of \citet{fang2014inference}, Assumption 4 of \citet{fang2014inference} holds. 

By a proof similar to that of Theorem S.1.1 in \citet{fang2014inference}, we can show that
\[
\hat{c}_{1-\alpha/2}\rightarrow_{p}c_{1-\alpha/2}\text{ and }\hat{c}%
_{\alpha/2}\rightarrow_{p}c_{\alpha/2}.
\]
By Proposition \ref{prop.c asymptotic limit} in the paper and Example 1.4.7 (Slutsky’s lemma), Theorem 1.3.6 (continuous mapping), and Theorem 1.3.4(vi) of \citet{van1996weak},
\[
\lim_{n\rightarrow\infty} \mathbb{P}\left(  \sqrt{T_{n}}\left(
\hat{c}\left(  L_{1},L_{2}\right)  -c\left(  L_{1},L_{2}\right)  \right)
>\hat{c}_{1-\alpha/2}\right)  =\alpha/2
\]
and%
\[
\lim_{n\rightarrow\infty}\mathbb{P}\left(  \sqrt{T_{n}}\left(
\hat{c}\left(  L_{1},L_{2}\right)  -c\left(  L_{1},L_{2}\right)  \right)
<\hat{c}_{\alpha/2}\right)  =\alpha/2,
\]
which imply that
\[
\lim_{n\rightarrow\infty}\mathbb{P}\left(  \sqrt{T_{n}}\left(  \hat
{c}\left(  L_{1},L_{2}\right)  -c\left(  L_{1},L_{2}\right)  \right)
\in\left[  \hat{c}_{\alpha/2},\hat{c}_{1-\alpha/2}\right]  \right)
=1-\alpha,
\]
or equivalently%
\[
\lim_{n\rightarrow\infty}\mathbb{P}\left(  c\left(  L_{1},L_{2}\right)
\in\left[  \hat{c}\left(  L_{1},L_{2}\right)  -T_{n}^{-1/2}\hat{c}%
_{1-\alpha/2},\hat{c}\left(  L_{1},L_{2}\right)  -T_{n}^{-1/2}\hat{c}%
_{\alpha/2}\right]  \right)  =1-\alpha.
\]
\end{proof}

\subsection{Proofs for Section \ref{sec.ISDC}}

\begin{proof}[Proof of Proposition \ref{prop.AmUISD}]
The proof closely follows the strategies of the proofs of Theorem 1 of \citet{leshno2002preferred} and Theorem 2.3 of \citet{aaberge2021ranking}.
We can show that
\begin{align*}
&  W_{P}\left(  F_{1}\right)  -W_{P}\left(  F_{2}\right)  \\
=&\,-P^{\left(  2\right)  }\left(  1\right)  \left(  \Lambda_{1}^{3}\left(
1\right)  -\Lambda_{2}^{3}\left(  1\right)  \right)  +\int_{0}^{1}P^{\left(
3\right)  }\left(  t\right)  \mathrm{d}\Lambda_{1}^{4}\left(  t\right)
-\int_{0}^{1}P^{\left(  3\right)  }\left(  t\right)  \mathrm{d}\Lambda_{2}%
^{4}\left(  t\right)  \\
=&\,-P^{\left(  2\right)  }\left(  1\right)  \left(  \Lambda_{1}^{3}\left(
1\right)  -\Lambda_{2}^{3}\left(  1\right)  \right)  +P^{\left(  3\right)
}\left(  t\right)  \left(  \Lambda_{1}^{4}\left(  t\right)  -\Lambda_{2}%
^{4}\left(  t\right)  \right)  |_{0}^{1}-\int_{0}^{1}P^{\left(  4\right)
}\left(  t\right)  \left(  \Lambda_{1}^{4}\left(  t\right)  -\Lambda_{2}%
^{4}\left(  t\right)  \right)  \mathrm{d}t\\
=&\,\sum_{j=3}^{m}\left(  -1\right)  ^{j-2}P^{\left(  j-1\right)  }\left(
1\right)  \left(  \Lambda_{1}^{j}\left(  1\right)  -\Lambda_{2}^{j}\left(
1\right)  \right)  +\int_{0}^{1}\left(  -1\right)  ^{m-1}P^{\left(  m\right)
}\left(  t\right)  \left(  \Lambda_{1}^{m}\left(  t\right)  -\Lambda_{2}%
^{m}\left(  t\right)  \right)  \mathrm{d}t.
\end{align*}
The remainder of the proof is similar to that of Proposition \ref{prop.AmULD}.
\end{proof}

\begin{proof}[Proof of Lemma \ref{lemma.cmu properties}]
The proof is similar to that of Lemma \ref{lemma.c properties}.    
\end{proof}

\begin{proof}[Proof of Proposition \ref{prop.UISDC properties}]
 The results directly follow from Proposition \ref{prop.AmUISD} and Lemma \ref{lemma.cmu properties}. 
\end{proof}

\begin{proof}[Proof of Proposition \ref{prop.AmDISD}]
The proof closely follows the strategies of the proofs of Theorem 1 of \citet{leshno2002preferred} and Theorem 2.4 of \citet{aaberge2021ranking}.
By definition, for $m\ge 4$, we have that
\[
\tilde{\Lambda}_{j}^{m}\left(  p\right)  =\int_{p}^{1}\tilde{\Lambda}%
_{j}^{m-1}\left(  t\right)  \mathrm{d}t=\int_{0}^{1}\tilde{\Lambda}_{j}%
^{m-1}\left(  t\right)  \mathrm{d}t-\int_{0}^{p}\tilde{\Lambda}_{j}%
^{m-1}\left(  t\right)  \mathrm{d}t,
\]
which implies
${\mathrm{d}\tilde{\Lambda}_{j}^{m}\left(  p\right)  }/{\mathrm{d}%
p}=-\tilde{\Lambda}_{j}^{m-1}\left(  p\right)$.
Also, we have that $$\tilde{\Lambda}_{j}^{3}\left(  p\right)  =\int_{p}^{1}{\Lambda}_{j}^{2}\left(  t\right)  \mathrm{d}t=\int_{0}^{1}{\Lambda}_{j}
^{2}\left(  t\right)  \mathrm{d}t-\int_{0}^{p}{\Lambda}_{j}
^{2}\left(  t\right)  \mathrm{d}t.$$
By a strategy similar to that of the proof of Proposition \ref{prop.AmUISD}, we can show that
\begin{align*}
&  W_{P}\left(  F_{1}\right)  -W_{P}\left(  F_{2}\right)  \\
=&\,-P^{\left(  2\right)  }\left(  0\right)  \left(  \tilde{\Lambda}_{1}%
^{3}\left(  0\right)  -\tilde{\Lambda}_{2}^{3}\left(  0\right)  \right)
+\int_{0}^{1}P^{\left(  3\right)  }\left(  t\right)  \mathrm{d}\tilde{\Lambda
}_{1}^{4}\left(  t\right)  -\int_{0}^{1}P^{\left(  3\right)  }\left(
t\right)  \mathrm{d}\tilde{\Lambda}_{2}^{4}\left(  t\right)  \\
=&\,\sum_{j=3}^{m}-P^{\left(  j-1\right)  }\left(  0\right)  \left(
\tilde{\Lambda}_{1}^{j}\left(  0\right)  -\tilde{\Lambda}_{2}^{j}\left(
0\right)  \right)  +\int_{0}^{1}-P^{\left(  m\right)  }\left(  t\right)
\left(  \tilde{\Lambda}_{1}^{m}\left(  t\right)  -\tilde{\Lambda}_{2}%
^{m}\left(  t\right)  \right)  \mathrm{d}t.
\end{align*}
The remainder of the proof is similar to that of Proposition \ref{prop.AmDLD}.
\end{proof}

\begin{proof}[Proof of Lemma \ref{lemma.cmd properties}]
 The proof is similar to that of Lemma \ref{lemma.c properties}.   
\end{proof}

\begin{proof}[Proof of Proposition \ref{prop.DISDC properties}]
 The results directly follow from Proposition \ref{prop.AmDISD} and Lemma \ref{lemma.cmd properties}. 
\end{proof}

\begin{proof}[Proof of Proposition \ref{prop.cm inverse asymptotic limit}]
The results follow from \eqref{eq.Lambda weak convergence}, continuous mapping theorem, Fubini's theorem, and Theorem 2.1 of \citet{fang2014inference}. 
\end{proof}

\begin{proof}[Proof of Proposition \ref{prop.confidence interval ISDC}]
The proof is similar to that of Proposition \ref{prop.confidence interval}.    
\end{proof}

\subsection{Proofs for Section \ref{sec.ASD}}
\begin{proof}[Proof of Lemma \ref{lemma.cm properties}]
The proof is similar to that of Lemma \ref{lemma.c properties}.
\end{proof}

\begin{proof}[Proof of Proposition \ref{prop.SDC properties}]
The results directly follow from Lemma \ref{lemma.cm properties} and Theorem 4 of \citet{tsetlin2015generalized}.    
\end{proof}

\begin{proof}[Proof of Lemma \ref{lemma.F asymptotic limit}]
The proof for independent samples is trivial and is omitted. For matched pairs, let $\mathcal{H}_{1}=\{1_{[a,x]\times[a,b]}:x\in[a,b]\}$,
$\mathcal{H}_{2}=\{1_{[a,b]\times[a,x]}:x\in[a,b]\}$, and
$\mathcal{H}=\mathcal{H}_{1}\cup\mathcal{H}_{2}$. Define%
\[
\hat{G}_{n}\left(  f\right)  =\sqrt{n}\left\{  \frac{1}{n}\sum_{i=1}%
^{n}f\left(  X_{i},Y_{i}\right)  -E\left[  f\left(  X_{i},Y_{i}\right)
\right]  \right\}
\]
for every $f\in\mathcal{H}$. Since $\mathcal{H}_{1}$ and $\mathcal{H}_{2}$ are
both Donsker classes, by Example 2.10.7 of \citet{van1996weak}, $\hat{G}_{n}\leadsto\mathbb{G}$ in $\ell^{\infty
}\left(  \mathcal{H}\right)  $ for some zero-mean Gaussian process $\mathbb{G}$. Define a map $T$ such that
for every $\psi\in\ell^{\infty}\left(  \mathcal{H}\right)  $ and every
$x\in[a,b]$,
\[
T\left(  \psi\right)  \left(  x\right)  =\binom{\psi\left(  1_{[a
,x]\times[a,b]}\right)  }{\psi\left(  1_{[a,b]\times[a
,x]}\right)  }.
\]
Then by the continuous mapping theorem,
\[
T\left(  \hat{G}_{n}\right)  =\binom{\sqrt{n}\{  \hat{F}_{1}%
-F_{1}\}  }{\sqrt{n}\{  \hat{F}_{2}-F_{2}\}  }\leadsto
T\left(  \mathbb{G}\right)  ,
\]
where for every $x\in[a,b]$,
\[
T\left(  \mathbb{G}\right)  \left(  x\right)  =\binom{\mathbb{G}\left(
1_{[a,x]\times[a,b]}\right)  }{\mathbb{G}\left(  1_{[a,b]\times[a,x]}\right)  }.
\]
By the continuous mapping theorem again, we obtain that
\[
\sqrt{T_{n}}\left\{  (\hat{F}_{1}-\hat{F}_{2})-(F_{1}-F_{2})\right\}
\leadsto\mathbb{G}_{F},
\]
where for every $x\in[a,b]$,
\[
\mathbb{G}_{F}\left(  x\right)  =\sqrt{1-\lambda}\mathbb{G}\left(
1_{[a,x]\times[a,b]}\right)  -\sqrt{\lambda}\mathbb{G}\left(
1_{[a,b]\times[a,x]}\right)
\]
with $\lambda=1/2$.
Then it is easy to show that for all $x,x^{\prime}\in[a,b]$,
\begin{align*}
&  E\left[  \mathbb{G}_{F}\left(  x\right)  \mathbb{G}_{F}(  x^{\prime
})  \right]  \\
=&\,\left(  1-\lambda\right)  E\left[  \mathbb{G}\left(  1_{[a
,x]\times[a,b]}\right)  \mathbb{G}\left(  1_{[a,x^{\prime}%
]\times[a,b]}\right)  \right]  -\sqrt{\lambda\left(  1-\lambda\right)
}E\left[  \mathbb{G}\left(  1_{[a,x]\times[a,b]}\right)
\mathbb{G}\left(  1_{[a,b]\times[a,x^{\prime}]}\right)  \right]  \\
&  -\sqrt{\lambda\left(  1-\lambda\right)  }E\left[  \mathbb{G}\left(
1_{[a,b]\times[a,x]}\right)  \mathbb{G}\left(  1_{[a
,x^{\prime}]\times[a,b]}\right)  \right]  +\lambda E\left[
\mathbb{G}\left(  1_{[a,b]\times[a,x]}\right)  \mathbb{G}\left(
1_{[a,b]\times[a,x^{\prime}]}\right)  \right]  \\
=&\,\left(  1-\lambda\right)  F_{1}(  x\wedge x^{\prime})  -\left(
1-\lambda\right)  F_{1}\left(  x\right)  F_{1}(  x^{\prime})
-\sqrt{\lambda\left(  1-\lambda\right)  }F_{12}(  x,x^{\prime})
  \\
& +\sqrt{\lambda\left(  1-\lambda\right)  }F_{1}\left(  x\right)  F_{2}(
x^{\prime}) -\sqrt{\lambda\left(  1-\lambda\right)  }F_{12}(  x^{\prime},x)
+\sqrt{\lambda\left(  1-\lambda\right)  }F_{1}(  x^{\prime})
F_{2}\left(  x\right)\\
  &+\lambda F_{2}(  x\wedge x^{\prime})
-\lambda F_{2}\left(  x\right)  F_{2}(  x^{\prime}).
\end{align*}
\end{proof}

\begin{proof}[Proof of Proposition \ref{prop.cm asymptotic limit}]
As shown in Lemma \ref{lemma.F asymptotic limit},
\[
\sqrt{T_{n}}\left\{  (  \hat{F}_{1}-\hat{F}_{2})  -\left(
F_{1}-F_{2}\right)  \right\}  \leadsto\mathbb{G}_F.
\]
Then by the linearity and continuity of $\mathcal{I}_m$ defined in \eqref{eq.Im} and the continuous mapping theorem, it follows that
\begin{align*}
\sqrt{T_{n}}\left(  \hat{\phi}_{m}-\phi_{m}\right)    & =\sqrt{T_{n}}\left\{
\mathcal{I}_{m}(  \hat{F}_{1}-\hat{F}_{2})  -\mathcal{I}_{m}\left(
F_{1}-F_{2}\right)  \right\}  \\
& =\sqrt{T_{n}}\mathcal{I}_{m}\left(  (  \hat{F}_{1}-\hat{F}_{2})
-\left(  F_{1}-F_{2}\right)  \right)  \leadsto\mathcal{I}_{m}\left(
\mathbb{G}_F\right) \equiv\mathbb{G}_{m}.
\end{align*}
For every $x\in[a,b]$,
\[
Var\left(    \mathbb{G}_{m}  \left(  x\right)
\right)  =E\left[  (\mathcal{I}_{m}\left(  \mathbb{G}_{F}
\right)\left(  x\right))  ^{2}\right]  -E\left[  \mathcal{I}_{m}\left(  \mathbb{G}_{F}  \right)\left(
x\right)  \right]  ^{2},
\]
where, by Fubini's theorem (see, for example, Theorem 2.37 of \citet{folland2013real}),
\begin{align*}
&E\left[ ( \mathcal{I}_{m}\left(  \mathbb{G}_{F}  \right)\left(  x\right))^{2}\right]\\    =&\,\int\left(  \int_{a}^{x}\cdots\int_{a}^{t_{3}}\int
_{a}^{t_{2}}\mathbb{G}_{F}\left(  \omega\right)  \left(  t_{1}\right)
\mathrm{d}t_{1}\mathrm{d}t_{2}\cdots\mathrm{d}t_{m-1}\right)  ^{2}%
\mathrm{d}\mathbb{P}\left(  \omega\right)  \\
=&\,\int_{a}^{x}\cdots\int_{a}^{t_{3}^{\prime}}\int_{a}^{t_{2}^{\prime}%
}\left(  \int_{a}^{x}\cdots\int_{a}^{t_{3}}\int_{a}^{t_{2}}E\left[
\mathbb{G}_{F}\left(  t_{1}\right)  \mathbb{G}_{F}\left(  t_{1}^{\prime
}\right)  \right]  \mathrm{d}t_{1}\mathrm{d}t_{2}\cdots\mathrm{d}%
t_{m-1}\right)  \mathrm{d}t_{1}^{\prime}\mathrm{d}t_{2}^{\prime}%
\cdots\mathrm{d}t_{m-1}^{\prime}%
\end{align*}
and
\[
E\left[  \mathcal{I}_{m}\left(  \mathbb{G}_{F}\right)\left(  x\right)  
\right]  =0.
\]

By a proof similar to that of Proposition \ref{prop.c asymptotic limit}, we can show the weak convergence
\[
\sqrt{T_{n}}\left\{  \hat{c}_{m}\left(  F_{1},F_{2}\right)  -c_{m}\left(
F_{1},F_{2}\right)  \right\}  \leadsto\mathcal{F}_{\phi_{m}}^{\prime}\left(
\mathbb{G}_{m}\right).
\]
\end{proof}

\begin{proof}[Proof of Proposition \ref{prop.confidence interval SDC}]
The proof is similar to that of Proposition \ref{prop.confidence interval}.
\end{proof}
\putbib
\end{bibunit}
\end{document}